\documentclass[a4paper,UKenglish,leqno]{article}
\usepackage[utf8]{inputenc}
\usepackage{authblk}

\usepackage[a4paper, top=25.4mm, bottom=25.4mm, left=25.4mm, right=25.4mm, includefoot]{geometry}
\usepackage{amsmath,amssymb,mathtools}
\usepackage{enumitem}
\usepackage{tabularx}

\usepackage{xcolor}
\definecolor{dark-blue}{rgb}{0.05,0.25,0.85}
\usepackage{caption}

\usepackage{amsthm}

\usepackage{hyperref}
\hypersetup{colorlinks=true,linkcolor=dark-blue,citecolor=dark-blue}
\usepackage[capitalise,nameinlink]{cleveref}

\newtheorem{theorem}{Theorem}[section]
\crefname{theorem}{Theorem}{Theorems}
\newtheorem{lemma}[theorem]{Lemma}
\crefname{lemma}{Lemma}{Lemmas}
\newtheorem{corollary}[theorem]{Corollary}
\crefname{corollary}{Corollary}{Corollaries}
\newtheorem{observation}[theorem]{Observation}
\crefname{observation}{Observation}{Observations}
\newtheorem{conjecture}[theorem]{Conjecture}
\crefname{conjecture}{Conjecture}{Conjectures}

\theoremstyle{definition}
\newtheorem{definition}[theorem]{Definition}
\crefname{definition}{Definition}{Definitions}
\newtheorem{example}[theorem]{Example}
\crefname{example}{Example}{Examples}

\theoremstyle{remark}
\newtheorem*{remark}{Remark}
\newtheorem{claim}{Claim}[theorem]
\crefname{claim}{Claim}{Claims}

\newenvironment{claimproof}{\begin{proof}[Proof.]}{\end{proof}}

\crefformat{enumi}{#2#1#3}
\Crefformat{enumi}{#2#1#3}

\newcommand{\crefthm}[2]{%
  \hyperref[#2]{Part \ref*{#1:#2} of \cref{#1}}%
}

\newcommand{\crefdef}[2]{%
  \hyperref[#2]{\cref{#1}~\ref{#1:#2}}%
}

\usepackage{imakeidx}

\makeindex
\makeindex[name=Symbol,title=Symbols]

\newcommand{\Symbol}[1]{\index[Symbol]{#1}}

\usepackage{tikz}
\usetikzlibrary{arrows}
\usetikzlibrary{shapes.geometric}
\usetikzlibrary{positioning,arrows.meta}
\usetikzlibrary{decorations.pathmorphing}

\newcommand{\N}{\mathbb{N}}

\newcommand{\AAA}{\mathcal{A}}
\newcommand{\BBB}{\mathcal{B}}
\newcommand{\CCC}{\mathcal{C}}
\newcommand{\DDD}{\mathcal{D}}
\newcommand{\FFF}{\mathcal{F}}
\newcommand{\GGG}{\mathcal{G}}
\newcommand{\HHH}{\mathcal{H}}
\newcommand{\KKK}{\mathcal{K}}
\newcommand{\LLL}{\mathcal{L}}
\newcommand{\MMM}{\mathcal{M}}
\newcommand{\PPP}{\mathcal{P}}
\newcommand{\SSS}{\mathcal{S}}
\newcommand{\TTT}{\mathcal{T}}

\newcommand{\problemdef}[3]{
  \begin{center}
    \begin{minipage}{0.95\textwidth}
      \noindent
         \vspace{-0.5em}\colorbox{gray!8!white}{\textsc{#1}}

      \vspace{4pt}
      \setlength{\tabcolsep}{3pt}
      \begin{tabularx}{\textwidth}{@{}lX@{}}
        \textbf{Input:} 		& #2 \\
        \textbf{Question:} 	& #3
      \end{tabularx}
    \end{minipage}
  \end{center}
 }

\newcommand*{\Abs}[1]{| #1 |}
\newcommand\restr[2]{{%
  \left.\kern-\nulldelimiterspace
  #1
  \vphantom{|}
  \right|_{\text{\scriptsize $#2$}}
  }}

\newcommand{\sth}{\mathrel : }

\newcommand{\head}{\operatorname{head}}
\newcommand{\tail}{\operatorname{tail}}
\newcommand{\leaves}[1]{\mathrm{leaves}(#1)}
\newcommand{\loops}{\mathrm{loops}}
\newcommand{\spl}{\operatorname{split}}

\newcommand{\insc}{\operatorname{ins}}
\newcommand{\outc}{\operatorname{outs}}
\newcommand{\LL}{\mathfrak{L}}
\newcommand{\MM}{\mathfrak{M}}
\newcommand{\m}{\mathfrak{m}}
\newcommand{\n}{\mathfrak{n}}
\newcommand{\dom}{\operatorname{dom}}
\newcommand{\stitch}{\operatorname{stitch}}
\newcommand{\rip}{\operatorname{rip}}
\newcommand{\knit}{\operatorname{knit}}
\newcommand{\torso}{\operatorname{torso}}

\newcommand{\cw}[1]{\mathrm{cw}(#1)}
\newcommand{\tw}[1]{\mathrm{tw}(#1)}
\newcommand{\w}{\operatorname{width}}
\newcommand{\rt}{\operatorname{root}}

\newcommand{\lenks}[1]{\mathrm{left}(#1)}
\newcommand{\riets}[1]{\mathrm{right}(#1)}

\author[1]{Dario Cavallaro}
\author[2,3]{Ken-ichi Kawarabayashi}
\author[1]{Stephan Kreutzer}

\affil[1]{\small Technische Universität Berlin}
\affil[2]{\small National Institute of Informatics, Japan}
\affil[3]{ \small The University of Tokyo, Japan}

\begin{document}
\title{Well-Quasi-Ordering Eulerian Digraphs: Bounded~Carving~Width}
\date{}
\maketitle
\begin{abstract}
    We prove that every class of Eulerian directed graphs of bounded carving width (equivalently of bounded degree and treewidth) is well-quasi-ordered by strong immersion. In fact, we prove a stronger result, namely that every class of Eulerian directed graphs of bounded carving width, where every vertex is additionally labeled from a well-quasi-order, fixes a linear order on its incident edges, and may impose further restrictions on how the immersion is allowed to route paths through it, is well-quasi-ordered by an adequate notion of strong immersion. To this extent, we develop a framework seemingly suited to prove well-quasi-ordering for classes of Eulerian directed graphs by (strong) immersion and present a first meta theorem in that direction.
    We complement our results by observing that the class of Eulerian directed graphs of unbounded degree is \emph{not} well-quasi-ordered by \emph{strong} immersion, even if we assume the treewidth of the class to be at most two. We conclude with a dichotomy result, proving for a very restricted class of Eulerian directed graphs of unbounded degree that it is not well-quasi-ordered by strong immersion, but it is well-quasi-ordered by weak immersion.
\end{abstract}
\section{Introduction}
Well-quasi-orderings of mathematical objects have been studied for many decades in various mathematical as well as computer science research areas \cite{kruskal72,wqo_lang,Higman1952,wqo_survey}. In general, given a set of objects~$V$ and a \emph{binary relation}~$\preceq$ on $V\times V$ that is reflexive and transitive, we call $(V,\preceq)$ a \emph{quasi-order}. We call an infinite sequence $(e_i)_{i \in \N}$ a \emph{chain} if for all $i \leq j$ it holds $e_i \preceq e_j$. We call it an \emph{antichain} if for all $i,j \in \N$ the elements $e_i,e_j$ are incomparable with respect to $\preceq$. Finally,~$V$ is called \emph{well-quasi-ordered by $\preceq$} if for every infinite sequence~$e_1,e_2,\ldots$ of objects~$e_i \in V$, there exist $j_1,j_2 \geq 1$ with~$j_1 < j_2 $ such that~$e_{j_1} \preceq e_{j_2}$. In this case we refer to $(V,\preceq)$ as a \emph{well-quasi-order}. 

    In 1960, Kruskal \cite{Kru60} proved one of the first impactful graph-theoretic results regarding well-quasi-ordering, namely that trees are well-quasi-ordered by topological containment\footnote{The exact statement is stronger and more complex, and so is its proof.}. Later, Nash-Williams \cite{nash63} provided a simplified proof that finite trees are well-quasi-ordered by topological containment. Since then a whole field of research emerged, trying to identify classes of graphs (or related data structures) that are well-quasi-ordered with respect to some well-known ``containment relation''. In the late $20$th century, Robertson and Seymour finally gave a proof of Wagner's Conjecture \cite{GMXX}, stating that undirected graphs are well-quasi-ordered by the \emph{minor} relation---$H$ is a \emph{minor} of $G$ if it can be obtained from a subgraph of $G$ by contracting edges---marking a major breakthrough in the area. The proof of this result culminated in the creation of a whole field of research: Graph minor theory, laying the groundwork for more general structural graph theory. A core idea of the proof is to decompose graphs into ``structurally well-behaved'' pieces forming a ``tree-like'' structure. Once given such a tree-like structure called \emph{tree decomposition}, Robertson and Seymour showed that one may lift the simplified proof due to Nash-Williams for well-quasi-ordering trees, which in a nutshell shows that it suffices to prove that the ``structurally well-behaved'' pieces are well-quasi-ordered by the minor relation. Specifically, in \cite{GMIV} Robertson and Seymour prove a generalisation of Nash-Williams' result for trees with bounded edge-weights satisfying a certain linkedness property, and whose vertices are labelled from a well-quasi-order, essentially representing the ``structurally well-behaved'' pieces. Intuitively speaking, they leverage Nash-Williams' proof to develop a framework that can be used to prove well-quasi-ordering for data structures encoded in trees in a fairly general way. In particular, graphs of bounded treewidth can be encoded using their framework, and they use it to prove that graphs of bounded treewidth are well-quasi-ordered by the minor relation. Similarly, as alluded to above, the graph minor structure theorem\footnote{Technically speaking, a different variant of encoding in trees is used, i.e., the ``tree of tangles''.} allows to decompose graphs via tree decompositions into pieces which are either of bounded treewidth or ``structurally well-behaved". This can again be thought of as a labelled tree with bounded edge-weights (the adhesion) and labelled vertices. A lot of effort is then put into proving that their framework may again be applied, culminating in a proof of Wagner's Conjecture. This highlights that the fact that their framework works for graphs of bounded treewidth  with vertices labelled from a well-quasi-order (and additional constraints) marks not only a base case, but a fundamental tool for the general proof of Wagner's Conjecture. 
    
\smallskip

Besides the graph minor relation, another prominent relation between graphs is the \emph{immersion} relation. 
We say that a graph~$H$ \emph{immerses} in a graph~$G$ if there exists a map~$\gamma : H \to \{P \mid P \subseteq G\}$ such that~$\restr{\gamma}{V(H)}\colon V(H) \to V(G)$ is injective and~$\restr{\gamma}{E(H)} \colon E(H) \to \{P \mid P \text{ is a trail in } G\}$ guarantees that~$P \coloneqq \gamma(\{u,v\})$ is a trail starting in~$\gamma(u)$ and ending in~$\gamma(v)$, and such that for two distinct edges $e_1,e_2 \in E(H)$, the trails $\gamma(e_1),\gamma(e_2)$ are edge-disjoint. We call $\gamma$ an \emph{immersion}\footnote{Note that immersion is often defined via mapping edges to paths; both notions are equivalent on (directed) graphs.}. 

The immersion $\gamma$ is \emph{strong} if no trail $\gamma(e)$ contains a vertex of $\gamma(V(H))$ as an internal vertex. We write $\gamma\colon H \hookrightarrow G$ to mean that $\gamma$ strongly immerses $H$ in $G$, and $\gamma: H\hookrightarrow^* G$ to mean that $\gamma$ immerses $H$ in $G$; we may call the latter \emph{weak} immersion, to emphasise that it may not be strong. Immersions (weak and strong) are defined analogously for directed graphs, mapping directed edges to directed trails.  

Later on, still part of their seminal work on graph minors, Robertson and Seymour derived Nash-William's Conjecture from their work. That is, they  proved that undirected graphs are well-quasi-ordered by \emph{weak} immersion \cite{GMXXIII}, marking another breakthrough in the area. Their proof does not use \emph{strong} immersions\footnote{Note that their results imply a proof for strong immersions on graph classes of bounded degree.} and, as Robertson and Seymour discussed in \cite{GMXXIII}, their framework cannot easily be adapted to prove well-quasi-ordering of undirected graphs by \emph{strong} immersion. In fact, they claim that a proof of the strong immersion result is way more complicated: ``\textit{It seemed to us at one time that we had a proof of the stronger [immersion conjecture], but even if it was correct it was very much more complicated, and it is unlikely that we will write it down}'' \cite[p2)]{GMXXIII}. As far as we are aware, the strong immersion conjecture is still open. We want to emphasise that their proof for well-quasi-ordering graphs of bounded treewidth, does imply well-quasi-ordering for bounded treewidth graphs by \emph{weak} immersion, but it does not naturally lift to a proof for \emph{strong} immersion. In fact, we are not aware of a result in the literature proving that graphs of bounded treewidth are well-quasi-ordered by strong immersion, highlighting once more the importance of results and new frameworks in the area.

\paragraph{Directed Graphs.} After proving Wagner's and Nash-Williams' Conjecture, Robertson and Seymour turned their interest to directed graphs (also referred to as \emph{digraphs}). For, it is natural to ask ``What are good relations suited for well-quasi-ordering directed graphs?''. Together with Johnson and Thomas, they introduced a notion of \emph{directed treewidth} \cite{dtw_def} seemingly suited to their needs---but unfortunately the natural extension of minors to digraphs, which is \emph{butterfly minors}, does \emph{not} yield a well-quasi-ordering for digraphs, not even on digraphs of low directed treewidth (we will elaborate on this below). 

It quickly turned out that directed graphs behave very differently from undirected graphs: For most well-studied containment relations on digraphs, counterexamples to well-quasi-ordering have been found. 
Liu and Muzi \cite{alternatingpaths} proved that digraphs are not well-quasi-ordered by immersion nor by strong immersion. Their counterexample is a growing sequence of slightly modified alternating paths (see the thick highlighted part of \cref{fig:antichain} for an illustration), which is a counterexample to well-quasi-ordering with respect to many more minor relations on digraphs including the above mentioned butterfly minor relation. In the same paper they prove that alternating paths are, in a sense, the relevant obstruction to well-quasi-ordering with respect to strong immersion, for they prove that the class of digraphs \emph{not} containing long alternating paths is indeed well-quasi-ordered by the strong immersion relation. It turns out that the same observation can be made for butterfly-minors \cite{muzi2017paths}, i.e., classes of digraphs that do not admit ``long'' alternating paths are well-quasi-ordered by butterfly-minors. Unfortunately, not containing long alternating paths is a very strict restriction for digraphs, in particular such graphs must have low directed and even low undirected treewidth \cite{alternatingpaths,KreutzerT2012}. 

Besides this, only a few positive results on well-quasi-ordering directed graphs seem to be known. One of the most prominent ones, due to Chudnovsky and Seymour \cite{tournaments}, states that \emph{tournaments}---graphs obtained from orienting edges of undirected cliques---are well-quasi-ordered by strong immersion; note that tournaments may contain arbitrarily long alternating paths. This result was later extended to semi-complete digraphs by Barbero, Paul, and Pilipczuk \cite{BarberoPP2019}.

\paragraph{Eulerian Digraphs}
A class of directed graphs that lies somewhat between general directed and undirected graphs is the class of \emph{Eulerian digraphs}. Given a digraph $G$ we call a vertex $v \in V(G)$ \emph{Eulerian}, if its in-degree equals its out-degree. Then, a digraph $G$ is called \emph{Eulerian} if every vertex in $V(G)$ is Eulerian\footnote{It is common in the literature to assume Eulerian digraphs to be weakly connected. We do not impose this here, but mention that throughout most of the paper we will assume our digraphs to be weakly connected, as the results easily transfer to connected components.}. Equivalently, $G$ is the union of a set of pairwise edge-disjoint directed cycles. Eulerian digraphs have many nice properties, making them a particularly interesting class of digraphs to study; we will discuss some of those in \cref{sec:preliminaries}. See \cite[Chapter 4]{Bang-JensenG2018} for an introduction to Eulerian digraphs, as well as \cite{Frank_2path,frank95,EDP_Euler,medina2019well,Johnson2002} for further structural results regarding Eulerian digraphs. 

A noteworthy structural result tying Eulerian digraphs to undirected digraphs is the following: It turns out that when restricting to a class $\CCC$ of Eulerian digraphs of \emph{bounded degree}, then the undirected and directed treewidth are qualitatively equivalent \cite{dtw_def}. This is in stark contrast to general digraphs; consider an acyclic orientation of a large clique. This is another indication that, at least structurally speaking, Eulerian digraphs seem to be closely related to undirected digraphs. 

It turns out that restricting to classes of bounded degree is crucial when working with strong immersion on Eulerian digraphs. On the one hand note that when working with the general class of Eulerian digraphs, one must loosen the assumptions in the edge-disjoint paths problem to allow for the paths to pass through terminal vertices to get an $FPT$-algorithm\footnote{This is a common assumption, even in the undirected setting.} as discussed in \cite{EDP_Euler}. This artefact disappears when working with classes of bounded degree. On the other hand, our first result is negative in the sense that it yields a counterexample to well-quasi-ordering Eulerian digraphs by strong immersion. The following construction relies on the aforementioned construction due to Liu and Muzi \cite{alternatingpaths}, leveraging that alternating paths are not well-quasi-ordered by strong immersion. 
\begin{figure}
        \centering
        \begin{tikzpicture}[>=Stealth]
            \tikzstyle{vertex}=[circle, fill=black, inner sep=2pt]
            \node[vertex, fill=gray, label=above:{\small $b^k$}] (b) at (0,2) {};
            \node[vertex,label=left:{\small $v^k_0$}] (l) at (-3,0) {};
            \node[vertex,label=right:{\small $v^k_{2k}$}] (r) at (3,0) {};
            \node[vertex] (v1) at (-2,0) {};
            \node[vertex] (v2) at (-0.5,0) {};
            \node[vertex] (v3) at (0.5,0) {};
            \node[vertex] (v4) at (2,0) {};
            \node[vertex,label=left:{\small $l^k_1$}] (l1) at (-3.5,0.5) {};
            \node[vertex,label=left:{\small $l^k_2$}] (l2) at (-3.5,-0.5) {};
            \node[vertex,label=right:{\small $r^k_1$}] (r1) at (3.5,0.5) {};
            \node[vertex,label=right:{\small $r^k_2$}] (r2) at (3.5,-0.5) {};
            \foreach \s/\t in {l2/l,l1/l,r1/r,r2/r,l/v1,v2/v1,r/v4,v3/v4}
            \draw[very thick,->] (\s) to (\t) ;
            \draw[very thick,dotted] (v2) to (v3);
         \foreach \s/\t in {b/v2,b/v3,v1/b,v4/b}
         {
            \draw[thick,->,gray,bend right=8] (\s) to (\t) ;
            \draw[thick,->,gray,bend left=8] (\s) to (\t); }
         \draw[thick,->,gray,bend right=8] (r) to (b) ;
         \draw[thick,->,gray,bend left=8] (l) to (b) ;
         \draw[thick,->,gray,bend left=8] (b) to (r1) ;
         \draw[thick,->,gray,bend right=8] (b) to (l1) ;
         \draw[thick,->,gray,bend left=85] (b) to (r2) ;
         \draw[thick,->,gray,bend right=85] (b) to (l2) ;
        \end{tikzpicture}
        \caption{An infinite antichain of Eulerian digraphs $G_k$ with respect to strong immersion. The thick highlighted subgraph forms an infinite antichain for digraphs with respect to strong immersions.}
        \label{fig:antichain}
    \end{figure}
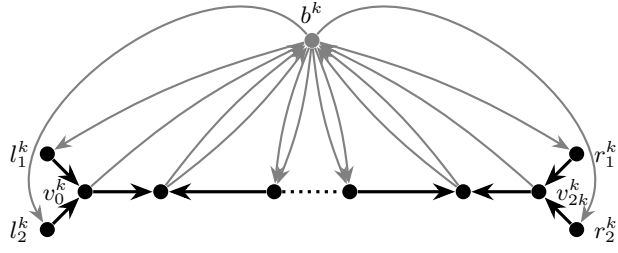
\begin{observation}\label{thm:antichain}
    The class of Eulerian digraphs is not well-quasi-ordered by strong immersion.
\end{observation}
\begin{proof}
    For every $k\geq 1$ let $G_k$ be defined as follows; see the thick highlighted subgraph in \cref{fig:antichain} for an illustration. We set $V(G_k) \coloneqq \{ l^k_1, l^k_2, v^k_0, \dots, v^k_{2k}, r^k_1, r^k_2 \}$ and add edges from $l^k_1$ and $l^k_2$ to $v^k_0$, from $r^k_1$ and $r^k_2$ to $v^k_{2k}$, from $v^k_0$ to $v^k_1$, and from $v^k_{2k}$ to $v^k_{2k-1}$. Finally, for all $1 \leq i < k$ we add edges $(v^k_{2i}, v^k_{2i-1})$ and $(v^k_{2i}, v^k_{2i+1})$. 

    Thus, for every $k \geq 1$, the graph $G_k$ consists of two special vertices, $v^k_0$ and $v^k_{2k}$, which are the only vertices of degree $3$ and between them a long ``path'' with edges in alternating directions. 

    We claim that if $i \neq j$ then $G_i \not\hookrightarrow G_j$. This is obvious if $i > j$, thus we assume $i < j$. But then, the two vertices $v^i_0$ and $v^i_{2i}$ of degree $3$ in $G_i$ must be mapped to the two vertices $v^j_0$ and $v^j_{2j}$ of degree $3$ in $G_j$, as no other vertex in $G_j$ has degree $3$. This implies that the vertices $v^i_{1}, \dots, v^i_{2i-1}$ on the alternating path in $G_i$ must be mapped to vertices on the alternating path in $G_j$. As $j>i$ there must be two adjacent vertices $v^i_{t}$ and $v^i_{t+1}$ in $G_i$ that are mapped to two vertices $v^j_{t_1}$ and $v^j_{t_2}$ which are not adjacent in $G_j$. But then there is no directed path in $G_j$ between $v^j_{t_1}$ and $v^j_{t_2}$ and thus the edge between $v^i_{t}$ and $v^i_{t+1}$ cannot be immersed into $G_j$.

    So far the example only shows that the class of directed graphs is not well-quasi-ordered under the strong immersion relation. But the digraphs are not yet Eulerian. We can easily make them Eulerian using a standard construction from the theory of Eulerian digraphs: For $i>1$ let $G_i'$ be the digraph obtained from $G_i$ by adding a fresh vertex $b^i$ and for all  $v \in V(G_i)$ that have more incoming than outgoing edges we add sufficiently many edges from $v$ to $b^i$ so that the in-degree of $v$ equals its out-degree in $G_i'$ and likewise we add edges from $b^i$ to $v \in V(G)$ if $v$ has more outgoing than incoming edges. See \cref{fig:antichain} for an illustration of the construction. (To avoid parallel edges one may subdivide parallel edges if need be.)

    It is easily seen that the resulting graphs $G'_i$ are Eulerian. But it is still the case that for $i < j$ there is no strong immersion of $G_i'$ in $G'_j$. For, any such immersion must map the high degree vertex $b^i$ of $G'_i$ to the corresponding vertex $b^j$ in $G'_j$ and so the remaining vertices must be mapped to each other. Thus the same argument as above shows that this is impossible.
\end{proof}

Note that, complementing the above discussion, the class $\CCC := \{ G_k' \sth k > 1 \}$ has unbounded maximum degree and up to the vertices $b^k$ for $k \geq 1$, every vertex is of degree at most four. Observe further that the undirected treewidth of the graphs $G'_k$ in the proof of \cref{thm:antichain} is at most two for every $k\geq 1$ and the graphs admit a planar embedding as depicted in \cref{fig:antichain}. This yields the following.

\begin{corollary} \label{cor:antichain_tw}
    The class of planar Eulerian digraphs of treewidth at most $2$ is not well-quasi-ordered by strong immersion.
\end{corollary}

Finally, we emphasise that the class $\CCC$ above is only an antichain for the \emph{strong} immersion relation, but not for the weak immersion relation; we will get back to this observation below.

\paragraph{Bounded Degree.}\cref{thm:antichain} shows that to obtain a well-quasi-ordering of Eulerian digraphs by strong immersion we must restrict to classes of bounded maximum degree. In contrast to and complementing \cref{thm:antichain} we propose the following, which was already conjectured by Johnson \cite{Johnson2002}.
\begin{conjecture}\label{conj:wqo_dreg}
  For every $d \geq 1$, the class of Eulerian digraphs of maximum degree $2d$ is well-quasi-ordered by strong immersion.
\end{conjecture}

It should be noted that in his dissertation, Johnson \cite{Johnson2002} already discovered the strong structural similarities between Eulerian digraphs and undirected graphs. In fact, he proved a structure theorem for internally $6$-connected $4$-regular Eulerian digraphs, and mentions that he anticipated to use it towards proving a well-quasi-order result for Eulerian digraphs by immersion\footnote{His definition slightly differs from the standard definition.} in the future. Unfortunately, to the best of our knowledge, he never published such results. 

Note further that, in contrast to the discussion about the results in \cite{alternatingpaths,muzi2017paths}, if \cref{conj:wqo_dreg} were true, then the obstructions to well-quasi-ordering Eulerian digraphs by strong immersion are not long alternating paths, but rather vertices of arbitrarily large degree. 

Let us discuss a few immediate consequences and important applications that would follow a proof of \cref{conj:wqo_dreg}. 

\paragraph{Circle Graphs and Pivot-Minors.} The most prominent one may be that a proof for $d=2$ would imply that \emph{circle graphs} are well-quasi-ordered by the \emph{pivot-minor} relation. This is a longstanding open problem, marking a crucial step towards the more general conjecture that undirected graphs are well-quasi-ordered by the pivot-minor relation (to the best of our knowledge, this conjecture is still open even in the weaker case of vertex-minors). We leave out the details and refer the interested reader to \cite{circlegraph_survey,wqo_rankwidth,rose_vertex_minors}. It turns out that every circle graph $G$ naturally corresponds to a $4$-regular Eulerian digraph $D(G)$ with a choice of Eulerian cycle, and every $4$-regular Eulerian digraph naturally corresponds to a class of circle graphs (depending on a choice of Eulerian cycle) such that the following holds. Given two circle graphs $H$ and $G$, $H$ is a pivot-minor of $G$ if and only if $D(H) \hookrightarrow D(G)$, i.e., strong immersion on the respective Eulerian digraphs corresponds to the pivot-minor relation on the circle graphs. It may be worth exploring whether the cases $d > 2$ may be of further help towards proving the general pivot-minor well-quasi-order conjecture, for example, by examining to which classes of undirected graphs they correspond to when trying to translate strong immersions of Eulerian digraphs to pivot-minors of undirected graphs.

\paragraph{Membership Testing.} Decision problems (in and outside of graph theory) can naturally be rephrased as \emph{Membership Testing problems}. Let $\SSS$ be a class of structures and $\PPP \subseteq \SSS$, then we define the following:
\problemdef{$(\PPP,\SSS)$-Membership Testing}{A structure $S \in \SSS$.}{Is $S$ contained in $\PPP$?}

A common example is \textsc{Planarity Testing}, where $\SSS$ is the class of undirected graphs and $\PPP$ is the class of planar graphs, i.e., given a graph $G$ one has to decide whether $G$ is planar. By a well-known result due to Kuratowski and Wagner, this reduces to the question of whether $G$ admits $K_{3,3}$ or $K_5$ as a minor. More generally, one may ask: ``Given a list of graphs $\FFF$, does $G$ admit a graph in $\FFF$ as a minor?''. As we discuss next, this can be rephrased as a Membership Testing problem of particular importance.

Robertson and Seymour \cite{GMXIII} proved that for any fixed graph $H$, the question of whether $G$ admits $H$ as a minor can be decided in time $O(f(H)n^3)$ for some function $f$ (and $O(f(H)n^2)$ in \cite{KAWARABAYASHI2012424}), and this has recently been brought down to almost linear time by Korhonen, Pilipczuk, and Stamoulis \cite{korhonen2024minor}. This result, combined with the fact that undirected graphs are well-quasi-ordered by the minor relation, implies the existence of a plethora of polynomial time algorithms for testing relevant graph properties by expressing them as Membership Testing problems. To make this more precise, let $\preceq$ induce a quasi-order on a class of structures $\SSS$; for example $\SSS$ could be the class of undirected graphs and $\preceq$ the minor relation. A class~$\PPP \subseteq \SSS$ of structures is called \emph{$\preceq$-closed}---for example minor-closed---if  $S \in \PPP$ implies that $S' \in \PPP$ for every $S' \preceq S$. Let now $ \neg\PPP \subseteq \SSS$ be chosen so that $\neg \PPP \cap \PPP = \emptyset$ and such that $\neg \PPP$ is \emph{$\preceq$-minimal} in the sense that for every $ S \in \neg \PPP$ and $S' \prec S$ it holds that $S' \in \PPP$. If the class $\SSS$ is well-quasi-ordered by $\preceq$, then $\neg \PPP$ is guaranteed to be finite. Thus, since $\PPP$ is $\preceq$-closed, to decide whether $S \in \PPP$ for some $S \in \SSS$, it suffices to check whether for all $S' \in \neg \PPP$ we have $S' \not\preceq S$. In particular, if $S' \preceq S$ can be efficiently decided, say in polynomial time for fixed $S'$, then so can the \textsc{$(\PPP,\SSS)$-Membership Testing} problem by iterating over the finite list $\neg \PPP$. 

Robertson and Seymour \cite{GMXIII} use this argument to prove that every minor-closed property of undirected graphs can be tested in $O(n^3)$-time (now improved to almost linear time by the above \cite{korhonen2024minor}). A proof of \cref{conj:wqo_dreg} together with the main result of \cite{EDP_Euler}---The \textsc{Edge-Disjoint Paths} problem is fixed-parameter tractable for Eulerian digraphs---would imply that every property of Eulerian digraphs of bounded degree closed under strong immersion can be tested in polynomial time. Note here that in \cite{EDP_Euler}, the authors do not solve the more general \textsc{Folio} problem (see \cite{GMXIII}), hence their results imply that the \textsc{Immersion Testing} problem---given two Eulerian digraphs $H$ and $G$, does $H \hookrightarrow G$?---is polynomial-time solvable depending on $H$, that is, it has a running time of $O(\Abs{V(G)}^{\Abs{V(H)}})$. 

Finally, combining this with the above discussion on pivot-minors and circle graphs, a proof of \cref{conj:wqo_dreg} would further imply that \textsc{$(\PPP,\SSS)$-Membership Testing}, for any fixed pivot-minor closed class $\PPP$, can be decided in polynomial time on the class of circle graphs $\SSS$. \textsc{Pivot-Minor Testing} as well as \textsc{Vertex-Minor Testing} are known to be NP-complete \cite{computingpivotminors,dahlberg2022complexity}) in general, and there has been extensive research towards positive results ever since. 

\paragraph{Unbounded Degree.} Recall that in the above discussion we had to restrict to classes of Eulerian digraphs of \emph{bounded degree}. Unfortunately, \Cref{thm:antichain} shows that there is no hope to extend \cref{conj:wqo_dreg} to the class of \emph{all} Eulerian digraphs. However, we believe that this is an artefact of the \emph{strong} immersion relation. As the construction in \cref{thm:antichain} shows, immersion forces us to map the high degree vertices to each other, and thus they become unusable for routing other paths when dealing with strong immersion. For weak immersion this is not the case and the high degree vertex in the target digraph can still be used to route further paths. We believe that this is not just an artefact of our construction, but a fundamental difference; we conjecture the following. 

\begin{conjecture}\label{conj:wqo_gen}
    The class of Eulerian digraphs is well-quasi-ordered by weak immersion.
\end{conjecture}

This marks a crucial distinction to undirected graphs: If the strong Nash-Williams Conjecture is true (which is commonly believed in the area), then undirected graphs are well-quasi-ordered by strong immersion, whereas Eulerian digraphs are not, but a proof of \cref{conj:wqo_gen} would imply that both are well-quasi-ordered by \emph{weak} immersion. We want to emphasise that a proof of \cref{conj:wqo_gen} would immediately lift the above applications regarding \textsc{Membership Testing} to weak immersion and the class of all Eulerian digraphs. That is, any property closed under weak immersion could be tested in polynomial time on the class of Eulerian digraphs. 

\medskip

At the time of writing, we believe that we have a clear strategy to tackle both \cref{conj:wqo_dreg,conj:wqo_gen}. This paper is a crucial first step in that direction, and its goal is twofold. On the one hand, we establish the main framework---see \cref{sec:knitworks}---developed towards tackling \cref{conj:wqo_dreg,conj:wqo_gen}, and prove important results, developing some key tools culminating in a ``meta theorem'', opening the way to an inductive proof strategy for both conjectures, via decomposing our graphs into pieces of desirable structure; see \cref{sec:decomposition_theorem}. Such decomposition theorems lie at the core of well-quasi-ordering results exploiting structural graph theory as developed by Robertson and Seymour \cite{GMIV,GMXIII,GMXIX}. Note that the framework we develop is designed towards \emph{strong} immersions. In light of the introductory discussion, we believe that the presented results may be of independent interest, and that the techniques and tools may be extended to undirected graphs in the hopes of helping prove the strong Nash-Williams Conjecture and related results.

On the other hand, we prove the first important base case of \cref{conj:wqo_dreg}, handling the case where our graphs lack ``structural richness''. The correctness of the base case discussed in this paper is a crucial first step towards a full proof of \cref{conj:wqo_dreg,conj:wqo_gen}, and the form in which we prove it can additionally function as another strong ``meta theorem'', suited for inductive arguments when proving well-quasi-ordering results by strong immersion, as we briefly elaborate on below. 
The following is a much weaker form of said theorem.
\begin{theorem}\label{thm:intro}
    The class of Eulerian digraphs admitting bounded carving width is well-quasi-ordered by strong immersion.
\end{theorem}
We defer its proof as well as the full statement of the stronger version to \cref{sec:carving_width}. In a nutshell, the stronger version is tailored to Eulerian digraphs ``rooted'' in cuts, where its vertices are labelled from a well-quasi-order, and come with additional ``conformity'' restrictions. We postpone a formal definition of carving width to \cref{sec:carving_width} but highlight that \cref{thm:intro} is equivalent to the following.
\begin{theorem}\label{thm:intro_2}
    The class of Eulerian digraphs admitting bounded treewidth and bounded degree is well-quasi-ordered by strong immersion.
\end{theorem}

We want to mention that very recently Lunel and Maria \cite{lunel2025well} independently proved a weaker variant of \cref{thm:intro}, i.e., they proved that $4$-regular Eulerian digraphs of bounded carving width are well-quasi-ordered by strong immersion.

\smallskip

Finally, we conclude the paper with a result that complements \cref{cor:antichain_tw} and \cref{conj:wqo_gen}, essentially proving that the class $\CCC(2,4)$ of planar Eulerian digraphs of treewidth at most $2$, admitting at most one vertex of degree larger than $4$ is \emph{not} well-quasi-ordered by strong immersion, but it \emph{is} well-quasi-ordered by weak immersion. This is a dichotomy result in the following sense: If we additionally bound the maximum degree of $\CCC(2,4)$, then the class is well-quasi-ordered with respect to both weak and strong immersion by \cref{thm:intro}. If we restrict to treewidth at most $1$, then it is not hard to see that this class is also well-quasi-ordered by both weak and strong immersion (it consists of bidirected trees, possibly with parallel edges).

\medskip

We start with an extensive preliminary chapter. Our notation is mostly standard, except for the following non-standard definitions: Directed graphs are viewed as incident structures $G=(V,E,\operatorname{inc})$ together with maps $\head,\tail: E \to V$, and are allowed to have loops and parallel edges. What is commonly referred to as trails (like in the above introduction) will be called \emph{paths} in our setting. In fact, in this exposition paths are sequences of non-repeating edges, where paths may visit vertices repeatedly, essentially mimicking ``standard'' paths in the linegraph\footnote{We omit an explicit definition here, see for example \cite{Diestel2017,Bang-JensenG2018}.}. Given a path $P=(e_1,\ldots,e_\ell)$ we write $V^\circ(P) \coloneqq \{\tail(e_j) \mid 2 \leq j \leq \ell\}$ to mean the \emph{internal} vertices of a path. And we call a path \emph{linear} if no internal vertex is visited twice (but the endpoints may still agree). A linkage $\LLL$ is a collection of edge-disjoint paths, and is called \emph{linear} if every $P \in \LLL$ is linear. This compromise allows us not to switch to a setting of directed graphs where some edges admit only ``half their incidences'', which, while being the most natural way to work with immersions, would only be needed in very few scenarios in this exposition, whence we decided to stick to the above definitions instead. We refer the reader to \cref{sec:preliminaries} for details and further definitions.

\section{Preliminaries}
\label{sec:preliminaries}

Let $E$ be a set of elements and let $E_1 \subseteq E$ be a non-empty subset of $k \geq 1$ elements. An \emph{ordering} $\pi=(e_1,\ldots,e_k)$ of $E_1$ is a tuple encoding a linear order on $E_1$. We write $\pi(e) < \pi(e')$ meaning that $e$ precedes $e'$ in the ordering. Let $E_2 \subseteq E$ be a non-empty subset of $\ell \geq 1$ elements disjoint from $E_1$ and let $\pi'=(f_1,\ldots,f_\ell)$ be an ordering of $E_2$. We define $\pi \circ \pi' = (e_1,\ldots,e_k,f_1,\ldots,f_\ell)$, which is an ordering of $E_1 \cup E_2$. Given the above sets and orderings and a map~$\gamma:E_1 \to E$ we define~$\gamma(\pi) = (\gamma(e_1),\ldots,\gamma(e_k))$ and thus~$\gamma$ naturally lifts to a map~$\gamma:E_1^k \to {E}^k$ by applying~$\gamma$ componentwise. Similarly for a set~$S$ we define~$\gamma(S) \coloneqq \{\gamma(s) \mid s \in S\}$. 

Let $I \subseteq \N$. Then we call $(e_i)_{i \in I}$ with $e_i \in E$ for every $i \in I$ a \emph{sequence in $E$}. It is \emph{finite} if $I$ is finite, in which case we may write it as a finite tuple $(e_1,\ldots,e_m)$ where $m \coloneqq \Abs{I}$ for simplicity. In particular a finite sequence $(e_1,\ldots,e_m)$ uniquely defines an ordering on $\{e_1,\ldots,e_m\}$. 

\subsection{Well-Quasi-Ordering }

 A \emph{quasi-order}~$\Omega = (V,\preceq)$ is a tuple consisting of a set $V$ of elements and a relation~$\preceq \subseteq V\times V$ satisfying the following. Let~$u,v,w \in V$ then~$\preceq$ is \textit{reflexive}, i.e., $v \preceq v$, and \emph{transitive}, i.e.,~$u \preceq v$ and~$v \preceq w$ implies~$u \preceq w$. If in addition the relation is \textit{antisymmetric}, i.e.,~$v \preceq w$ and~$w \preceq v$ implies~$ v = w$, we call~$\preceq$  a \emph{partial order}.

We say that~$\Omega=(V,\preceq)$ is a \emph{well-quasi-order} if for every infinite sequence~$(v_\ell)_{\ell \in \N}$ of elements~$v_\ell \in V$, there exist~$1 \leq i < j $ such that~$v_i \preceq v_j$. Given an infinite index set~$I \subseteq \N$ such that for all~$i,j \in I$ with~$i<j$ it holds~$v_i \not \preceq v_j$, we call~$(v_\ell)_{\ell \in I}$ a \emph{bad sequence (with respect to $\preceq$)}. If additionally for all~$i,j \in I$ it holds~$v_j \not \preceq v_i$, then we call it an \emph{antichain}, and if~$v_i \succ v_j$ for all~$i,j \in I$ with~$i < j$ we call it a \emph{strictly decreasing sequence}. Completing the above, if~$v_i \preceq v_j$ for every~$i,j \in I$ with~$i<j$ we call the sequence a \emph{chain with respect to~$\preceq$}. Clearly,~$\Omega$ is a well-quasi-order if and only if there exists no infinite bad sequence with respect to $\preceq$.  Given another well-quasi-order $\Omega^+=(V^+,\preceq^+)$ we say that $\Omega^+$ \emph{extends} $\Omega$---written $\Omega \subseteq \Omega^+$---if and only if $V \subseteq V^+$ and $\restr{\preceq^+}{V \times V} = \preceq$, and we call $\Omega^+$ an \emph{extension} of $\Omega$.

We gather some well-known results regarding well-quasi-orderings.

\begin{observation}\label{obs:wqo_yields_infinite_chain}
    Let~$\Omega$ be a well-quasi-order and~$(v_\ell)_{\ell \in \N}$ a sequence with~$v_\ell \in V(\Omega)$ for every~$\ell \in \N$. Then there exists an infinite index set~$I\subset \N$ such that~$(v_i)_{i \in I}$ is a chain with respect to~$\preceq$.
\end{observation}

\begin{observation}\label{obs:wqo_of_tuples}
    Let $k \in \N$ and let~$\Omega_1,\ldots,\Omega_k$ be well-quasi-orders with~$\Omega_i = (V_i,\preceq_i)$ for~every~$1 \leq i \leq k$. Then~$\Omega=(V(\Omega_1)\times \ldots \times V(\Omega_k), \preceq)$ where~$(v_1,\ldots,v_k)\preceq(u_1,\ldots,u_k)$ if and only if~$v_i \preceq_i u_i$ for all~$1 \leq i \leq k$ is a well-quasi-order.
\end{observation}

And finally we will need the following well-known result due to Higman \cite{Higman1952}.

\begin{theorem}[Higman's Theorem]\label{thm:higman}
   Let~$\Omega$ be a well-quasi-order and~$(X_i)_{i \in \N}$ a sequence of finite sets~$X_i \subseteq V(\Omega)$. Then there exists an infinite index set~$I\subseteq \N$ such that for every~$1 \leq i < j$ with~$i,j \in I$ there is an injective map~$\alpha_i^j:X_i \to X_j$ such that~$x \preceq \alpha_i^j(x)$ for every~$x \in X_i$.
\end{theorem}

\subsection{Graph Theory}

All the graphs considered in this paper have finite vertex and edge sets, unless stated otherwise. We denote directed graphs (or digraphs) by~$G=(V,E,\operatorname{inc})$ where~$\operatorname{inc}\subseteq (V\times E)\cup (E\times V)$ is a binary relation, and for every~$e \in E$ there are exactly two incidences in~$\operatorname{inc}$, i.e., there are two (possibly equal) unique vertices~$v_i,v_o \in V$ such that~$(v_i,e),(e,v_o) \in \operatorname{inc}$. We define functions~$\tail_G,\head_G: E \to V$ via~$\tail_G(e)=v_i$ for~$(v_i,e) \in \operatorname{inc}$ and~$\head_G(e) = v_o$ for~$(e,v_o) \in \operatorname{inc}$. If $G$ is clear from context we may omit the subscripts and simply write $\tail$ and $\head$. Since~$\tail,\head$ are uniquely defined we may equally well write~$G=(V,E,\tail,\head)$ for directed graphs meaning the obvious. 

As mentioned in the introduction, we have a strategy towards proving \cref{conj:wqo_dreg,conj:wqo_gen} and thus we want to keep notation throughout results in that direction consistent. In particular we will crucially use edges as ``stand-alone'' objects. However, for simplicity and ease of readability we may write~$e=(u,v)$ for~$e\in E$ to mean~$\tail(e) = u$ and~$\head(e)=v$; i.e., we may view~$E \subseteq V \times V$ as a multiset (we allow for parallel edges). Similarly, we may write~$G=(V,E)$ for some~$E\subseteq V \times V$ for convenience, meaning the obvious.
Note further that the above definition allows for \emph{loops}, where an edge~$(u,v) \in E(G)$ is called a \emph{loop} if and only if~$u = v$; we write~$\loops(v)$ for the set of loops incident to~$v$ and~$\loops(G)$ for the set of all loops in~$G$\footnote{We will ``get rid'' of loops in \cref{sec:knitworks}}. 

Since we will explicitly work with graphs $G,G'$ where $e \in E(G) \cap E(G')$ is allowed, $\tail_G(e)$ and $\tail_{G'}(e)$---as well as $\head_G(e)$ and $\head_{G'}(e)$---may be distinct vertices.

Given~$X \subseteq V(G)$ we write~$\bar{X} \coloneqq V(G) \setminus X$. Given~$U\subseteq V(G)$ we write~$N_G^+(U) \coloneqq \{v \mid (u,v) \in E(G), \text{ for some }u\in U\}$, and~$N_G^-(U) \coloneqq \{v \mid (v,u) \in E(G), \text{ for some }u\in U\}$, and~$N_G(U) \coloneqq N_G^+(U) \cup N_G^-(U)$. If~$G$ is clear from context, we may omit the subscript and if~$U = \{u\}$ we write~$N^+(u)$ instead of~$N^+(\{u\})$ and analogously for the other notions. Further we write $G[X]$ to mean the subgraph of $G$ \emph{induced by $X$}, i.e., it is given by $V(G[X]) = X$ and $E(G[X])  = X\times X \cap E(G)$ (or rather by the respective incidence structure).

Given two digraphs~$H=(V_H,E_H)$ and~$G=(V_G,E_G)$ with $V_H \subseteq V_G$ we write~$G-H \coloneqq (V_G,E_G \setminus E_H)$. Given~$v \in V(G)$ we call~$v$ \emph{Eulerian} if the number of edges~$e \in E(G)$ with $\tail(e) = v$ equals the number of edges~$e' \in E(G)$ with~$\head(e') = v$, i.e., the in- and out-degree at~$v$ are equal. We call~$G$ \emph{Eulerian} if every vertex in~$V(G)$ is Eulerian; notably we do not impose connectivity restrictions.

\paragraph{Paths, Cycles and Linkages.}As we mostly work with edge-disjointness, we use the following non-standard definitions of paths and cycles; intuitively a \emph{path} in our setting is a ``path'' in the ``linegraph'' according to standard literature, and similarly for cycles.

\begin{definition}[Paths, Cycles and Circles]\label{def:paths}
   Let~$G=(V,E,\operatorname{inc})$ be an Eulerian digraph. Let~$P =(e_1,\ldots,e_m)_G$ be a finite sequence of distinct edges~$e_1,\ldots,e_m \in E$ for some~$ m\in \N$ such that for every~$1 \leq i < m$ there is~$v_i \in V(G)$ with~$\head_G(e_i) = v_i = \tail_G(e_{i+1})$. Then we call~$P$ a \emph{path} in~$G$ and let~$m$ denote its \emph{length}. We define~$V^\circ(P) \coloneqq \{v_1,\ldots,v_{m-1}\}$\Symbol{V(P)@$V^\circ(P)$} to be the internal vertices of $P$. We further define $V(P) \coloneqq V^\circ(P) \cup \{v_0,v_m\}$ and~$E(P) \coloneqq \{e_1,\ldots,e_m\}$. The vertices~$v_0 = \tail_G(e_1)$ and~$v_{m}=\head_G(e_m)$---possibly equal---are called the \emph{endpoints} of~$P$ and the vertices in~$V^\circ(P)$ are called \emph{internal vertices} of~$P$. If the graph is clear from context, we may omit the subscript and write~$P = (e_1,\ldots,e_m)$. We call~$e_1,e_m$ the \emph{ends} of~$P$ and call~$e_1$ the \emph{first} edge of~$P$ and say that~$P$ \emph{starts in} $e_1$, and we call~$e_m$ the \emph{last} edge of~$P$ and say that~$P$ \emph{ends in}~$e_m$. 

    If $V^\circ(P) \cap \{v_0,v_m\} = \emptyset$, we call $P$ \emph{clean}.\index{trail!clean}\index{clean path}
    If in addition to being clean, for every~$1 \leq i ,j < m$ with~$i \neq j$ we have~$v_i \neq v_j$, then we call~$P$ \emph{linear}. Let~$1 \leq j < \ell \leq m$, then we call~$(e_j,\ldots,e_\ell) $ a  \emph{subpath} of~$P$, and we write $(e_j,\ldots,e_\ell) \subseteq P$. In particular, when writing~$(f_1,\ldots,f_t) \subseteq P$ we mean a subpath of~$P$ of length~$t$, i.e.,~$f_1 = e_j$ and~$f_t= e_{j+t}$ for some~$1 \leq j \leq j+t \leq m$.

   If in addition~$P'=(e_2,\ldots,e_m,e_1)_G$ is a path, then we call~$C \coloneqq (e_1,\ldots,e_m,e_1)_G$ \emph{a cycle} in~$G$ (omitting the subscript if clear from context). If $P$ and $P'$ are linear we call~$C$ a \emph{circle}, that is~$v_m=v_0$ and for every~$1 \leq i,j \leq m$ with~$i\neq j$ we have~$v_i \neq v_j$. 

   Given two edges~$e,e' \in E(G)$ and a path~$P $ in $G$ with ends~$e,e'$, we call~$P$ an \emph{$\{e,e'\}$-path}, and if~$P$ starts in~$e$ and ends in~$e'$ we call it an~$(e,e')$-path and define~$\tau(P) = (e,e')$\Symbol{TAUP@$\tau(P)$} to be its \emph{type}. Similarly, we refer to~$P$ as a \emph{$\{v_0,v_m\}$-path} if its endpoints are $v_0,v_m$ and a \emph{$(v_0,v_m)$-path} if~$v_0 = \tail(e)$ and~$v_m = \head(e')$ and call $v$-$\tau(P) \coloneqq (v_0,v_m)$ its \emph{vertex-type}. More broadly, if $P$ has one endpoint in $A \subseteq V(G)$ and the other in $B \subseteq V(G)$ we call it an \emph{$\{A,B\}$-path} and an $(A,B)$-path if its tail is in $A$.
\end{definition}
\begin{remark}
     Note that by definition, a (linear) path~$(e_1,\ldots,e_m)$ may have a single endpoint~$\tail(e_1) = v_0 = v_m = \head(e_m)$.
    Note further that, if we have two graphs~$H,G$ and a set of edges~$\{e_1,\ldots,e_m\} \subset E(G) \cap E(H)$ then we may write~$(e_1,\ldots,e_m)_H,(e_1,\ldots,e_m)_G$ to mean the respective paths in the respective graphs (if they are indeed paths). Furthermore, by definition, a linear path~$P=(e_1,\ldots,e_m)_G$ induces a linear order on~$V^\circ(P)$ given by~$\head_G(e_i) \prec \tail_G(e_{i+1})$ for all~$1\leq i < m$. Also, note that every strict subsequence of a circle is a path.
\end{remark}
Thus, a ``path'' in the standard literature---e.g. \cite{Diestel2017,Bang-JensenG2018}---is in our terminology a linear path with non-agreeing endpoints.

Given a graph $G$ and a path $P=(e_1,\ldots,e_m)$ in $G$, the graph $G_P \coloneqq (V(P), E(P)) $ is a subgraph of $G$. In particular, we may write $P \subseteq G$ for simplicity if the exact ordering of the edges is not needed, noting that this is ambiguous in the sense that, according to our definition, two distinct paths $P_1,P_2$ may satisfy $G_{P_1} = G_{P_2}$. Similarly, we may identify $P$ with $G_P$ for simplicity and write $P_1 \cup P_2$ to mean the graph consisting of $G_{P_1}\cup G_{P_2}$. We use analogous notation for cycles and circles.

The following is folklore.
\begin{observation}\label{obs:covering_eulerian_digraphs}
    Let $G$ be an Eulerian digraph. Then there exists a set $\CCC$ of edge-disjoint circles in $G$ such that $G = \bigcup_{C\in \CCC} C$. Further, there exists a cycle $C \subseteq G$ such that $G_{C} = G$. 
\end{observation}

In light of \cref{obs:covering_eulerian_digraphs} we define the following.
\begin{definition}
    Let $G$ be an Eulerian digraph, and $\CCC$ a set of edge-disjoint cycles (or circles) in $G$ such that $G = \bigcup_{C\in \CCC} C$. Then we call $\CCC$ a \emph{cycle-cover} of $G$ (or circle-cover respectively).

    If $C\subseteq G$ is a cycle satisfying $G_C = G$, then we call $C$ an \emph{Eulerian cycle}.
\end{definition}

We refer to the graph obtained from~$G$ by omitting the directions of its edges as the \emph{underlying undirected graph}. 
A directed graph $G$ is \emph{weakly connected} if the underlying undirected graph is connected, and it is \emph{strongly connected} if for every pair of distinct vertices $u,v \in V(G)$, there is a $(u,v)$-path and a $(v,u)$-path in $G$. Note that by \cref{obs:covering_eulerian_digraphs} every strongly connected Eulerian digraph can be traced by a single Eulerian cycle and thus a cycle as given by \cref{def:paths} is just a strongly connected Eulerian digraph with a fixed linear order on its edges.
\smallskip

For notational convenience, we may write paths and cycles as ordered tuples of vertices when needed. That is~$P=(e_1,e_2,\ldots,e_m)$ may be rewritten as~$P=(v_0,v_2,v_3,\ldots,v_{m})$ where~$e_i=(v_{i-1},v_{i})$ for~$1 \leq i \leq m$; note that this in only unambiguous if the respective edges are clear from context, as we allow for parallel edges. Similarly, we may write~$(v_0,e_1,v_1,e_2,v_2,\ldots,e_m,v_{m})$ for paths to highlight the adjacency of vertices and edges. 

 Further, to avoid corner cases, we allow paths of length~$0$ which consist of isolated vertices, i.e.,~$(v)$ is considered a path with both endpoints being~$v \in V(G)$.

\begin{definition}[Disjoint Paths and Linkages]
    Let~$G$ be a digraph and let~$k,\ell \in \N$. Two paths $P_1=(e_1,\ldots,e_k),P_2=(f_1,\ldots,f_\ell)$ in~$G$ are \emph{edge-disjoint} or \emph{vertex-disjoint} if~$E(P_1) \cap E(P_2) = \emptyset$ or~$V(P_1)\cap V(P_2) = \emptyset$ respectively.

    The paths~$P_1,P_2$ are \emph{internally edge-disjoint} if for~$e \in E(P_1) \cap E(P_2)$ it holds~$e\notin \{e_2,\ldots,e_{k-1}\}\cup\{f_2,\ldots,f_{\ell-1}\}$ and \emph{internally vertex-disjoint} if~$V^\circ(P_1) \cap V^\circ(P_2) = \emptyset$.

    Given a set of edge-disjoint paths~$\LLL$ we call~$\LLL$ a \emph{linkage (in $G$)}\index{linkage}. We call the linkage \emph{clean}\index{linkage!clean}\index{clean linkage} if every path in $\LLL$ is clean, and \emph{linear}\index{linkage!linear}\index{linear linkage} if each path in~$\LLL$ is linear. We call $\LLL$  \emph{strong} if it is clean and, given the set $T\subseteq V(G)$ of endpoints of paths in $\LLL$, no vertex in $T$ is an internal vertex of a path in $\LLL$. 
    
    We call~$k\coloneqq\Abs{\LLL}$ the \emph{order}\index{linkage!order} of the linkage, and call $\LLL$ a $k$-linkage, respectively. Finally, we define~$\tau(\LLL) \coloneqq \{\tau(P) \mid P \in \LLL\}$\Symbol{TAULLL@$\tau(\LLL)$} as the \emph{type} of the linkage and define $v$-$\tau(\LLL)$\Symbol{VTAULLL@$v{-}\tau(\LLL)$} analogously.
    
    We fix $E(\LLL) \coloneqq \bigcup_{L \in \LLL} E(L)$ and we say that a path $P \subseteq G$ is \emph{edge-disjoint from $\LLL$} if and only if $E(P) \cap E(\LLL) = \emptyset$.
\end{definition}
\begin{remark}
    Note that a strong linkage is a clean linkage but it is not necessarily linear.
\end{remark}

The definitions of edge-disjointness and vertex-disjointness lift to cycles and circles in a straightforward way. The following is obvious by re-routing.
\begin{observation}\label{obs:linkage_gives_linear_linkage}
    Let~$G$ be a digraph and let~$\LLL=\{L_1,\ldots,L_t\}$ be a linkage in~$G$ of order~$t \geq 1$. Let $\CCC=\{C_1,\ldots,C_t\}$ a set of edge-disjoint cycles in $G$. Then the following hold true.
    \begin{enumerate}
        \item There exists a linear linkage~$\LLL'=\{L_1',\ldots,L_t'\}$ in $G$ of order~$t$ with~$E(L_i') \subseteq E(L_i)$ and $\tau(L_i') = \tau(L_i)$ for every $1 \leq i \leq t$.\label{obs:linkage_gives_linear_linkage:1}
        \item  There exists a clean linkage $\LLL''$ of order $t$  with $\tau(\LLL'') =\tau(\LLL)$ such that every path in $\LLL'$ is a subpath of a path in $\LLL$.\label{obs:linkage_gives_linear_linkage:2}
        \item  There exists a set $\CCC'$ of edge-disjoint circles in $G$ together with a partition $\{\CCC_1,\ldots,\CCC_t\}$ of $\CCC'$ such that $\bigcup_{C\in \CCC_i} C = C_i$ for every $1 \leq i \leq t$.\label{obs:linkage_gives_linear_linkage:3}
    \end{enumerate}
\end{observation}
\begin{remark}
    Note that while $\LLL'$ is linear (and thus clean) it does not necessarily satisfy the properties of $\LLL''$ which in turn may not be linear.
\end{remark}

Refining on the notion of type of linkage, we define the following.
\begin{definition}
    Let~$G$ be a digraph and let~$A,B \subseteq V(G)$. Let~$\LLL$ be a linkage in~$G$ such that every~$P \in \LLL$ has one endpoint in~$A$ and one endpoint in~$B$, then we call~$\LLL$ an \emph{ $\{A,B\}$-linkage}.

    Similarly, let~$E,F \subset E(G)$ be non-empty and disjoint. Let~$\LLL$ be a linkage in~$G$ such that every~$P \in \LLL$ has one end in~$E$ and one end in~$F$ and is otherwise edge-disjoint from~$E\cup F$, then we call~$\LLL$ a \emph{$\{E,F\}$-linkage}. If every path in $L$ starts in $E$ and ends in $F$ we call $\LLL$ an \emph{$(E,F)$-linkage}.
\end{definition}

\begin{definition}[Concatenation of Paths]\label{def:concatenation_of_paths}
      Let $P_1=(e_1,\ldots,e_k),P_2=(f_1,\ldots,f_\ell)$ be internally edge-disjoint paths in some graph $G$ such that~$e_k = f_1$ and $e_1 \neq f_\ell$. Then we define~$P = P_1 \circ P_2 \coloneqq (e_1,\ldots,e_k,f_2,\ldots,f_\ell)$ and call~$P_1,P_2$ the \emph{summands} of~$P$.
      Similarly if~$P_1,P_2$ are edge-disjoint and there exists~$v \in V(G)$ with~$ (e_k,v) \in \operatorname{inc}$ and~$(v,f_1) \in \operatorname{inc}$, we define~$P \coloneqq P_1 \circ P_2 = (e_1,\ldots,e_k,f_1,\ldots,f_\ell)$.
\end{definition}
It is straightforward to verify that~$P$ is a path using our definition of paths.

\begin{observation}\label{obs:concat_paths}
    Let $P_1=(e_1,\ldots,e_k),P_2=(f_1,\ldots,f_\ell)$ be internally edge-disjoint paths in a graph~$G$ such that~$e_k = f_1$ and $f_\ell \neq e_1$. Then~$P_1 \circ P_2$ is a path in~$G$. If~$P_1$ and~$P_2$ are edge-disjoint and there exists~$v \in V(G)$ with~$(e_k,v) \in \operatorname{inc}$ and~$ (v,f_1) \in \operatorname{inc}$, then~$P_1 \circ P_2$ is a path.
    
    If in addition to their edge-disjointness,~$P_1$ and $P_2$ are internally vertex-disjoint linear paths, then $P_1 \circ P_2$ is again a linear path.
\end{observation}

\paragraph{Undirected graphs.} We abuse notation and denote undirected graphs by~$G=(V,E)$ where~$E \subseteq \{ \{v,w\} \mid v,w \in V(G)\}$ may be a multiset, noting that it will be clear from the context whenever we talk about undirected graphs. A \emph{tree}~$T=(V,E)$ is a connected acyclic undirected graph. We call vertices of degree~$1$ in~$V(T)$ \emph{leaves}. If $T$ is a tree, we denote the set of its leaves by $\leaves{T} \subseteq V$. We call a tree \emph{cubic} if apart from its leaves all its vertices are of degree three.

\subsection{Immersions}
The following is the minor relation of interest in this paper.

\begin{definition}[Immersion]
\label{def:immersion}
    Let~$G,H$ be directed graphs. Then \emph{$H$ immerses in $G$}, or \emph{$G$ immerses $H$}, if there exists a map $\gamma$ defined on~$V(H) \cup E(H)$ satisfying the following:
    \begin{enumerate}[label=(\arabic*)]
        \item $\restr{\gamma}{V(H)}: V(H) \to V(G)$ is injective,\label{def:immersion:1}
        \item for every $e \in E(H)$, $\gamma(e)$ is a path in $G$ and for distinct~$e_1,e_2 \in E(H)$ the paths $\gamma(e_1),\gamma(e_2)$ are edge-disjoint, and\label{def:immersion:2}
        \item let~$e=(u,v)\in E(H)$ for some~$u,v \in V(H)$, then~$\gamma(e)$ starts in~$\gamma(u)$ and ends in~$\gamma(v)$.\label{def:immersion:3}
    \end{enumerate}
We call an immersion \emph{strong} if for every~$e\in E(H)$,~$V^\circ(\gamma(e)) \cap \gamma(V(H)) = \emptyset$. We write $H \hookrightarrow G$ if $G$ strongly immerses $H$ and $\gamma\colon H \hookrightarrow G$ if $\gamma$ is a map witnessing it. Similarly we write $H\hookrightarrow^* G$ if $G$ immerses $H$. We may say that $G$ \emph{weakly} immerses $H$ to emphasize that the immersion is not strong.

More generally, we write $\gamma\colon V(H) \cup E(H) \to G$ for a map $\gamma$ with $\restr{\gamma}{V(H)}\colon V(H) \to V(G)$ and $\restr{\gamma}{E(H)}\colon E(G) \to \sigma(E(G))$ where $\sigma(E(G))$ is the set of finite sequences in $E(G)$.
\end{definition}
\begin{remark}
    As we allow for loops~$e=(v,v) \in E(G)$, the above definition is only correct since we allow paths to start and end at the same vertex, i.e.~$\gamma(e)$ starts and ends in~$\gamma(v)$.

\end{remark}

Given an immersion $\gamma: H\hookrightarrow G$ and a path $P=(e_1,\ldots,e_k)$ in $H$ we define $\gamma(P)\coloneqq \gamma(e_1) \circ \ldots \circ \gamma(e_k)$, which is easily seen to be a well-defined path in $G$, summarized as follows.

\begin{observation}\label{obs:immersion_maps_path_to_path}
    Let~$\gamma:E(H) \cup V(H) \to G$ be an immersion of directed graphs~$H$ and $G$. Let~$P$ be a path in~$H$. Let $\LLL$ be a $k$-linkage for some $k \geq 1$. Then~$\gamma(P)$ is a path and $\gamma(\LLL)$ is a $k$-linkage in $G$. If additionally $\gamma$ is strong then $\LLL$ is strong.
\end{observation}
\begin{remark}
    Note that linear paths may be mapped to non-linear paths, i.e., they may self-intersect at vertices, and similarly if $\gamma$ is not strong, then strong linkages and linear linkages may be mapped to simple linkages that are neither strong nor linear.
\end{remark}

The following is clear by the above and the \cref{def:immersion} of immersions.
\begin{observation}\label{obs:strong_immersion:yields_strong_linkage}
Let~$G,H$ be directed graphs. Let $\gamma: H\hookrightarrow G$, then $\gamma(E(H))$ is a strong linkage in~$G$. 
Let $\eta: H\hookrightarrow^* G$, then $\eta(E(H))$ is a linkage in $G$, but it is not necessarily strong.
\end{observation}

We want to emphasize here, that our definition of immersion is slightly different, than in the general literature, where $\gamma(e)$ is in general assumed to be a \emph{linear} path. Note that \cref{obs:linkage_gives_linear_linkage} implies that the definitions are equivalent. However, we have a good reason and will clarify why we chose this definition in \cref{sec:knitworks}. The following is well-known (even using the above definition).
\begin{observation}
(Strong) Immersion defines a quasi-order on digraphs.
\end{observation}

We say that two (di)graphs $G,H$ are \emph{isomorphic} and write $H \cong G$ if there is a map $\eta:H \to G$ that bijectively maps vertices to vertices and edges to edges respecting their incidences. Note that if~$G \hookrightarrow H$ and~$H \hookrightarrow G$ then~$H \cong G$.

In light of \cref{def:immersion} we define \emph{immersion models}, or simply \emph{models}, as follows.

\begin{definition}[Immersion model]
     Let~$H$ and $G$ be Eulerian digraphs. Let $V_H \subseteq V(G)$ and let $\LLL$ be a (strong) linkage in $G$ of order $\Abs{E(H)}$. If there are bijections $\gamma_V:V(H) \to V_H$ and $\gamma_E:E(H) \to \LLL$ such that for every $e=(u,v) \in E(H)$ it holds $v$-$\tau(\gamma_E(e)) = (\gamma_V(u),\gamma_V(v))$, we call $(V_H,\LLL)$ a \emph{(strong) immersion model of $H$ (in $G$)}.
\end{definition}

The following is imminent.
\begin{observation}\label{obs:model=immersion}
    Let~$H$ and $G$ be Eulerian digraphs. Then $G$ (strongly) immerses $H$ if and only if there is a (strong) immersion model of $H$ in $G$.
\end{observation}
\begin{proof}
       We give a proof for the strong variant for the other is analogous. If $\gamma:H\hookrightarrow G$ is a strong immersion, then $(\gamma(V(H)),\gamma(E(H))$ is a strong model of $H$ in $G$ by \cref{def:immersion}. Conversely, let $(V_H,\LLL)$ be a strong model of $H$ in $G$. By \cref{obs:linkage_gives_linear_linkage} there exists a strong model $(V_H,\LLL')$ of $H$ where $\LLL'$ is a linear linkage. Let $\eta:V(H) \to V_H$ and $\zeta:E(H) \to \LLL'$ be the bijection as in the definition of strong model. Define $\gamma:V(H)\cup E(H) \to G$ via $\restr{\gamma}{V(H)} = \eta$ and $\restr{\gamma}{E(H)}=\zeta$, then this does the trick.
\end{proof}

\paragraph{Splitting Off.} There are different ways to define minors for undirected graphs, one of which is via graph-modifying operations, namely taking subgraphs and contracting edges. Regarding immersions we have the following operation.

\begin{definition}[Splitting off] \label{def:splitting_edgepairs}
    Let~$G=(V,E,\operatorname{inc})$ be a directed graph, let~$v \in V(G)$, and let~$e,e' \in E(G)$ such that~$(e,v),(v,e') \in \operatorname{inc}$. Let~$u\coloneqq\tail(e)$ and~$w \coloneqq \head(e')$. \emph{Splitting off~$(e,e')$ (at~$v$) in~$G$} results in a new graph~$G'$ by deleting~$e,e' \in E(G)$ and~$(u,e),(e,v),(v,e'),(e',w) \in \operatorname{inc}$ and subsequently adding a new edge~$f \notin V(G) \cup E(G)$ together with the incidences~$(u,f),(f,w)$. We remove~$v$ from the vertex set if and only if it is an isolated vertex (it has no incidences left) after splitting off; we define~$(G',f) \coloneqq \spl(G;(e,e'))$ or simply write~$G' \coloneqq \spl(G;(e,e'))$ if we do not specifically need~$f$ for simplicity.
\end{definition}

Note that $e$ and $e'$ may be the same edge, i.e., a loop. The following is folklore \cite{frank95,Bang-JensenG2018,EDP_Euler}.

\begin{observation}\label{obs:immersion_robust_under_splitting_off}
    Let~$H,G$ be Eulerian digraphs such that $G$ immerses $H$. Then there exists an Eulerian subgraph~$G'$ of~$G$ such that $H$ can be constructed from $G$ by sequentially splitting off edge-pairs.
    
    Conversely, given a graph $G$ and edges $e,e' \in E(G)$ with $\head(e) = \tail(e')$, it holds~${\spl(G;(e,e')) \hookrightarrow G}$.
\end{observation}

\subsection{Induced Cuts}
We start by defining  cuts and induced cuts. The way we define them is purely undirected and is thus not equivalent to the standard directed definitions of these notions. 
Observe that an Eulerian digraph is strongly connected if and only if it is weakly connected. In particular, we will only be interested in weak connectivity: We therefore fix that a digraph $G$ is \emph{connected} if it is weakly connected and \emph{disconnected} if it is not. If we need strong connectivity we will explicitly say so.

\begin{definition}[Cuts]
    Let~$G$ be a connected digraph. A \emph{cut}, is a set~$F \subseteq E(G)$ such that the digraph~$G-F$ is disconnected. We refer to~$\Abs{F}$ as the \emph{order} of the cut.

    Let~$X_1,X_2 \subset V(G)$ be disjoint strict subsets. We call~$F$ an \emph{$\{X_1,X_2\}$-cut} if there is no connected component in $G-F$ that contains a vertex of $X_1$ and a vertex of $X_2$. We define~$\delta(X_1,X_2)$ to be the minimal order of any~$\{X_1,X_2\}$-cut in $G$.

\end{definition}
\begin{remark}
    By definition if~$X_1 \cap X_2 \neq \emptyset$ then there exists no~$\{X_1,X_2\}$-cut. 
\end{remark}
The definitions can be extended to disconnected directed graphs by extending it component-wise in the obvious way.

\begin{definition}[Induced Cuts]
    Let~$G$ be a digraph and let~$X \subset V(G)$ be a strict subset. 
    
    We define~$\rho(X,\bar{X}) \coloneqq \{(u,v) \in E(G) \mid u \in X \text{ and } v\in \bar{X}\}$ and $\rho^{+}(X) \coloneqq \rho(X, \bar X)$ and $\rho^{-}(X) \coloneqq \rho(\bar{X},X)$. 
  Finally, we let~$\rho(X) \coloneqq \rho^-(X) \cup \rho^+(X)$ and refer to it as \emph{the cut induced by~$X$}. We refer to~$\delta(X) \coloneqq \Abs{\rho(X)}$ as the \emph{order of the cut}. If $\Abs{\rho^-(X)} = \Abs{\rho^+(X)}$ we call the cut \emph{Eulerian}. If $X$ induces a cut of order $k \in \N$ we say that $X$ \emph{induces a $k$-cut}.
\end{definition}
\begin{remark}
    Clearly~$\delta(X) = \delta(\bar{X})$ by definition, and similarly~$\rho^-(X) = \rho^+(\bar{X})$. Note that we have abused notation, using~$\rho$ and~$\delta$ on sets and pairs of sets of vertices for the sake of readability. It may never cause confusion which one we are talking about, since they are defined on disjoint domains.
\end{remark}

If~$X=\{v\}$ we may write~$\rho(v)$ instead of~$\rho(\{v\})$ for simplicity. Further, we let~$E(v) \coloneqq \{e \in E(G) \mid v \in \tail(e)\cup \head(e)\}$ denote the set of edges incident to~$v$. The following is straightforward from the definition.
\begin{observation}\label{obs:cut_at_vertex}
    Let~$G$ be a digraph and~$v \in V(G)$, then~$\rho(v) = E(v) \setminus \loops(v)$.
\end{observation}

It is well-known that on induced cuts the function $\delta$ is symmetric and sub-modular. That is~$\delta(X) = \delta(\bar{X})$ and~$\delta(X) + \delta(Y) \geq \delta(X \cap Y) + \delta(X \cup Y)$ for any pair~$X,Y \subseteq V(G)$. And, of course, the notions of cut and induced cuts are closely related: The following is well-known and straightforward to verify (for it is essentially defined on the undirected underlying graph).

\begin{observation}\label{lem:relating_cut_to_induced_cuts}
    Let~$G$ be a directed graph and let~$X_1,X_2 \subset V(G)$ be strict subsets such that $X_1 \cap X_2 = \emptyset$. Then~$\delta(X_1,X_2)$ is given by the minimum of~$\delta(X)$ over all~$ X \subseteq V(G)$ with~$X_1 \subseteq X$ and~$X \cap X_2 = \emptyset$.
\end{observation}

The way~$\delta$ is defined it does not come with a nice \emph{directed} ``Menger-property'' in the context of general digraphs. That is, it is not true that~$\delta(X_1,X_2)$ is equal to the maximum order of a \emph{directed}~$\{X_1,X_2\}$-linkage: Take for example a simple alternating path of length at least two and let~$X_1,X_2$ be its extremities. These pathologies disappear when restricted to Eulerian digraphs; the following is folklore.

\begin{observation}
    Let~$G$ be an Eulerian digraph and let $X \subseteq V(G)$, then $X$ induces an Eulerian cut. In particular~$\Abs{\rho(\cdot,\cdot)}$ is symmetric and~$\delta(X) \in 2\N$.
\end{observation}

Together with \cref{lem:relating_cut_to_induced_cuts} this implies a directed Menger-property for Eulerian digraphs. This is an easy consequence of the directed version of Menger's theorem for directed cuts together with the fact that, if~$G$ is Eulerian and there is a directed linear path from~$u$ to~$v$, then there exists another linear path---edge-disjoint from the first one---that connects~$v$ to~$u$. See also \cite{frank95,EDP_Euler} for further connectivity properties of Eulerian digraphs.

\begin{lemma}\label{lem:Menger_for_Euler}
  Let~$G$ be an Eulerian digraph and~$X_1, X_2 \subset V(G)$ be non-empty and disjoint.
  \begin{enumerate}
  \item For every~$k \in \N$, either there is a linear $\{X_1,X_2\}$-linkage~$\LLL$ of order $2k$ in~$G$ containing~$k$ paths from~$X_1$ to~$X_2$ and~$k$ paths from~$X_2$ to~$X_1$, such that for every $L \in \LLL$ it holds $V^\circ(P) \cap (X_1\cup X_2) = 
  \emptyset$, or $\delta(X_1,X_2) < 2k$.\label{lem:Menger_for_Euler:1}
  \item The maximum number of edge-disjoint paths from~$X_1$ to~$X_2$ is the same as the maximum number of edge-disjoint paths from~$X_2$ to~$X_1$. \label{lem:Menger_for_Euler:2}
  \end{enumerate}
\end{lemma}

\cref{lem:Menger_for_Euler} implies that the Eulerianness of a digraph guarantees that the smallest undirected~$\{X_1,X_2\}$-cut suffices to witness the maximum number of directed edge-disjoint paths between the resulting components. In particular Eulerianness guarantees this cut to be of even order, and when looking at the directions of the cut-edges, half of them go from one side to the other of the partition and vice-versa.

\section{Knitworks}
\label{sec:knitworks}

We start with defining the basic framework and the tools needed throughout the paper and future work. The main novel concept is that of \emph{$\Omega$-knitworks}, a data structure developed to handle inductive arguments when dealing with (strong) immersion of (directed) graphs. It allows us to store ``immersion types'' when carving graphs via suited labelling functions. As so often when dealing with well-quasi-ordering graphs by some containment relation \cite{GMIV,Gee02}, it will be very helpful to switch to a \emph{rooted} setting simplifying inductive reasoning for surface-embeddings by exploiting structural properties of the graph class (and embeddings) at hand. 

Throughout this section we will mostly work with loopless graphs, as it eases definitions and loops can be directly encoded in $\Omega$-knitworks as we discuss in \cref{subsec:manipulate_knitworks}. Nonetheless, we will highlight whenever graphs are assumed to be loopless for completeness (if not necessarily clear from context). Furthermore, throughout the rest of the paper, we will implicitly assume that all digraphs we consider are \emph{weakly connected} unless explicitly stated otherwise.

\subsection{The General Framework} 
\label{subsec:general framework}

\paragraph{Rooted Digraphs.} We start with a rather general definition of \emph{rooted digraphs}: We root (Eulerian) digraphs in induced cuts, i.e., our roots will be incidences of vertices with edges. While for most of the exposition it would suffice to root our graphs in a \emph{single} vertex inducing a cut, the more general definitions and results are of independent interest and will be needed for future work. 

\begin{definition}[Rooted Digraphs]\label{def:rooted_graph}
Let $\ell \geq 1$ and let~$G$ be a digraph. Let $X_1,\ldots, X_\ell \subset V(G)$ be pairwise disjoint non-empty strict subsets of $V(G)$. Let $\pi_G$ be a map assigning to $X_i$ an ordering of $\rho(X_i)$ for every $1 \leq i \leq \ell$. Then we call $\bar{G} \coloneqq (G,\pi_G, X_1,\ldots,X_\ell)$ a \index{digraph!rooted}\index{rooted digraph}\emph{rooted digraph} and $\bigcup_{i=1}^\ell X_i$ the \emph{root set}. Let $k = \sum_{1 \leq i \leq \ell} \delta(X_i)$, then $k$ is \index{rooted digraph!order}\emph{the order} of $\bar{G}$ and $\ell$ is its \index{rooted digraph!index}\emph{index}. 
If for every $1 \leq i \leq \ell$ it holds $\Abs{X_i} \leq 1$ then we call $\bar{G}$ \index{rooted digraph!controlled}\index{controlled!rooted digraph}\emph{controlled}. If $\ell=1$ then we call $\bar G$ \index{clamped!rooted digraph}\index{rooted digraph!clamped}\emph{clamped}. If $\bar G$ is clamped and controlled we call it \index{rooted digraph!planted}\index{planted!rooted digraph}\emph{planted}.
\end{definition}
\begin{remark}
Note that by definition there may be an edge $e \in E(G)$ with one incidence in $X_i$ and one incidence in $X_j$ for distinct $i,j$. In particular $(X_1,\ldots,X_\ell)$ may partition $V(G)$. Also, if the graphs we are talking about are clear from context we write~$\pi$ instead of~$\pi_G$, and vice-versa we may add the subscript to remove ambiguity.
\end{remark}

Whenever we write $\bar G$ for a rooted digraph, we implicitly denote by $G$ the underlying directed graph.

We lift the definition of isomorphism to rooted digraphs in a natural way by taking an isomorphism between the underlying digraphs and additionally requiring it to map roots to roots respecting their order.

\begin{definition}
    Let~$\bar{G}=(G,\pi_G, X_1,\ldots,X_\ell)$ and $\bar{H}=(H,\pi_H, A_1,\ldots,A_t)$ be rooted digraphs for $\ell,t \geq 1$. We say that~$\bar{G}$ and~$\bar{H}$ are \emph{isomorphic}, denoted as $\bar{G} \cong \bar{H}$, if and only if they share the same index, i.e., $\ell = t$, and $\delta(X_i) = \delta(A_i)$ for every $1 \leq i \leq t$, and further there exists a bijection~$\xi: G \to H$ such that~$\restr{\xi}{V(G)}: V(G) \to V(H)$ and~$\restr{\xi}{E(G)}: E(G) \to E(H)$ are bijections satisfying the following
    \begin{enumerate}[label=(\arabic*)]
    \item $\xi(e) = (\xi(u), \xi(v))$ for all  $e = (u,v) \in E(G)$, and
     \item for~$\pi(X_i)=(e_1,\ldots,e_{k_i})$ and~$\pi(A_i)=(f_1,\ldots,f_{k_i})$ it holds~$\xi(e_j) = f_j$, i.e.,~$\xi(\pi(X_i)) = \pi(A_i)$ for every $1 \leq j \leq k_i$ and $1 \leq i \leq t$ where $k_i = \delta(X_i)$.
    \end{enumerate}
    We write~$\xi:\bar G \cong \bar H$ and simply $\bar G \cong \bar H$ if such a map exists.
\end{definition}

We continue by lifting the definition of immersions to rooted digraphs as follows. 

\begin{definition}[Immersion of Rooted Digraphs]\label{def:rooted_immersion}
    Let $\ell \geq 1$. Let~$\bar{G} = (G,\pi_G,X_1,\ldots,X_\ell)$ and~$\bar{G'}=(G',\pi_{G'}, X_1',\ldots,X_\ell')$ be rooted digraphs of common index $\ell$. We say that~$\bar{G}'$ \emph{(strongly) immerses}~$\bar G$ if $\delta(X_i) = \delta(X_i')$ for every $1 \leq i \leq \ell$ and there is a map $\gamma: V(G)\cup E(G) \to G'$ satisfying the following:
    \begin{enumerate}[label=(\arabic*)]
        \item $\gamma: G \hookrightarrow^* G'$ is a (strong) immersion, \label{def:rooted_immersion:1}
        \item $\gamma$ maps $X_i$ to $X_i'$ and $\bar X_i$ to $\bar X_i'$, for every $1 \leq i \leq \ell$, and \label{def:rooted_immersion:2}
        \item  for every $1 \leq i \leq \ell$, given the natural bijection~$\eta_i: \rho_G(X_i) \to \rho_{G'}(X_i')$ satisfying~$\eta_i(\pi_G(X_i)) = \pi_{G'}(X_i')$, the path~$\gamma(e)$ for~$e \in \rho_G(X_i)$ contains~$e' \in \rho_{G'}(X_i')$ if and only if~$\eta_i(e) = e'$. We say that \emph{$\gamma$ respects the order of the roots}.\label{def:rooted_immersion:3}
    \end{enumerate}
    We call~$\gamma$ a \emph{(strong) rooted immersion} or simply a (strong) immersion of~$\bar{G}$ in~$\bar{G'}$ and write~$\gamma: \bar{G} \hookrightarrow \bar{G}'$ for strong immersion and $\gamma: \bar{G} \hookrightarrow^* \bar{G}'$ for immersion. If a rooted immersion is not strong, we may additionally call it \emph{weak} to highlight that fact.
\end{definition}
\begin{remark}
    Note that \crefdef{def:rooted_immersion}{2} is not necessarily implied by \crefdef{def:rooted_immersion}{3}, but it will be implied implicitly for most of the following results. We impose it in the definition nonetheless for rigorousity.   
\end{remark}

Next we define what we mean by \emph{rooted cuts}.

\begin{definition}[Rooted Cut]\label{def:rooted_cut}
    Let $\ell \geq 1$ and~$\bar{G}= (G,\pi, X_1,\ldots,X_\ell)$ be a rooted digraph. Let~$Y \subset V(G)$ be a strict subset. Then $Y$ \emph{induces a rooted cut (in $\bar G$)}\index{rooted cut}\index{cut!rooted} if for every $1 \leq i \leq \ell$ either~$X_i \subseteq Y$ or $X_i \cap Y= \emptyset$. 

    Given a rooted cut $Y \subset V(G)$ we define $\insc_G(Y) \coloneqq \{X_i \mid X_i \subseteq Y,\ 1 \leq i \leq \ell\}$ and $\outc_G(Y) \coloneqq \{X_i \mid X_i \cap Y = \emptyset,\ 1 \leq i \leq \ell\}$ omitting the subscript if clear from context. We call a rooted cut \emph{proper}\index{rooted cut!proper}\index{cut!rooted proper} if $\insc(Y) \neq \emptyset$. 
\end{definition}
\begin{remark}
   Note that $\insc(Y)$ and $\outc(Y)$ partition $\{X_1,\ldots,X_\ell\}$ by definition. Also note that we will mainly work with proper rooted cuts.
\end{remark}

We refine the notion of rooted digraphs for Eulerian digraphs as follows.
\begin{definition}[Quasi-Eulerian Digraphs]
    Let $G$ be a digraph such that every vertex in $V(G)$ is Eulerian or of degree one. Then we call $G$ \emph{quasi-Eulerian}\index{quasi-Eulerian}. Let $V \subseteq V(G)$ be the set of degree-one vertices. We say that $G$ is \emph{Eulerian up to $V$}. 

    We define \Symbol{V-G@$V^-(G)$}$V^-(G) \subseteq V$ and \Symbol{V-G@$V^+(G)$}$V^+(G) \subseteq V$ to be the vertices of in-degree one and out-degree one respectively. 
    If $\Abs{V^+(G)}= k\in \N$ say, then we say that $G$ admits \index{defect}\emph{defect $2k$}. 
    Let $X \subseteq V(G)$. We define the \index{balance}\emph{balance of $X$ (in $G$)} via $\operatorname{bal}_G(X) \coloneqq \Abs{\Abs{X \cap V^+(G)} - \Abs{X \cap V^-(G)}}$, omitting the subscript if clear from context. We call $X$ \index{balanced}\emph{balanced} if $\operatorname{bal}(X) = 0$. 
\end{definition}
\begin{remark}
    Note that by a simple counting argument we derive that $\Abs{V^+(G)} = \Abs{V^-(G)}$. 
\end{remark}

It turns out that for quasi-Eulerian digraphs balance is symmetric in the following sense.
\begin{observation}\label{obs:balance_is_symmetrical}
    Let $G$ be quasi-Eulerian of defect $\tau \in 2\N$. Let $X \subset V(G)$ induce a cut in $G$. Then~$\operatorname{bal}(X) = \operatorname{bal}(\bar X)$.
\end{observation}
\begin{proof}
    Since $G$ is quasi-Eulerian we have $\Abs{V^-(G)} = \Abs{V^+(G)}$; let $M \in \operatorname{Match}(V^-(G),V^+(G))$ be a perfect matching. The graph $G^* \coloneqq G+M$ is Eulerian by construction, hence $\delta_{G^*}(X) \in 2\N$. Let $M^* \coloneqq M \cap \rho_{G^*}(X)$. One easily verifies that $\operatorname{bal}_G(X) = \Abs{M^*}$; the claim follows.
\end{proof}
Finally we define rooted Eulerian digraphs as follows.

\begin{definition}[Rooted Eulerian Digraphs]\label{def:rooted_eulerian_digraph}
    Let $(G,\pi, X_1,\ldots,X_\ell)$ be a rooted digraph for some $\ell \in \N$. If $G$ is Eulerian up to $X \subset V(G)$, $X_i$ is balanced for every $1 \leq i \leq \ell$ and $X \subseteq \bigcup_{i=1}^\ell X_i$, then we call $(G,\pi, X_1,\ldots,X_\ell)$ a \emph{rooted Eulerian digraph} and say that $G$ is rooted in $(\pi,X_1,\ldots,X_\ell)$.
\end{definition}
\begin{remark}
    Note that, by definition, rooted ``Eulerian'' digraphs is a misnomer for they are only required to be quasi-Eulerian.
    Since we will almost exclusively work with quasi-Eulerian digraphs, we omitted the word ``quasi'' in the above definition for simplicity. At places where we require rooted Eulerian digraphs to be strictly Eulerian we emphasise this explicitly. Note that if $\bar G$ is controlled, then $G$ is actually Eulerian.
\end{remark}
Intuitively, a rooted Eulerian digraph is Eulerian up to a few degree one vertices that are ``guarded'' by a few balanced sets. The fact that the sets need to be balanced has several reasons: The most intuitive one being, that there is a natural way to identify with $X_i$ a quasi-Eulerian digraph that is ``quasi'' a subgraph of $G$; we elaborate on this later when defining ``torsos'' and ``pieces''.

We derive  the following.

\begin{lemma}\label{lem:rooted_Eulerian_cuts_are_Eulerian}
     Let $(G,\pi, X_1,\ldots,X_\ell)$ be a rooted Eulerian digraph. Let $Y \subset V(G)$ induce a rooted cut. Then $Y$ induces a Eulerian cut.
\end{lemma}
\begin{proof}
    For every $1 \leq i \leq k$ identify the vertices in $X_i$ to a single vertex $x_i^*$. The resulting graph $G^*$ is Eulerian since each partition $X_i$ is balanced whence it induces a balanced cut. Let $Y^* \subset V(G^*)$ be the set of vertices obtained from $Y$ after the above identifications. By \cref{def:rooted_cut} of rooted cut, either $X_i \subseteq Y$ or $X_i \cap Y = \emptyset$, which implies that $\rho(Y^*) = \rho(Y)$. But $G^*$ is Eulerian and thus $Y^*$ induces a Eulerian cut. This concludes the proof.
\end{proof}
\begin{remark}
    In particular, the cuts $\rho(X_i)$ are rooted for every $1 \leq i \leq k$. 
    The above construction of $G^*$ will essentially be the above mentioned ``torso'' of $G$, but we omitted a general definition here, as it will make more sense after introducing more tools.
\end{remark}

Intuitively \cref{lem:rooted_Eulerian_cuts_are_Eulerian} guarantees that rooted Eulerian digraphs ``behave'' like Eulerian digraphs. From here on we will tacitly use that rooted cuts are of even order implicitly applying \cref{lem:rooted_Eulerian_cuts_are_Eulerian} for simplicity.

\subsection{$\Omega$-Knitworks}\label{subsec:knitworks} 
\paragraph{Matchings and Links.} 
We introduce \emph{links} to keep track of ``how'' certain parts of a graph have been ``knitted'' to the graph.

\begin{definition}[Matchings]
  Let~$E_1,E_2$ be two non-empty sets of distinct elements with~$\Abs{E_1} = \Abs{E_2}$. A \emph{matching of $(E_1,E_2)$} is a set~$M\subseteq E_1 \times E_2$ such that for every element~$e_1 \in E_1$ there is at most one element~$m=(e,e') \in M$ with~$e_1 = e$---if it exists we say that $e_1$ is \emph{matched (by $M$)}---and similarly for~$e_2 \in E_2$ there is at most one element~$m=(e,e') \in M$ with~$e_2 = e'$ and if it exists we say that $e_2$ is \emph{matched (by $M$)}.  We call a matching~$M$ \emph{perfect} if every element~$e \in E_1\cup E_2$ is matched by $M$. We denote by~$\operatorname{Match}(E_1,E_2)$ the set of all matchings of~$(E_1,E_2)$.
\end{definition}

\begin{definition}[Links]\label{def:well-linked_links}
  Let $E_1,E_2$ be a set of distinct elements with $\Abs{E_1} = \Abs{E_2}$.  A \emph{link}\index{link} is a set~$\mathfrak{M} \subseteq \operatorname{Match}(E_1,E_2)$ of matchings. We call $\mathfrak{M}$ \emph{reliable}\index{link!reliable}\index{reliable link} if $M' \in \mathfrak{M}$ for every $M \in \mathfrak{M}$ and every $M' \subseteq M$. 

  We call~$\mathfrak{M}$ \emph{linkable}\index{linkable} if for every~$(e_1,e_2) \in E_1 \times E_2$ there exists~$M \in \mathfrak{M}$ with~$(e_1,e_2) \in M$. 

  We call~$\mathfrak{M}$ \emph{well-linked} if for all for all partitions~$E_j= X_j \cup Y_j$ with~$\Abs{X_1} = \Abs{X_2}$ and~$\Abs{Y_1} = \Abs{Y_2}$ for~$j=1,2$ there exists~$M \in \mathfrak{M}$ such that~$M = M_X \cup M_Y$ where~$M_X \in \operatorname{Match}(X_1,X_2)$ and~$M_Y \in \operatorname{Match}(Y_1,Y_2)$ are perfect matchings.

    We call~$\mathfrak{M}$ \emph{inter-linked} if for all pairs~$E_1',E_2'$ with~$E_j' \subseteq E_j$ and~$\Abs{E_1'} = \Abs{E_2'}$  and for all partitions~$E_j'= X_j \cup Y_j$ with~$\Abs{X_1} = \Abs{X_2}$ and~$\Abs{Y_1} = \Abs{Y_2}$ for~$j=1,2$ there exists~$M \in \mathfrak{M}$ such that~$M = M_X \cup M_Y$ where~$M_X \in \operatorname{Match}(X_1,X_2)$ and~$M_Y \in \operatorname{Match}(Y_1,Y_2)$ are perfect matchings.
\end{definition}
\begin{remark}
    Note that by definition an inter-linked link is well-linked, and a well-linked link is linkable.
\end{remark}

The intuition of links is to keep track of ``how'' a (strong) immersion is allowed to interact with parts of the graph that at some point have been ``cut out'' and replaced by a single vertex but will later be ``knitted'' back in again.

The following is the main data structure for the remainder of this paper and will play a crucial role in future work. 
\begin{definition}[$\Omega$-Knitworks]\label{def:knitwork}
    Let~$\Omega = (V(\Omega),\preceq)$ be a well-quasi-order. An \emph{$\Omega$-knitwork} is a tuple~$\GGG=(\bar{G},\mu_G,\m_G,\Phi_G)$ such that
    \begin{enumerate}[label=(\arabic*)]
        \item $\bar{G}$ is a loopless rooted Eulerian digraph; let $X \subseteq V(G)$ be the root set of $\bar G$,\label{def:knitwork:1}
        \item $\mu_G$ is a map with~$\dom(\mu_G) \subseteq V(G)\setminus X$ and~$\mu_G(v)$ is an ordering of~$\rho(v)$,\label{def:knitwork:2}
        \item $\m_G$ is a map with~$\dom(\m_G) \subseteq \dom(\mu_G)$ and~$\m_G(v) \subseteq \operatorname{Match}(\rho^-(v),\rho^+(v))$ is a reliable link, and\label{def:knitwork:3}
        \item $\Phi_G$ is a function with~$\dom(\Phi_G) \subseteq V(G)\setminus X$ and~$\Phi_G(v) \in V(\Omega)$.\label{def:knitwork:4}
    \end{enumerate}
    If for all~$v \in \dom(\m_G)$ the links~$\m_G(v)$ are  well-linked we call~$\GGG$ \emph{well-linked}, and if they are inter-linked we call $\GGG$ \emph{inter-linked}. 
    
   Given the rooted Eulerian digraph~$\bar{G}$ we call~$\GGG$ an~\emph{$\Omega$-knitwork for $\bar G$} and call $\ell$ its \emph{index}. Similarly given $\GGG$ we call $\bar{G}$ its \emph{underlying rooted Eulerian digraph}. If $\bar{G}$ is controlled or planted we call $\GGG$ \emph{controlled} or \emph{planted} respectively.

    Given a rooted Eulerian digraph $\bar G$ we call $(\mu_G,\m_G,\Phi_G)$ an \emph{$\Omega$-sleeve} for $\bar G$ if $\GGG \coloneq (\bar G, \mu_G, \m_G ,\Phi_G)$ is an $\Omega$-knitwork. Finally we may omit (or add) the subscript $G$ writing $\GGG=(\bar G, \mu, \m ,\Phi)$ if clear from context.
\end{definition}
\begin{remark}
    Note that two distinct rooted Eulerian digraphs may share a common $\Omega$-sleeve---whence adding subscripts may help to distinguish them---and a single rooted Eulerian digraph may admit several distinct $\Omega$-sleeves.
\end{remark}
Whenever we write $\GGG$ for an $\Omega$-knitwork, we implicitly mean that the underlying rooted digraph is denoted by $\bar G$. The following is imminent from the definition, and an important observation for future work.

\begin{observation}\label{obs:deg4_well=inter}
    Let $\Omega$ be a well-quasi-order. Let $\GGG$ be a well-linked $\Omega$-knitwork such that $G$ is of maximum degree four. Then $\GGG$ is inter-linked.
\end{observation}

The notions established for rooted Eulerian digraphs naturally lift to~$\Omega$-knitworks by omitting the $\Omega$-sleeve, as exemplified by the following.
\begin{definition}
    Let~$\Omega$ be a well-quasi-order and let $\GGG$ be an $\Omega$-knitwork. Then we say that~$Y \subset V(G)$ induces a \emph{rooted cut} in~$\GGG$ if and only if it induces a rooted cut in~$\bar{G}$.
\end{definition}

Similarly we may lift isomorphisms to knitworks as follows.
\begin{definition}[Isomorphic Knitworks]\label{def:isomorphism_knitwork}
     Let~$\Omega$ be a well-quasi-order and let~$\GGG=(\bar G,\mu,\m,\Phi)$ and~$\HHH=(\bar H,\nu,\n,\Psi)$ be~$\Omega$-knitworks. Then we say that~$\GGG \cong \HHH$ if and only if there is a map~$\xi:\bar G \cong \bar H$ such that for every $x \in V(H)$ it holds~$\Psi(\xi(x)) = \Phi(x)$,~$\nu(\xi(x)) = \xi(\mu(x))$ and~$\n(\xi(x)) = \xi(\m(x))$. We write~$\xi:\GGG \cong \HHH$ or simply $\GGG \cong \HHH$ if such a $\xi$ exists.
\end{definition}

And finally we define what we mean by a \emph{sub-knitwork}.
\begin{definition}[Sub-Knitwork]\label{def:subknitwork}
    Let $\Omega$ be a well-quasi-order. Let $\GGG = (\bar G,\mu
    _\GGG,\m_\GGG,\Phi_\GGG)$ be an $\Omega$-knitwork rooted in $(\pi,X_1,\ldots,X_\ell)$ for some $\ell \geq 1$.
    Let $H \subseteq G$ such that $\bigcup_{i=1}^\ell X_i \subseteq V(H)$ and $\rho_{G}(X_i) = \rho_H(X_i)$ for all $1 \leq i \leq \ell$. We call $\bar H \coloneqq (H,\pi,X_1,\ldots,X_\ell)$ a \emph{rooted sub-digraph of $\bar G$}. Fix $\dom(\mu_\HHH) \coloneqq \dom(\mu_\GGG) \cap V(H)$ as well as $\dom(\m_\HHH) \coloneqq \dom(\m_\GGG) \cap V(H)$ and define the $\Omega$-sleeve as follows.
    \begin{enumerate}[label=(\roman*)]
        \item Let $v \in \dom(\mu_\HHH)$ and $\mu_\GGG(v) = (e_1,\ldots,e_k)$ for some respective $k \in 2\N$. Let $\rho_H(v) = \{e_{i_1},\ldots,e_{i_t}\}$ for some $t\in 2\N$ and $1 \leq i_1 < \ldots < i_t \leq k$. Then we set $\mu_\HHH(v)\coloneqq (e_{i_1},\ldots,e_{i_t})$.
        \item let $v \in \dom(\m_\HHH)$. For every $\MMM \in \m_\GGG(v)$ define $\MMM' \coloneqq \MMM \cap \operatorname{Match}(\rho_H^-(v),\rho_H^+(v))$. Then we define $\m_\HHH(v) \coloneqq \{ \MMM' \mid \MMM \in \m_\GGG(v)\}$, and
        \item define $\Phi_\HHH \coloneqq \restr{\Phi_\GGG}{V(H)}$.
    \end{enumerate} 
    We call $(\mu_\HHH,\m_\HHH,\Phi_\HHH)$ the \emph{$\Omega$-sleeve induced by $H$ (from $(\mu_\GGG,\m_\GGG,\Phi_\GGG)$)} and we call $\HHH$ a \emph{sub-knitwork of $\GGG$} and write $\HHH \subseteq \GGG$.
\end{definition}
\begin{remark}
    One easily verifies that $\HHH$ is indeed an $\Omega$-knitwork, where each of the maps $\mu_\HHH,\m_\HHH$ and $\Phi_\HHH$ are the ``naturally'' induced maps by $H$. Note that $\HHH$ and $\GGG$ are rooted in the same set by definition. 
\end{remark}

The following is a straightforward consequence to \cref{def:subknitwork}.
\begin{lemma}\label{lem:subknitwork_difference_eulerian}
     Let $\Omega$ be a well-quasi-order. Let $\GGG$ be an $\Omega$-knitwork and let $\HHH \subseteq \GGG$. Then $G - H$ is an Eulerian digraph.
\end{lemma}
\begin{proof}
    Let $\GGG$ and $\HHH$ be rooted in $(\pi,X_1,\ldots,X_\ell)$ for some $\ell \geq 1$. Let $G^*$ be obtained from $G$ by contracting $X_i$ to $x_i^*$ for every $1 \leq i \leq \ell$, and let $H^*$be defined from $H$ analogously. Then $G^*$ and $H^*$ are Eulerian and $\rho_{G^*}(x_i^*) = \rho_{H^*}(x_i^*)$ for every $1 \leq i \leq \ell$. Let $C \coloneqq G^* - H^*$ deleting any resulting isolated vertices. Then $C$ is Eulerian with $C \subseteq G - H$ by construction concluding the proof.
\end{proof}

Abstractly speaking, one may think of an~$\Omega$-knitwork~$(\bar{G},\mu,\m,\Phi)$ as a rooted Eulerian digraph together with a set of vertices~$\dom(\mu)$ which are placeholders marking cuts---namely~$\rho(v)$---where we will later ``knit'' another rooted Eulerian digraph, $(H,\pi, Y)$ say, into that cut by replacing~$v$, identifying~$\rho_G(v)$ with~$Y$ such that~$\mu_G(v)$ and~$\pi_H(Y)$ agree, and so that the ``immersion-type'' of~$(H,\pi, Y)$ agrees with the label~$\Phi(v)$. The map~$\m$ is used to keep track of the ``feasible linkages'' in~$H$ with both ends in $Y$, making sure that the immersion adheres to the restrictions posed by respective parts we carved out before (this will mainly be of importance in future work, once we carve out structurally rich pieces, i.e., of large carving width).

Matching the intuition behind $\m$ given above we define the following; recall that for a loopless digraph $G$, \cref{obs:cut_at_vertex} implies $E(v) = \rho(v)$ for every $v \in V(G)$.

\begin{definition}[Matchings of~$\LLL$ and~$\gamma$]\label{def:matching_of_gamma}
    Let~$G$ be a loopless digraph and~$\LLL$ a linkage in~$G$. Let~$v \in V(G)$. We define~\Symbol{MLL@$M_\LLL(v)$}$M_\LLL(v) \coloneqq \{(e,e') \mid e,e' \in \rho(v), \ \head(e) = v = \tail(e'), \text{ and } (e,e') \subseteq L \text{ for some } L \in \LLL \}$.
    
    Let~$G'$ be another loopless digraph and~$\gamma: E(G) \cup V(G) \to G'$ be an immersion. Let~$v' \in V(G')$. 
    We define~\Symbol{MGAMMAV@$M_\gamma(v')$}$M_\gamma(v') := \{ (e, e') \in M_\gamma(v') \sth $ there is~$e^* \in E(G)$ with $(e,e') \subseteq \gamma(e^*) \}    
    \in \operatorname{Match}(\rho^-(v'),\rho^+(v'))$.

\end{definition}
\begin{remark}
Note that if some path in $\LLL$ starts in an edge $e \in \rho(v)$, say, then $(e,e') \not\in M_\LLL(v)$ for every $e' \in E(G)$. Furthermore,~$\gamma$ is strong if and only if $M_\gamma(v') = \emptyset$ for every~$v' \in \gamma(V(G))$.
\end{remark}

Finally, the following definition of \emph{immersion of~$\Omega$-knitworks} may strengthen the intuition discussed above.

\begin{definition}[$\Omega$-knitwork immersion]\label{def:knitwork_immersion}
    Let~$\Omega=(V(\Omega),\preceq)$ be a well-quasi-order.
    
    Let $(\bar{H},\nu,\n,\Psi)$ and $(\bar{G},\mu,\m,\Phi)$ be~$\Omega$-knitworks.
    We say that~$(\bar{G},\mu,\m,\Phi)$ \emph{(strongly) immerses}\index{knitwork!immersion}\index{immersion!knitwork} $(\bar{H},\nu,\n,\Psi)$ if and only if there exists a sub-knitwork $(\bar{G}',\mu',\m',\Phi') \subseteq (\bar G,\mu,\m,\Phi)$ and a  map~$\gamma: V(H) \cup E(H) \to G'$ satisfying the following
    \begin{enumerate}[label=(\arabic*)]
        \item  \label{def:knitwork_immersion:1} $\gamma$ (strongly) immerses~$\bar{H}$ into~$\bar{G}'$; in particular~$\gamma$ respects the order of the roots whence $\bar H$ and $ \bar G'$, and thus $\bar H$ and $\bar G$, have some common index $\ell \in \N$,
        \item\label{def:knitwork_immersion:2}  for each~$v \in V(H)$,~$\gamma(v) \in \dom(\mu')$ if and only if~$v \in \dom(\nu)$, and~$\gamma(v) \in \dom(\m')$ if and only if~$v \in \dom(\n)$,
        \item\label{def:knitwork_immersion:3} for each~$v \in \dom(\nu)$ the index $k$ of~$\nu(v)=(e_1,\ldots,e_k)$ and $k'$ of ~$\mu'(\gamma(v))=(e_1',\ldots,e_k')$ agree and are even and for~$1 \leq i \leq k$ the edge~$e_i'$ is contained in~$\gamma(e_i)$; we say that $\gamma$ \emph{respects~$\nu$ and~$\mu'$ (or $\mu$)},
        \item\label{def:knitwork_immersion:4} for each~$v' \in \dom(\m') \setminus \gamma(\dom(\n))$,~$M_\gamma(v') \in \m'(v')$, 
        \item\label{def:knitwork_immersion:5} for each~$v \in \dom(\n)$ there is an injection $\alpha: \n(v) \to \m'(\gamma(v))$ such that 
        if $(e_1,\ldots,e_k)=\nu(v)$ and~$(e_1',\ldots,e_k')=\mu'(\gamma(v))$ then, for every $M\in \n(v)$, $(e_i,e_j) \in M$ if and only if $(e_i',e_j') \in \alpha(M)$, and
        \item\label{def:knitwork_immersion:6} for every~$v \in \dom(\Psi)$ we have~$\Psi(v) \preceq \Phi'(\gamma(v))$, in particular~$\gamma(\dom(\Psi)) \subseteq \dom(\Phi')$.
    \end{enumerate}
    We call the map~$\gamma$ a \emph{(strong) $\Omega$-knitwork immersion} or simply \emph{(strong) immersion} if clear from the context, and write~$\gamma:(\bar{H},\nu,\n,\Psi) \hookrightarrow (\bar{G},\mu,\m,\Phi)$\Symbol{HOOKRIGHT@$\hookrightarrow$} for strong immersion and $\gamma:(\bar{H},\nu,\n,\Psi) \hookrightarrow^* (\bar{G},\mu,\m,\Phi)$\Symbol{HOOKRIGHTSTAR@$\hookrightarrow^*$} for immersion, and we call $\GGG'$ a \emph{witness (for $\gamma$)}. We call $\gamma$ \emph{stable} if $\GGG' = \GGG$ is a witness.
    
    If the $\Omega$-knitwork immersion is not strong we may call it \emph{weak} to highlight that fact. 
\end{definition}
\begin{remark}
    It is a natural idea to, instead of requiring the existence of a sub-knitwork, simply loosen the restriction of Condition (3) and only require the index of~$\mu(\gamma(v))$ to be at least as large as the index of~$\nu(v)$ (and adapt Condition (4) accordingly). Although one easily verifies that for 
    \emph{strong} $\Omega$-knitwork immersion this results in an equivalent definition, for \emph{weak} $\Omega$-knitwork immersion, however, this results in a strictly weaker notion. This is is why settle for the above definition. Note that in the setting of bounded degrees both notions turn out to be equally powerful. However, when working with weak $\Omega$-knitwork immersion one may have to loosen $(3)$ to accommodate for ``unbounded" degrees, where immersion allows to route through vertices of the immersion model. We will discuss one work-around in this exposition in \cref{sec:bded_treewidth}, and note that to prove \cref{conj:wqo_gen} without labelled vertices, it seems up until now that the stronger variant of weak immersion as defined above is enough.
    
    Condition $(4)$ may intuitively be thought of as a check that, whenever $\gamma$ routes through a vertex $v' \in \dom(\m')$, then $\gamma$ is indeed \emph{realisable} in $v'$, where $v'$ is to be thought of as a placeholder vertex encoding the feasible linkages through the cut $\rho(v')$. Note that, when dealing with weak immersion, one may be tempted to extend $(4)$ to $\gamma(\dom(\n))$ which is not needed due to our choice of definition, i.e.,  $\dom(\n) \subseteq \dom(\nu)$ and, by $(3)$, every vertex in $\gamma(\dom(\n))$ is already ``fully used up'', i.e., no edge adjacent to a vertex in $\gamma(\dom(\n))$ can be used to route further paths. Compare this to the previous discussion on loosening $(4)$.
    
    Finally, condition $(5)$ may intuitively be thought of as a guarantee that the placeholder vertices and their respective cuts are mapped properly.
\end{remark}

    Note that \crefdef{def:knitwork_immersion}{4} would need a more cumbersome phrasing if we were to allow loops in $G$, since then a path could use a loop edge in $G'$ where one single edge would technically need to be viewed as $2$ different incidences with the same vertex (using the incidence-model of digraphs this would not be very hard to track). While this is quite natural when dealing with immersions, since there are easier ways to deal with loops, we decided to omit it in the definition for simplicity.

The following is immediate by \cref{def:subknitwork,def:knitwork_immersion}.
\begin{observation}\label{obs:subknitworks_immerse}
    Let $\Omega$ be a well-quasi-order. Let $\GGG$ be an $\Omega$-knitwork and let $\HHH \subseteq \GGG$, then $\HHH \hookrightarrow \GGG$.
\end{observation}
For future work it will be helpful that $\gamma:\HHH \hookrightarrow \GGG$ is not required to be \emph{stable}, in particular the immersion may be witnessed by a strict sub-knitwork. However, all of the results presented in this paper hold in the restricted case of stable $\Omega$-knitwork immersion as we discuss later.

In light of \crefdef{def:knitwork_immersion}{4}, we define the following.
\begin{definition}
    Let $\Omega$ be a well-quasi-order and $\GGG=(\bar G, \mu, \m, \Phi)$ be an $\Omega$-knitwork. Let $\LLL$ be a linkage in $G$. We say that $\LLL$ is \emph{$\m$-respecting}\index{mrespec@$\m$-respecting} if for every $v \in \dom(\m)$ it holds $M_{\LLL}(v) \in \m(v)$.
\end{definition}
\begin{remark}
    Thus, given an $\Omega$-knitwork immersion $\gamma: \GGG \hookrightarrow^* \GGG'$, the linkage $\LLL = \gamma(E(G))$ is $\m_{G'}$-respecting by definition.
\end{remark}

In the same spirit we lift the definition of immersion model as follows.
\begin{definition}[Immersion model]
    Let $\Omega$ be a well-quasi-order and $\HHH=((H,\pi_H,A_1,\ldots,A_\ell),\nu,\n,\Psi)$ and $\GGG=(\bar G,\pi_G,X_1,\ldots,X_\ell),\mu,\m,\Phi)$ be $\Omega$-knitworks of some common index $\ell \in 2\N$. Let $V_H \subseteq V(G)$ with $\bigcup_{i=1}^\ell X_i \subseteq V_H$ and let $\LLL$ be a (strong) $\m_G$-respecting linkage in $G$ of order $\Abs{E(H)}$ with all its endpoints in $V_H$. If there is a map $\gamma$ with $\dom(\gamma) = V(H) \cup E(H)$ such that 
    \begin{enumerate}[label=(\roman*)]
        \item $\gamma:V(H) \to V_H$ and $\gamma:E(H) \to \LLL$ are bijections, such that $\gamma:\bar H \hookrightarrow^* \bar G$ (is strong), in particular $(V_H,\LLL)$ is a (strong) immersion model of $H$ in $G$,
        \item for each $v \in \dom(\nu)$ with $\nu(v) = (f_1,\ldots,f_t)$ and $\mu(v) = (e_1,\ldots,e_k)$ it holds $k \geq t$ and there are $1 \leq i_1 < \ldots < i_t \leq k$ such that $e_{i_j} \in \gamma(f_j)$ for all $1 \leq j \leq t$,
        \item $\LLL$ is $\m_\GGG$-respecting,
        \item $\gamma$ satisfies \crefdef{def:knitwork_immersion}{5}, and
        \item for every $v \in \dom(\Psi)$ we have $\Psi(v) \preceq \Phi(\gamma(v))$.
    \end{enumerate}
   Then we call $(V_H,\LLL)$ a \emph{(strong) immersion model of $\HHH$ (in $\GGG$)}. \index{immersion!model}\index{model!immersion}
\end{definition}
The following is immediate.
\begin{observation}
    Let $\Omega$ be a well-quasi-order and~$\HHH$ and $\GGG$ be $\Omega$-knitworks. Then $\GGG$ (strongly) immerses $\HHH$ if and only if there is a (strong) immersion model of $\HHH$ in $\GGG$.
\end{observation}

We briefly take the time to discuss the \cref{def:knitwork_immersion} of $\Omega$-knitwork immersion. While for standard immersion of digraphs it is natural to ask for the immersion $\gamma$ to map edges to linear paths with two distinct endpoints---that is``paths'' in the standard literature---since both concepts are equivalent (see \cref{obs:linkage_gives_linear_linkage}), this is \emph{not} true for $\Omega$-knitworks, as here the matchings imposed by the maps $\m_G$ have to be preserved. In fact, requiring $\gamma$ to map edges to linear paths is a \emph{strictly} stronger assumption, even when restricting to well-linked $\Omega$-knitworks as we elaborate on now. 

First note that for \emph{inter-linked} $\Omega$-knitworks $\GGG$, strong linear linkages in $G$ can be transformed into $\m_G$-respecting strong linear linkages in $\GGG$ on the same set of ends as follows\footnote{We will make use of a weaker form of this result in future work, but the result as stated may be of independent interest.}.
\begin{lemma}\label{lem:from_linear_to_respecting_linkage_interlinked}
    Let $\Omega$ be a well-quasi-order and $\GGG$ an inter-linked $\Omega$-knitwork rooted in $(\pi,X_1,\ldots,X_\ell)$ for some $\ell \geq 1$. Let $E_1,E_2 \subseteq E(G)$ and let $\LLL$ be a strong linear $\{E_1,E_2\}$-linkage in $G$ such that $V^\circ(L)\cap V(X_i)=\emptyset$ for all $1\leq i \leq \ell$. Then there is an $\m_G$-respecting strong linear $\{E_1,E_2\}$-linkage $\LLL'$ in $G$ that satisfies $\Abs{\LLL'} = \Abs{\LLL}$ and $E(\LLL') \subseteq E(\LLL)$ and the set of ends of $\LLL$ and $\LLL'$ agree.
\end{lemma}
\begin{proof}
    If for every $L \in \LLL$ we have $V^\circ(L) = \emptyset$, then it is $\m_G$-respecting by definition, and there is noting to show. Thus we may assume that $E_1\neq E_2$.

    Towards a contradiction we assume the lemma to be wrong and let~$\LLL$ be a strong linear $\{E_1,E_2\}$-linkage refuting the lemma, chosen such that 
    \begin{enumerate}[label=(\roman*)]
        \item $\Abs{E(\LLL)}$ is minimum, and
        \item with respect to $(i)$, the set of vertices~$F \coloneqq \{v \in \dom(\m_G) \mid M_{\LLL}(v) \not \in \m_G(v)\}$ has minimal cardinality.
    \end{enumerate}
    By assumption we have $\Abs{F} \geq 1$. Let~$v \in F$. Note that $v \notin X_i$ for $1\leq i \leq \ell$ by assumption of the lemma, and further $v$ is not an endpoint of a path in $\LLL$ since $\LLL$ is strong. Then, since~$\LLL$ is a linkage, every~$(f,f') \in M_{\LLL}(v)$ is a subpath of exactly one path in~$\LLL$, and since the linkage is linear, every $P \in \LLL$ contains at most one tuple $(f,f') \in M_{\LLL}(v)$ as a subpath. Let $\LLL=\{P_1,\ldots,P_t\}$ for $t = \Abs{\LLL}$ and without loss of generality let~$\{(f_j,f_j') \mid 1 \leq j \leq \tau\} = M_{\LLL}(v)$ for some~$\tau \leq t$ and assume that~$(f_j,f_j') \subseteq P_j$ for every~$1 \leq j \leq \tau$. Let~$E^-\coloneqq \{f_j \mid 1 \leq j \leq \tau\}$ and~$E^+ \coloneqq \{f_j' \mid 1\leq j \leq \tau\}$. Then
     \begin{equation}\label{eq:well-linked_uno}
         E^-\subseteq \rho^-(v),\ E^+ \subseteq \rho^+(v), \text{ and } \Abs{E^-} = \Abs{E^+}.
     \end{equation} 
     For every $1 \leq j \leq \tau$ let~$P_j^1,P_j^2 \subset P_j$ be the unique edge-disjoint subpaths satisfying~$E(P_j) = E(P_j^1) \cup E(P_j^2)$ such that~$P_j^1$ ends in~$f_j$ and~$P_j^2$ starts in~$f_j'$; in particular~$V(P_j^1) \cap V(P_j^2) = \{v\}$ and $P_j = P_j^1 \circ P_j^2$. Let~$\LLL^1 \coloneqq  \{P_j^q \mid P_j^q \text{ has an end in } E_1, \text{ for } 1 \leq j \leq \tau,\text{ and } q\in\{1,2\}\}$ and define~$\LLL^2$ analogously using $E_2$. Clearly~$\LLL^1 \cup \LLL^2 = \{P_j^q \mid 1 \leq j \leq \tau \text{ and } q \in \{1,2\}\}$, as well as  $\LLL^1 \cap \LLL^2 = \emptyset$ by construction. 

     Define~$E_1^- \subset E^-$ via~$e \in E_1^-$ if and only if~$e \in E^-$ and~$e$ is an end of a path in~$\LLL^1$. Analogously define~$E_2^-\subset E^-$ via~$e \in E_2^-$ if and only if~$e \in E^-$ and~$e$ is an end of a path in~$\LLL_2$. Define~$E_1^+$ and~$E_2^+$ analogously using~$E^+$. Note that the linear linkage~$\LLL$ witnesses
     \begin{equation}\label{eq:well-linked_dos}
         \Abs{E_1^- } = \Abs{ E_2^+} \text{ and } \Abs{E_1^+ } = \Abs{ E_2^-}.
     \end{equation}

     Finally, by \cref{eq:well-linked_uno} and \cref{eq:well-linked_dos} together with the \cref{def:well-linked_links} of inter-linked links and the fact that $\GGG$ is an inter-linked $\Omega$-knitwork, there exists~$M\in \m_i(v)$ such that~$M= M_1 \cup M_2$ where~$M_1 \in \operatorname{Match}(E_1^-,E_2^+)$ and $M_2 \in\operatorname{Match}(E_2^-,E_1^+)$ are perfect matchings. Let~$(f_p,f_q') \in M_1$ for respective~$1 \leq p,q \leq \tau$, then~$P_p^1 \circ P_q^2$ is a path---a priori not necessarily linear---with first edge in~$E_1$ and last edge in~$E_2$. Similarly for~$(f_p,f_q') \in M_2$ for respective~$1 \leq p,q \leq \tau$ we derive that~$P_p^1 \circ P_q^2$ is a path with first edge in~$E_2$ and last edge in~$E_1$. It is easily verified that the collection of paths~$\PPP \coloneqq \{P_p^1 \circ P_q^2 \mid (f_p,f_q') \in M_1\cup M_2 \text{ for } 1\leq p,q \leq \tau\}$ is a strong~$\{E_1,E_2\}$-linkage by construction, and finally~$\LLL'=\PPP \cup\{P_p \mid \tau+1 \leq p \leq t\}$ is a strong~$\{E_1,E_2\}$-linkage of order $t$ with $E(\LLL') = E(\LLL)$, such that the set of ends of $\LLL'$ and $\LLL$ agree, and further~$M_{\LLL'}(v) \in \m_G(v)$ by construction.

     \begin{claim}
         $\LLL'$ is linear.
     \end{claim}
     \begin{claimproof}
         To see this note that we established $E(\LLL')=E(\LLL)$, both are strong $\{E_1,E_2\}$-linkages, both have the same order $t$ and their sets of ends agree. If $\LLL'$ were not linear, then this implies the existence of some path $L \in \LLL$ such that $L$ admits a subpath $P \subset L$ starting and ending in the same vertex $x \in V(G)$. In particular, $L = P_p^1 \circ P_q^2 \in \PPP$ for some respective $1\leq p,q \leq \tau$ as we did not alter any other paths. Let now $\LLL^*$ be the linear $\{E_1,E_2\}$-linkage of order $t$ obtained from $\LLL$ via \cref{obs:linkage_gives_linear_linkage}; essentially rerouting paths at cycles, e.g., rerouting $L$ at $x$ omitting $P$. Note that $\LLL^*$ is still strong and $\tau(\LLL^*) = \tau(\LLL')$, in particular their sets of ends agree. By construction, $\LLL^*$ has strictly less edges than $\LLL'$ and hence less edges than $\LLL$. Thus, by our choice of minimality for $\LLL$ with respect to $(i)$, we derive that there is an $m_\GGG$-respecting strong linear $\{E_1,E_2\}$-linkage $\LLL^{\star\star}$ in $G$ that satisfies $\Abs{\LLL^{\star\star}}=\Abs{\LLL^\star}$, and hence $\Abs{\LLL^{\star\star}} = \Abs{\LLL}$, as well as $E(\LLL^{\star\star}) = E(\LLL^\star)$ hence $E(\LLL^{\star\star}) \subseteq E(\LLL)$, and it is on the same set of ends as $\LLL^\star$ and hence $\LLL$. In summary, $\LLL^{\star\star}$ contradicts our choice of $\LLL$ refuting the lemma; by contraposition, $\LLL'$ is linear
      \end{claimproof}
     
     Since~$M_{\LLL'}(w) = M_{\LLL}(w)$ for all~$w \neq v$ with~$w \in \dom(\m_G)$ by construction,~$\LLL'$ refutes the choice of~$\LLL$ respecting $(ii)$; a contradiction.
\end{proof}
\begin{remark}
    Note here that it is not necessarily guaranteed that $\tau(\LLL) = \tau(\LLL')$.
\end{remark}

For well-linked $\Omega$-knitworks one can only guarantee a much weaker version of \cref{lem:from_linear_to_respecting_linkage_interlinked}, in particular we cannot guarantee resulting linkages to be linear, which in turn is the main reason why $\Omega$-knitwork immersion is not defined in terms of linear paths. Note that this highlights once more that \cref{lem:from_linear_to_respecting_linkage_interlinked} may be of independent interest, as it guarantees stronger structure, without allowing for ``full power'' regarding the routing handler $\m_\GGG$, i.e., letting $\m_\GGG(v) = \operatorname{Match}(\rho^-(v),\rho^+(v))$\footnote{Such $\Omega$-knitworks will play a crucial role and will be referred to as \emph{free} in future work.}. The proof follows the same proof strategy as for \cref{lem:from_linear_to_respecting_linkage_interlinked}.

\begin{figure}
    \centering
    \includegraphics[width=0.8\linewidth]{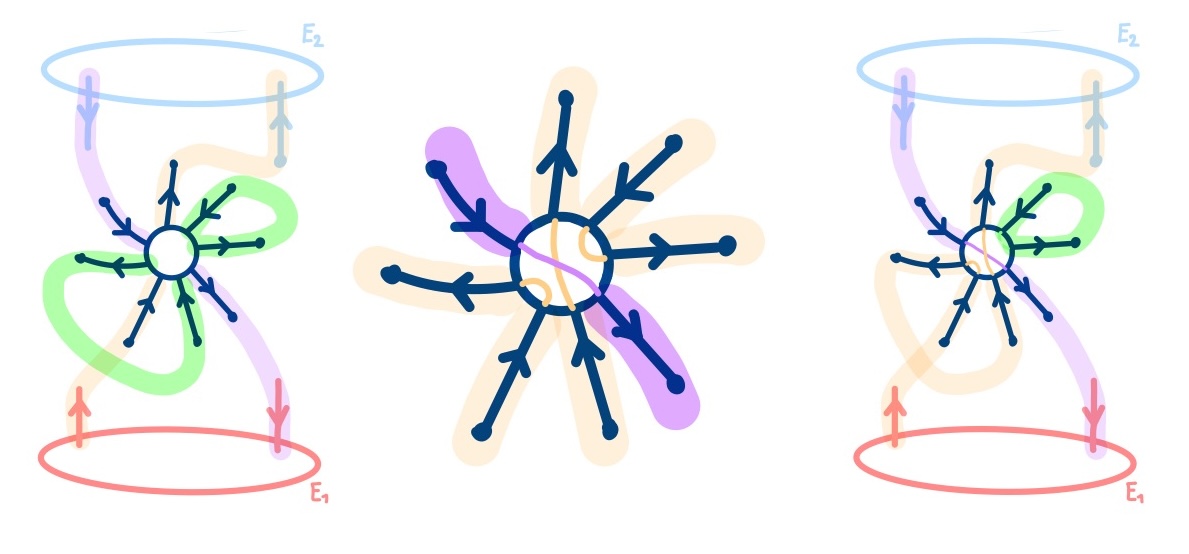}
    \caption{A schematic representation of how to reroute Euler-covers at well-linked vertices. The left figure highlights a Euler-cover, the middle figure highlights a possible link and the right figure highlights a rerouting of the Euler-cover that respects the link.}
    \label{fig:rerouting_well_linked}
\end{figure}

\begin{lemma}\label{lem:from_linear_to_respecting_linkage_well-linked}
    Let $\Omega$ be a well-quasi-order and $\GGG$ a well-linked $\Omega$-knitwork rooted in $(\pi,X_1,\ldots,X_\ell)$ for some $\ell \geq 1$. Let $X,Y \subset V(G)$ be non-empty with $X \cap Y = \emptyset$ such that they both induce rooted cuts and such that $\bar X \cap \bar Y$ is disjoint from $X_i$ for every $1\leq i \leq \ell$. Let $E_1 = \rho_G(X)$ and $E_2 = \rho_G(Y)$ such that~$k \coloneqq \Abs{E_1} = \Abs{E_2}$. If $\delta(X,Y) \geq k$, then there is an $\m_G$-respecting strong $\{E_1,E_2\}$-linkage $\LLL$ in $G$ of order~$k$.
\end{lemma}
\begin{proof}
Note that if $\bar X \cap \bar Y = \emptyset$, then there is nothing to show for then $\{X,Y\}$ partitions $V(G)$ and $\rho(X) = \rho(Y)$ is said linkage. Thus assume it is non-empty.

    By \cref{lem:Menger_for_Euler} there is a linear $\{E_1,E_2\}$-linkage $\LLL$ in $G$ (possibly not $\m_G$ -respecting) such that for every $L \in \LLL$ it holds $V^\circ(L) \cap (X \cup Y) = \emptyset$. 
    Let $Z \in \{X,Y\}$. Since $Z$ induces a rooted cut, we derive that for every $1 \leq i \leq \ell$ either $X_i \subset Z$ or $X_i \subset \bar Z$. Since further $U\coloneqq \bar X \cap \bar Y$ is disjoint from $X_i$ for every $1 \leq i \leq \ell$ we derive that $X_i \cap U = \emptyset$.  Let $G^*$ be obtained from $G$ by contracting $X$ into $x^*$ and $Y$ into $y^*$, then $G^*$ is Eulerian since $X$ and $Y$ induce Eulerian cuts by \cref{lem:rooted_Eulerian_cuts_are_Eulerian} and furthermore every vertex in~$U$ is Eulerian. Then $\LLL$ is a linear linkage in $G^*$ on the same set of edges. Thus, it suffices to prove the lemma in the setting of $G^*$.

    For the proof we extend the definition of $M_\LLL(v)$ as follows. Let $\KKK = \LLL \cup \CCC$ where $\LLL$ is a strong linkage and $\CCC$ is a set of edge-disjoint cycles such that for every $L \in \LLL$ and $C \in \CCC$ it holds $E(L) \cap E(C) = \emptyset$; we say that $\KKK$ is \emph{edge-disjoint} and call it a \emph{Euler-cover} and write $\KKK= \LLL \cup \CCC$ to mean the partition into linkage and cycles. Then we define $M_\KKK(v)$ for $v \in V(G^*)$ via $M_\LLL(v) \cup M_\CCC(v)$ where $$M_\CCC(v) \coloneqq \{(e,e') \mid e,e' \in \rho(v)\ \head(e) = v = \tail(e'), \text{ and } (e,e') \subseteq C \text{ for some } C \in \LLL \}.$$

    Note that $M_\KKK(x^*) = \emptyset$ as well as $M_\KKK(y^*) = \emptyset$ since $\LLL$ is strong and $x^*,y^*$ are not part of any cycle in $\CCC$.
\begin{claim}\label{lem:from_linear_to_respecting_linkage_well-linked_claim1}
        There is a set of edge-disjoint cycles $\CCC$ such that $\KKK=\LLL \cup \CCC$ is a Euler-cover of $G^*$.
    \end{claim} 
    \begin{claimproof}
        Let $H \coloneqq G^* - \bigcup_{L \in \LLL}E(L)$, then $H$ is again Eulerian (possibly disconnected). Let $\CCC$ be a cycle cover of $H$ of minimal order which exists by \cref{obs:covering_eulerian_digraphs}, then $\KKK = \LLL \cup \CCC$ does the trick
    \end{claimproof}

    Towards a contradiction assume the lemma to be wrong and choose a Euler-cover~$\KKK = \LLL^* \cup \CCC$ such that $\LLL^*$ is a strong $\{E_1,E_2\}$-linkage refuting the lemma and such that the set of vertices~$F \coloneqq \{v \in \dom(\m_G) \mid M_{\KKK}(v) \not \in \m_G(v)\}$ has minimal cardinality. By assumption we derive that~$\Abs{F} \geq 1$. 
    
   The proof now follows almost verbatim to that of \cref{lem:from_linear_to_respecting_linkage_interlinked}; see \cref{fig:rerouting_well_linked} for a schematic representation of the rerouting step. Let $v \in F$. Let $\LLL_\CCC(v)$ be the set consisting for each $C \in \CCC$ of the maximum length subpath $L_C(u)$ of $C$ that starts and ends in the same vertex $u \in V(C)$ different than $v$ (this exists since the graphs are loopless); note that $E(L_C(u)) = E(C)$, recalling that paths may start and end in the same vertex according to our \cref{def:paths}. Further, by our \cref{def:paths} of paths and the fact that $\CCC$ is edge-disjoint, $\LLL_\CCC(v)$ is a linkage. Define $\LLL\coloneqq\LLL^*\cup \LLL_\CCC(v)$ and assume that $\LLL = \{P_1,\ldots,P_k\}$ for some respective $k\geq 1$. For each $1 \leq i \leq k$ let $M_{P_i}(v)=\{(e_1^i,f_1^i),\ldots,(e_{\ell_i}^i,f_{\ell_i}^i)\}$ for some $\ell_i^i \geq 0$ (possibly empty). Without loss of generality assume that every path visits $v$ (the other case is analogous; see \cref{lem:from_linear_to_respecting_linkage_interlinked} for details).
   
    For every $1 \leq j \leq k$ let~$P_j^{\text{in}},P_j^{\text{out}} \subset P_j$ be the unique edge-disjoint subpaths satisfying that $P_j^{\text{in}}$ starts in the first edge of $P_j$ and  ends in~$e_1^j$ and~$P_j^{\text{out}}$ starts in~$f_{\ell_j}^j$ and ends in the last edge of $P_j$; in particular~$V^\circ(P_j^\text{in}) \cap V^\circ(P_j^{\text{out}}) = \{v\}$. For every $1 \leq j \leq k$ such that $\ell_j > 1$, and for every $1 \leq i \leq \ell_{j}-1$, let $P_j^i$ be the unique subpath of $P_j$ starting in $f_i^j$ and ending in $e_{i+1}^j$. Then, by construction, $P_j \coloneqq P_j^\text{in} \circ P_j^1 \circ \ldots \circ P_j^{\ell_j-1}\circ P_j^{\text{out}}$.

     Define~$E_1^- \subset \rho^-(v)$ via~$e \in E_1^-$ if and only if~$e \in E(P_j)$ for some $1 \leq j \leq k$ that starts with an edge in $E_1$, or $e \in E(C)$ for some $C \in \CCC$. Define $E_2^+ \subset \rho^+(v)$ via $e \in E_2^+$ if and only if $e \in \rho^+(v)$ and $e \in E(P_j)$ for some $1 \leq j \leq k$ that ends in an edge in $E_2$, or $e \in E(C)$ for some $C \in \CCC$. Define~$E_2^-\subset \rho^-(v)$ via~$e \in E_2^-$ if and only if~$e \in E(P_j)$ for some $1 \leq j \leq k$  and~$E_1^+$ via $e \in E_i^+$ if and only if $e \in \rho^+(v)$ and $e \in E(P_j)$ for some $1 \leq j \leq k$ that ends in an edge in $E_i$. Note that the linkage~$\LLL$ witnesses
     \begin{equation}\label{lem:from_linear_to_respecting_linkage_well-linked_claim2}
         \Abs{E_1^- } = \Abs{ E_2^+} \text{ and } \Abs{E_1^+ } = \Abs{ E_2^-}.
     \end{equation} 
which is clear for $\LLL^*$ as it is a strong $\{E_1,E_2\}$-linkage, and it follows for $\LLL_\CCC(v)$ from the fact that $L_C(v)$ has an equal number of in and out-edges at $v$ for every $C \in \CCC$.
    
    Further more since $E(\LLL) = E(G^*)$ we derive

     \begin{equation}\label{lem:from_linear_to_respecting_linkage_well-linked_claim3}
        E_1^- \cup E_2^- = \rho^-(v) \text{ and }  E_1^+ \cup E_2^+ = \rho^+(v).
     \end{equation}

     Finally, combining \cref{lem:from_linear_to_respecting_linkage_well-linked_claim1}, \cref{lem:from_linear_to_respecting_linkage_well-linked_claim2} and \cref{lem:from_linear_to_respecting_linkage_well-linked_claim3} together with the \cref{def:well-linked_links} of well-linked links and the fact that $\GGG$ is a well-linked $\Omega$-knitwork, there exists~$M\in \m_i(v)$ such that~$M= M_1 \cup M_2$ where~$M_1 \in \operatorname{Match}(E_1^-,E_2^+)$ and $M_2 \in\operatorname{Match}(E_2^-,E_1^+)$ are perfect matchings. 
     By concatenating the subpaths $\{P_j^{\text{in}},P_j^i,P_j^{\text{out}} \mid 1 \leq j \leq k\text{ and } 1 \leq i <\ell_j-1\}$ with respect to the matching $M$ analogously as in \cref{lem:from_linear_to_respecting_linkage_interlinked} we get an Euler-Cover $\KKK' \coloneqq \LLL' \cup \CCC'$ satisfying $E(\KKK') = E(G^*)$ where $\tau(\LLL') = \tau(\LLL^*)$ (and $\CCC'$ possibly empty), as well as~$M_{\KKK'}(w) = M_{\KKK}(w)$ for all~$w \neq v$ with~$w \in \dom(\m_G)$ by construction and $M_{\KKK'}(v)\in \m_\GGG(v)$ by construction. One easily verifies that $\LLL'$ is a strong linkage, as we never interfered with incidences of its endpoints $x^*$ and $y^*$. Thus~$\KKK'$ refutes the choice of~$\LLL^*$ minimising $\Abs{F}$; a contradiction.
\end{proof}
\begin{remark}
    This highlights, that we cannot assure linkages to be linear when working with well-linked $\Omega$-knitworks. In particular, we may not be able to reroute linear linkages to $\m_G$-respecting linear linkages of ``similar type'', as we may need the cycles $\CCC$ in the Euler-cover to reroute in order to stay $\m_G$-respecting. In particular there is no analogue to \cref{obs:linkage_gives_linear_linkage}, and loosening the \cref{def:immersion} to non-linear paths is crucial. 
\end{remark}

Complementing the above remark, we have the following.

\begin{lemma}\label{lem:knitwork_immersion_weak_vs_strong}
\begin{enumerate}
    \item Let $\Omega$ be a well-quasi-order. There exists a well-linked $\Omega$-knitwork $\GGG$ with a strong $\m_G$-respecting linkage $\LLL$ of order $2$ such that there exists no $\m_G$-respecting linear linkage $\LLL'$ of equal order satisfying $\tau(\LLL) = \tau(\LLL')$.\label{lem:knitwork_immersion_weak_vs_strong:1}
    \item There exists a well-quasi-order $\Omega$ on two elements such that the following holds. There are $\Omega$-knitworks $\HHH$ and $\GGG$ such that $\HHH \hookrightarrow^* \GGG$, but given $\LLL = \gamma(E(H))$, there is no linear linkage $\LLL'$ of equal order in $\GGG$ satisfying $\tau(\LLL') = \tau(\LLL)$.\label{lem:knitwork_immersion_weak_vs_strong:2}
\end{enumerate}
\end{lemma}
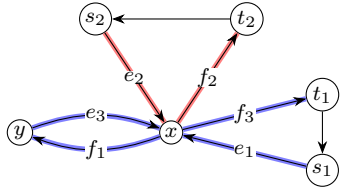
\begin{figure}
    \centering
    \begin{tikzpicture}[>=Stealth, every node/.style={font={\footnotesize}}]
        \tikzstyle{v}=[circle, draw=black, inner sep=1pt]
        \tikzstyle{ne}=[fill=white,inner sep=0pt]
        \tikzstyle{l}=[blue!50!white, line width=2pt]
        \foreach \n/\l/\x/\y in {x/x/0/0,t1/t_1/2/0.5,s1/s_1/2/-0.5,t2/t_2/1/1.5,s2/s_2/-1/1.5,y/y/-2/0}
        \node[v] (\n) at (\x,\y) {$\l$};

        \draw[l] (s1) to (x) (x) to [bend left=20] (y) (y) to [bend left=20] (x) (x) to (t1) ;
        
        \draw[l, red!50!white] (s2) to (x) to (t2) ;
        \foreach \s/\t/\n in {s2/x/e_2,x/t2/f_2,x/t1/f_3,s1/x/e_1}
        \draw[->,black] (\s) to node[ne] {$\n$}  (\t);
        \draw[->] (t2) to (s2) ; 
        \draw[->](t1) to (s1) ; 
        \draw[->](x) to [bend left=20] node[ne] {$f_1$ } (y) ;
        \draw[->](y) to [bend left=20] node[ne] {$e_3$} (x);
    \end{tikzpicture}
    \caption{Construction in proof of \cref{lem:knitwork_immersion_weak_vs_strong}. $L_1$ is marked in blue and $L_2$ is marked in red.}
    \label{fig:knitwork_immersion_weak_vs_strong}
\end{figure}
\begin{proof}
We will use the same construction for both 1. and 2. Let $\Omega$ be arbitrary for now. Let $G =(V,E)$ be the Eulerian digraph obtained by letting $V=\{s_1,s_2,t_1,t_2,x,y\}$ and $E=\{(t_i,s_i),e_1=(s_1,x),f_1=(x,y), e_3=(y,x),f_3=(x,t_1),e_2=(s_2,x),f_2=(x,t_2) \mid i=1,2\}$. Root $G$ arbitrarily in $s_1$ say. Let $(\mu,\m,\Phi)$ be an $\Omega$-sleeve for $\bar G$ defined via $\dom(\mu) = \{x\} = \dom(\m)$ and $\Phi$ is nowhere defined for now. Let 
    $\MMM\subset \operatorname{Match}(\{e_1,e_2,e_3\},\{f_1,f_2,f_3\})$ be maximal such that for every $M \in \MMM$ with two distinct elements $(e_{i_1},f_{j_1}),(e_{i_2},f_{j_2}) \in M$ it holds $i_1<i_2 \iff j_1 < j_2$. One easily verifies that $\MMM$ is a well-linked link. 
    Let $\LLL=\{L_1,L_2\}$ with $L_1=(e_1,f_1,e_3,f_3)$ and $L_2 = (e_2,f_2)$, then $\LLL$ is a non-linear (but strong) $2$-linkage with $\tau(\LLL) = \{(s_1,t_1),(s_2,t_2)\}$ that is $\m_G$-respecting. Let $\LLL'=\{L_1',L_2'\}$ be a linear linkage in $G$ with the same type, then there is only one choice for $\LLL'$, namely $L_1' = (e_1,f_3)$ and $L_2' = (e_2,f_2)$. But $\LLL'$ is not $\m_G$-respecting. This proves 1.

    Change $\Omega$ to be a well-quasi-order on two incomparable elements $\star_1,\star_2$. Let $H$ be the (disconnected) Eulerian digraph consisting of two circles $C_i=((s_i',t_i'),(t_i',s_i'))$ for $i=1,2$. Let $\HHH$ be obtained from $H$ by rooting $H$ in $s_1'$ and setting $\Phi_H(s_i') = \star_i =\Phi_H(t_i')$ and keeping $\mu_H,\m_H$ nowhere defined. Let $\GGG$ be obtained as above by fixing $\Phi_G(s_i') = \star_i =\Phi_G(t_i')$ for $i=1,2$. Then $\HHH \hookrightarrow^*\GGG$ by mapping $(s_i',t_i')$ to $L_i$ as defined above and $(t_i',s_i')$ to $(t_i,s_i)$. By 1. the claim follows.
\end{proof}
\begin{remark}
    This essentially boils down to ``cutting out'' an Eulerian digraph that does not allow for paths to ``cross'', and encoding the feasible linkages in the routing handler $\m_G$ of a vertex $x$ replacing the cut. 

    We emphasise that we believe that a similar construction should work to prove the same result for \emph{inter-linked} $\Omega$-knitworks, but we did not look into this further.
\end{remark}

On a positive note, and concluding our discussion regarding \cref{def:knitwork_immersion}, we have the following.
\begin{lemma}\label{lem:from_respecting_to_clean_respecting}
    Let $\Omega$ be a well-quasi-order and let $\GGG$ be an $\Omega$-knitwork. Let $\LLL$ be an $\m_G$-respecting linkage in $G$. Then there exists a clean  $\m_G$-respecting linkage $\LLL'$ in $G$ with $\Abs{\LLL'} = \Abs{\LLL}$, $v\text{-}\tau(\LLL') = v\text{-}\tau(\LLL)$ and every path in $\LLL'$ is a subpath of a path in $\LLL$.
\end{lemma}
\begin{proof}
    Let $L \in \LLL$ be a path that is not clean, then $L = P_1 \circ P_2$ where $P_1,P_2$ are paths of length at least $1$, both starting in the same vertex $v \in V(G)$ whence $v\text{-}\tau(P_2) = v\text{-}\tau(L)$. Since $\LLL$ is $\m_G$-respecting, We define $\LLL_1 \coloneqq \LLL\cup\{P_2\} \setminus \{L\}$. Clearly $\LLL_1$ is a linkage and it is still $\m_G$-respecting. Repeating this construction for every path until there is no non-clean path left in the linkage does the trick.
\end{proof}
This last result can essentially be used to force weak $\Omega$-knitwork immersion to result in clean paths in the spirit of \cref{obs:linkage_gives_linear_linkage} as we will make precise.

\subsection{Ordering by $\Omega$-Knitwork immersion} 
We start by proving that the newly introduced immersion relations induce a quasi-order on~$\Omega$-knitworks, i.e., they are transitive and reflexive. Indeed, the proof for weak immersion \emph{needs} the weaker assumption in \cref{def:knitwork_immersion} on edges to be mapped to paths that allow repeated vertices. In particular \cref{lem:knitwork_immersion_weak_vs_strong} can be used to prove that the following lemma would otherwise fail. Note that we define \emph{strong} $\Omega$-knitwork immersion also in terms of non-linear paths. Although we \emph{could} define them via linear paths (the following lemma would still hold true), we will \emph{need} the weaker notion of non-linear paths in light of \cref{lem:from_linear_to_respecting_linkage_well-linked}. We will mainly work with well-linked $\Omega$-knitworks (or variations thereof) that are not guaranteed to be inter-linked, and thus we may not be able to reroute paths to linear paths.

The following is straightforward from the definitions; we give a proof for completion.
\begin{observation}\label{obs:knitwork_imm_lifts_to_subknitwork}
     Let~$\Omega=(V(\Omega),\preceq)$ be a well-quasi-order and $\ell \geq 1$. Let $\GGG_1,\GGG_2$ be $\Omega$-knitworks. Let $\HHH_2\subseteq \GGG_2$ be a sub-knitwork. If $\GGG_1 \hookrightarrow \HHH_2$ then for every sub-knitwork $\HHH_1 \subseteq \GGG_1$ we derive $\HHH_1 \hookrightarrow \GGG_2$, and the same holds true for weak $\Omega$-knitwork immersion.
\end{observation}
\begin{proof}
    Let $\gamma: \GGG_1 \hookrightarrow \HHH_2$ with witness $\HHH_2'$ and let $\gamma(G_1) \subseteq H_2'$ be the respective subgraph. Let $\HHH_1 \subseteq \GGG_1$ and let $H_2'' \coloneqq \gamma(H_1)$ and let $\HHH_2''$ be the respective sub-knitwork of $\HHH_2'$ with underlying digraph $H_2''$, note that $\gamma$ is an immersion of the underlying directed graphs. One easily verifies that the map $\gamma':V(H_1)\cup E(H_1) \to H_2''$ obtained from restricting $\gamma$ to $H_1$ witnesses $\HHH_1 \hookrightarrow \HHH_2''$ and thus $\HHH_1 \hookrightarrow \GGG_2$ as desired.
\end{proof}

\begin{theorem}\label{lem:knitwork_imm_is_quasi_order}
    Let~$\Omega=(V(\Omega),\preceq)$ be a well-quasi-order and $\ell \geq 1$. Then $\hookrightarrow$ and $\hookrightarrow^*$ induce quasi-orders on the class of~$\Omega$-knitworks of index $\ell$.
\end{theorem}
\begin{proof}
    We provide a proof for~$\hookrightarrow$ since the proof for $\hookrightarrow^*$ is analogous (We will highlight the crucial points where assumptions for weak immersion are needed).
    
    The relation is clearly reflexive via the identity map, i.e., $\HHH \hookrightarrow \HHH$. We continue with transitivity. As a first step, note that it suffices to prove the claim for \emph{stable} $\Omega$-knitwork immersion.
    
    \begin{claim}\label{lem:knitwork_imm_is_quasi_order_claim_stable}
        Let $\hookrightarrow^+$ denote stable strong $\Omega$-knitwork immersion. If $\hookrightarrow^+$ induces a quasi-order on the class of $\Omega$-knitworks, then so does $\hookrightarrow$. The same holds true for weak $\Omega$-knitwork immersion.
    \end{claim}
    \begin{claimproof}
        For $i \in \{ 1,2,3\}$ let~$(\bar{G}_i,\mu_i,\m_i,\Phi_i)$ be~$\Omega$-knitworks with~$\bar{G}_i = (G_i,\pi_i, X_1^i,\ldots, X_\ell^i)$ and let $\gamma_1: (\bar{G}_1,\mu_1,\m_1,\Phi_1) \hookrightarrow (\bar{G}_2,\mu_2,\m_2,\Phi_2)$ and~$\gamma_2: (\bar{G}_2,\mu_2,\m_2,\Phi_2) \hookrightarrow (\bar{G}_3,\mu_3,\m_3,\Phi_3)$ be strong~$\Omega$-knitwork immersions. Thus, there exist $\GGG_2' \subseteq \GGG_2$ and $\GGG_3' \subseteq \GGG_3$ such that $\gamma_1:\GGG_1 \hookrightarrow^+ \GGG_2'$ and $\gamma_2: \GGG_2 \hookrightarrow^+ \GGG_3'$ satisfy $(1)-(6)$ of \cref{def:isomorphism_knitwork}; in particular both are stable strong $\Omega$-knitwork immersions. By \cref{obs:knitwork_imm_lifts_to_subknitwork} we derive the existence of a strong $\Omega$-knitwork immersion $\gamma_2':\GGG_2'\hookrightarrow \GGG_3'$, and thus by \cref{def:knitwork_immersion} we derive the existence of a witness $\GGG_3''\subseteq \GGG_3'$ and a stable strong $\Omega$-knitwork immersion $\gamma_2':\GGG_2' \hookrightarrow^+ \GGG_3''$. By assumption of the claim, $\hookrightarrow^+$ is transitive, hence there is a stable strong $\Omega$-knitwork immersion $\eta : \GGG_1 \hookrightarrow^+ \GGG_3''$ whence the claim follows by \cref{def:knitwork_immersion}.

        The proof for weak $\Omega$-knitwork immersion follows verbatim.
    \end{claimproof}
    
    By \cref{lem:knitwork_imm_is_quasi_order_claim_stable} we may assume without loss of generality that throughout the remainder of the proof, we work with stable strong $\Omega$-knitwork immersion; we denote it by $\hookrightarrow$ for simplicity. Hence, for $i \in \{ 1,2,3\}$ let~$(\bar{G}_i,\mu_i,\m_i,\Phi_i)$ be~$\Omega$-knitworks with~$\bar{G}_i = (G_i,\pi_i, X_1^i,\ldots, X_\ell^i)$ and let
    $\gamma_1: (\bar{G}_1,\mu_1,\m_1,\Phi_1) \hookrightarrow (\bar{G}_2,\mu_2,\m_2,\Phi_2)$ with witness $\GGG_2$ and~$\gamma_2: (\bar{G}_2,\mu_2,\m_2,\Phi_2) \hookrightarrow (\bar{G}_3,\mu_3,\m_3,\Phi_3)$ with witness $\GGG_3$ be stable strong~$\Omega$-knitwork immersions in particular they each satisfy $(1)$-$(6)$ of \cref{def:knitwork_immersion}; we omit the word stable from here on and assume it tacitly.

    Note that the three $\Omega$-knitworks are of common index $\ell \geq 1$ by \cref{def:knitwork_immersion}. Further,~$\gamma_1$ and $\gamma_2$ are strong rooted immersions by \crefdef{def:knitwork_immersion}{1} and thus~$\delta(X_j^i) = k_j \in 2\N$ for all~$i \in \{1,2,3\}$ and $1 \leq j \leq \ell$.
    
    Let~$\gamma$ be defined on $V(G_1) \cup E(G_1)$ via~$\gamma \coloneqq \gamma_2 \circ \gamma_1$. Recall that for $e \in E(G_1)$, if $\gamma_1(e) :=  (e_1', \dots, e_t')$ then $\gamma_2 \circ \gamma_1 (e) = \gamma_2((e_1',\ldots,e_t')) \coloneqq \gamma_2(e_1')\circ\ldots\circ \gamma_2(e_t')$ is the concatenation of the paths $\gamma_2(e'_j)$ for $1 \leq j \leq t$.
    \begin{claim}
        $\gamma$ strongly immerses~$\bar{G}_1$ into~$\bar{G}_3$.
    \end{claim}
    \begin{claimproof}
        Since~$\gamma_1,\gamma_2$ are strong immersions on the underlying unrooted digraphs, using the transitivity for strong immersion we have that~$\gamma: G_1 \hookrightarrow G_3$ is a strong immersion.
        
        Let $e \in E(G_1)$ and $\gamma(e) = \gamma_2 \circ \gamma_1 (e) = \gamma_2((e_1',\ldots,e_t')) \coloneqq \gamma_2(e_1')\circ\ldots\circ \gamma_2(e_t')$ for $t \geq 1$ and $e_1',\ldots,e_t' \in E(G_2)$ such that $\gamma_1(e) = (e'_1, \dots, e'_t)$. By \cref{obs:immersion_maps_path_to_path} we derive that $(e_1',\ldots,e_t')$ is a path in $G_2$ and combining \cref{obs:immersion_maps_path_to_path} together with \cref{obs:concat_paths} implies that $\gamma(e)$ is a path in $G_3$. (Here the proof would still work for strong immersion by insisting on linear paths, but the same would fail for weak immersion.)
        
        It remains to show that~$\gamma$ respects the order of the roots; note that $\gamma(X_j^1) \subseteq \gamma(X_j^3)$ for every $1 \leq j \leq \ell$ again by transitivity and definition. To this extent let~$\pi_i(X^i_j)=(e_{1,j}^i,\ldots,e_{k_j,j}^i)$ for every $1 \leq j \leq \ell$ and $i \in \{1,2,3\}$; by \crefdef{def:knitwork_immersion}{1} and \cref{def:rooted_immersion} we derive that~$e_{t,j}^2 \in \gamma_1(e_{t,j}^1)$ and~$e_{t,j}^3 \in \gamma_2(e_{t,j}^2)$ for~$1 \leq t \leq k_j$. By definition of~$\gamma$ together with \cref{obs:immersion_maps_path_to_path} this implies that~$e_{t,j}^3 \in \gamma(e_{t,j}^1)$. Clearly no~$e_{t',j}^3 \in \gamma(e_{t,j}^1)$ for~$t' \neq t$ since~$E(\gamma(e_{t',j}^3)) \cap E(\gamma(e_{t,j}^3)) = \emptyset$ using the fact that~$\gamma$ is an immersion on the underlying digraphs. 
    \end{claimproof}

    \begin{claim}
        For each~$v_1 \in V(G_1)$,~$\gamma(v_1) \in \dom(\mu_3)$ if and only if~$v_1 \in \dom(\mu_1)$, as well as $\gamma(v_1) \in \dom(\m_3)$ if and only if $v_1 \in \dom(\m_1)$.
    \end{claim}
    \begin{claimproof}
        For each~$v_1 \in V(G_1)$ we know by \crefdef{def:knitwork_immersion}{2} that~$v_2 \coloneqq \gamma_1(v_1) \in \dom(\mu_2)$ if and only if~$v_1 \in \dom(\mu_1)$. Then again by the same reasoning~$\gamma_2(v_2) \in \dom(\mu_3)$ if and only if~$v_2 \in \dom(\mu_2)$ concluding the proof of the first part of the claim; the second part is analogous.
    \end{claimproof}

    \begin{claim}
        For each~$v_1 \in \dom(\mu_1)$ the index $t$ of~$\mu_1(v_1)=(f_1^1,\ldots,f_t^1)$ and~$t'$ of $\mu_3(\gamma(v_1))=(f_1^3,\ldots,f_{t'}^3)$ are equal and even, 
        and for every~$1 \leq i \leq t$ the edge~$f_i^3$ is contained in~$\gamma(f_i^1)$.
    \end{claim}
    \begin{claimproof}
        Recall that following \cref{lem:knitwork_imm_is_quasi_order_claim_stable} we assume the immersions to be stable. Using \cref{def:knitwork_immersion}\,\cref{def:knitwork_immersion:3} for~$\gamma_1$ and~$\gamma_2$ we derive that~$t = t'$ by transitivity of equality, and since $\dom(\mu_i)$ is disjoint from the root-sets of $\bar G_i$ for every $i\in\{1,2,3\}$, the vertices $v_1$ and $\gamma(v_1)$ are Eulerian by \cref{def:rooted_eulerian_digraph} implying $t \in 2\N$. Further, \crefdef{def:knitwork_immersion}{3} implies that for~$1 \leq i \leq t$ the edge~$f_i^2$ is contained in~$\gamma_1(f_i^1)$ where~$\mu_2(\gamma_1(v_1)) = (f_1^2,\ldots,f_t^2)$. Now~$P_i^2\coloneqq \gamma_1(f_i^1)$ is a path in~$G_2$ containing the edge~$f_i^2$. By \cref{obs:immersion_maps_path_to_path}~$\gamma_2(P_i^2)$ is a path in~$G_3$ and since~$f_i^2 \in E(P_i^2)$ it follows from \crefdef{def:knitwork_immersion}{3} that~$f_i^2 \in \gamma_2(P_i^2) = \gamma_2(\gamma_1(e_i^1)) = \gamma(e_i^1)$. This concludes the proof.
    \end{claimproof}

      \begin{claim}
$M_\gamma(v_3) \in \m_3(v_3)$ for each~$v_3 \in \dom(\m_3)\setminus \gamma(\dom(\m_1))$.
    \end{claim}
    \begin{claimproof}
        Since~$v_3 \notin \gamma(\dom(\m_1))$ there are two cases to consider. Again let~$(f_1^3,\ldots,f_t^3) = \mu_3(v_3)$.

        Assume first that there is~$v_2 \in \dom(\m_2)$ with~$\gamma_2(v_2) = v_3$, which by the assumption of the claim implies~$v_2 \notin \gamma_1(\dom(\m_1))$ and thus~$v_2 \notin \gamma_1(V(G_1))$ by \crefdef{def:knitwork_immersion}{2} applied to $\gamma_1$. By \crefdef{def:knitwork_immersion}{4} we derive~$M_{\gamma_1}(v_2) \in \m_2(v_2)$. The claim follows by \crefdef{def:knitwork_immersion}{5} applied to~$\gamma_2$.

        Next assume that there is no such~$v_2 \in \dom(\m_2)$, whence~$M_{\gamma_2}(v_3) \in \m_3(v_3)$ by \crefdef{def:knitwork_immersion}{4} applied to~$\gamma_2$. Let~$E_2 \subseteq E(G_2)$ be the minimal set of edges for which~$\LLL^2_3=\{\gamma_2(e) \mid e \in E_2\}$ witnesses~$M_{\LLL^2_3}(v_3) = M_{\gamma_2}(v_3)$.
        Let~$E_1 \subseteq E(G_1)$ be the minimal set of edges maximising~$E_2 \cap E(\gamma_1(E_1))$. Since~$\gamma:G_1 \hookrightarrow G_3$
       is an immersion it follows that~$\LLL^1_3 = \{\gamma(e) \mid e \in E_1\}$ is a linkage in~$G_3$ such that if~$(f_i^3,f_j^3) \subset L$ is a subpath of some~$L \in \LLL^1_3$ then~$(f_i^3,f_j^3)$ is a subpath of some~$L \in \LLL^2_3$ for any~$1 \leq i,j \leq t$. The claim now follows by definition of~$M_\gamma(v_3)$ and the fact that all the links in $\m_3$ are reliable by \cref{def:knitwork}.
    \end{claimproof}
    
    \begin{claim}
        For each~$v_1 \in \dom(\m_1)$ with $\mu_1(v_1)=(f_1^1,\ldots,f_t^1)$ and~$\mu_3(\gamma(v_1))=(f_1^3,\ldots,f_{t}^3)$, there is an injection $\alpha_1^3:\m_1(v_1) \to \m_3(\gamma(v_1))$ such that for every $M \in \m_1(v_1)$ it holds ~$(f_i^1,f_j^1) \in M$ if and only if~$(f_i^3,f_j^3) \in \alpha_1^3(M)$ for all~$1\leq i,j \leq t$.
    \end{claim}
    \begin{claimproof}
        This follows at once from 
        \crefdef{def:knitwork_immersion}{5} once for~$\gamma_1$ and once for~$\gamma_2$, yielding two injections $\alpha_1^2:\m(v_1) \to \m_2(\gamma_1(v_1))$ and $\alpha_2^3:\m_2(\gamma_1(v_1)) \to \m_3(\gamma_2(\gamma_1(v_1)))$ which can be concatenated to form $\alpha_1^3$ satisfying the claim. 
    \end{claimproof}

    \begin{claim}
        For every~$v_1 \in V(G_1)$ we have~$\Phi(v_1) \preceq \Phi_3(\gamma(v_1))$.
    \end{claim}
    \begin{claimproof}
        This follows at once from \crefdef{def:knitwork_immersion}{6} for~$\gamma_1,\gamma_2$ and the fact that~$\preceq$ is transitive, i.e.,~$\Phi(v_1) \preceq \Phi(\gamma_1(v_1)) \preceq \Phi(\gamma_2(\gamma_1(v_1))) = \Phi(\gamma(v))$.
    \end{claimproof}
    All in all we verified all of the conditions for \cref{def:knitwork_immersion} concluding the proof for the transitivity, and thus the proof that (stable)~$\Omega$-knitwork immersion induces a quasi-order on~$\Omega$-knitworks; recall that we implicitly used \cref{lem:knitwork_imm_is_quasi_order_claim_stable}.
\end{proof}

We take a brief moment to note the following.

\begin{lemma}
    Let $\Omega$ be a well-quasi-order and let $\HHH,\GGG$ be $\Omega$-knitworks of index $\ell \geq 1$. If $\HHH \hookrightarrow^* \GGG$, then there exists $\gamma^*:\HHH \hookrightarrow^* \GGG$ such that $\gamma^*(E(H))$ is clean.
\end{lemma}
\begin{proof}
    Let $\gamma:\HHH \hookrightarrow^* \GGG$ and fix $\LLL \coloneqq \gamma(H)$. By \cref{lem:from_respecting_to_clean_respecting} there exists a clean $\m_\GGG$-respecting linkage $\LLL^*$ with $v\text{-}\tau(\LLL^*) = v\text{-}(\LLL)$. Let $\gamma^*$ be defined on $V(H) \cup E(H)$ via $\restr{\gamma^*}{V(H)} = \gamma$ and for $e \in E(H)$, $\gamma^*(e) = L \in \LLL^*$ if and only if $L \subseteq \gamma(e)$. One easily verifies that $(\gamma^*(V(H)),\LLL^*)$ is an immersion model of $\HHH$ in $\GGG$. Note that, if $\HHH$ is rooted in $(\pi_\HHH,A_1,\ldots,A_\ell)$ say, and $\GGG$ is rooted in $(\pi_G,X_1,\ldots,X_\ell)$, say, then we never reroute at vertices in $\bigcup_{i=1}^\ell X_i$, since $\delta(A_i) = \delta(X_i)$. Similarly, we never reroute at vertices in $\dom(\mu_\GGG)$ by \crefdef{def:knitwork_immersion}{3}.
\end{proof}

Next, we deal with two crucial operations that we call \emph{stitching} and \emph{knitting}.

\subsection{Stitching}\label{subsec:stitching}  Intuitively,
given a proper rooted cut~$Y$ we define the \emph{up-stitch} and \emph{down-stitch} of $G$ at $Y$ to be the graphs obtained by contracting~$\bar{Y}$ and~$Y$ into a single vertex~$y^*$ and~$y_*$ respectively. See \cref{fig:down-up-stitching} for an illustration. The following is the formal definition. Recall that given~$Y \subset V(G)$ we defined~$\bar{Y} = V(G) \setminus Y$.

\begin{definition}[Stitching]\label{def:stitching_std}
      Let~$\bar{G}=(G,\pi, X_1,\ldots,X_\ell)$ be a rooted Eulerian digraph of index $\ell \geq 1$ where~$G=(V,E,\operatorname{inc})$. Let~$\pi(X_i) \coloneqq (f_1^i,\ldots,f^i_t)$ for~$t \in 2\N$. Let $k\in 2\N$ and~$Y \subset V(G)$ induce a proper rooted~$k$-cut in~$\bar G$. Let $\tau \geq 1$ such that $\{X^{\insc}_1,\ldots,X^{\insc}_\tau\} = \insc(Y)$ and $\{X^{\outc}_1,\ldots,X_{\ell-\tau}^{\outc}\} = \outc(Y)$ (possibly empty). Let~$ \pi(Y) = (e_1,\ldots,e_k) = \pi(\bar{Y})$ be an ordering of~$\rho(Y)$ and~$\rho(\bar Y)$ respectively. Let~$y_\ast,y^\ast$ be two new elements that are not part of~$V \cup E$. We define~$G_Y \coloneqq (V_Y,E_Y,\operatorname{inc}_Y)$ via~
    \begin{itemize}
        \item $V_Y \coloneqq Y \cup \{y_\ast\}$,
        \item $E_Y \coloneqq E(G[Y]) \cup \{e_1,\ldots,e_k\}$,
        \item $\operatorname{inc}_Y = \operatorname{inc}_Y^- \cup \operatorname{inc}_Y^+$, with \begin{align*}
         \operatorname{inc}_Y^- &\coloneqq \{ (e,v) \sth (e,v) \in \operatorname{inc}, e \in E_Y,\text{ and }v\in Y\} \cup\{ (e,v) \sth e \in \rho^+(Y)\text{ and }v=y_\ast \},\text{ and}\\
        \operatorname{inc}_Y^+ &\coloneqq \{ (v,e) \sth   
         (v,e) \in \operatorname{inc}, e \in E_Y,\text{ and }v\in Y\} \cup \{(v,e) \sth e \in \rho^-(Y)\text{ and }v=y_\ast \}.
        \end{align*}
    \end{itemize} 
    
 We define~$\stitch(\bar{G};\pi,Y) \coloneqq (G_Y,\pi,X^{\insc}_1,\ldots,X^{\insc}_\tau)$ and say that~\emph{$G_Y$ is obtained from~$G$ by down-stitching~$Y$} and call~$y_\ast$ the \emph{down-stitch vertex resulting from~$Y$}. 

Similarly we define~$G^Y \coloneqq (V^Y,E^Y,\operatorname{inc}^Y)$ via~
\begin{itemize}
        \item $V^Y \coloneqq \bar{Y} \cup \{y^\ast\}$,
        \item $E^Y \coloneqq E(G[\bar{Y}]) \cup \{e_1,\ldots,e_k\}$,
        \item $\operatorname{inc}^Y = \operatorname{inc}^Y_- \cup \operatorname{inc}^Y_+$, with \begin{align*}
         \operatorname{inc}^Y_- &\coloneqq \{ (e,v) \in \operatorname{inc}, e \in E^Y, \text{ and } v\in \bar{Y}\} \cup\{ (e,v) \sth e \in \rho^+(\bar{Y}) \text{ and } v=y^\ast \},\text{ and}\\
        \operatorname{inc}^Y_+ &\coloneqq \{ (v,e) \sth   
         (v,e) \in \operatorname{inc}, e \in E^Y,\text{ and }v\in \bar Y\} \cup \{(v,e) \sth  e \in \rho^-(\bar{Y}) \text{ and } v=y^\ast\}.
        \end{align*}
    \end{itemize} 

    Note that $\rho(y^\ast) = \rho(\bar{Y})$;  fix $\pi(y^\ast) \coloneqq \pi(\bar Y)$. We define~$\stitch(\bar{G};\pi,\bar{Y}) \coloneqq (G^Y,\pi,y^\ast,{X_1}^{\outc}, \ldots,{X^{\outc}_{\ell - \tau}})$ and say that~\emph{$G^Y$ is obtained from~$G$ by up-stitching~$Y$} and call~$y^\ast$ the \emph{up-stitch vertex resulting from~$Y$}. 
\end{definition}
See \cref{fig:down-up-stitching} for a schematic representation of \cref{def:stitching_std}.
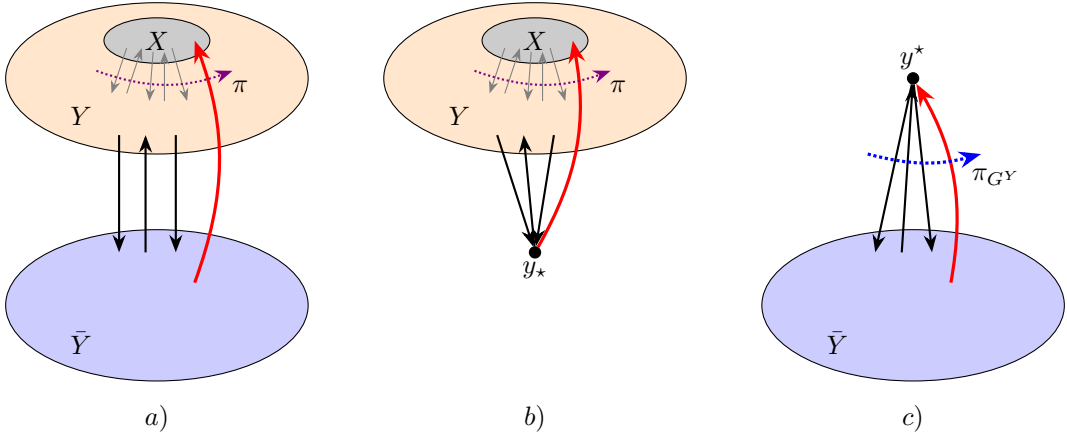
\begin{figure}
    \centering
    \begin{tikzpicture}[>=Stealth]
        \filldraw[fill=orange!20!white, draw=black] (0,0) ellipse [x radius=2, y radius=1] ;
       \filldraw[fill=black!20!white, draw=black] (0,0.5) ellipse [x radius=0.7, y radius=0.3] ;
        \filldraw[fill=blue!20!white, draw=black] (0,-3) ellipse [x radius=2, y radius=1] ;
        \node at (-1,-0.5) {$ Y$ } ;
        \node at (0,0.5) {$ X$ } ;
        
        \node at (-1,-3.5) {$ \bar Y$ } ;
        \draw [thick, black,->] (-0.5, -0.75) to (-0.5,-2.3) ;       
        \draw [thick, black,<-] (-0.15, -0.75) to (-0.15,-2.3) ;       
        \draw [thick, black,->] (0.25, -0.75) to (0.25,-2.3) ;       
        \draw [very thick, red , bend left=20, <-] (0.5, 0.5) to (0.5,-2.7) ;       
        \draw [gray, thin, ->] (-0.4, 0.4) to (-0.6,-0.2) ; 
        \draw [gray, thin, <-] (-0.2, 0.4) to (-0.4,-0.2) ; 
        \draw [gray, thin, ->] (-0.05, 0.35) to (-0.1,-0.3) ; 
        \draw [gray, thin, <-] (0.1, 0.35) to (0.1,-0.3) ; 
        \draw [gray, thin, ->] (0.2, 0.4) to (0.4,-0.3) ; 
         \draw[violet, thick, densely dotted, bend right=20,->] (-0.8,0.1) to (1,0.1) ;
        \node at (1.1,-0.1) {$\pi$};
        \begin{scope}[xshift=5cm]
        \filldraw[fill=orange!20!white, draw=black] (0,0) ellipse [x radius=2, y radius=1] ;
         \filldraw[fill=black!20!white, draw=black] (0,0.5) ellipse [x radius=0.7, y radius=0.3] ;
        \node[circle, draw=black, fill=black, inner sep=1.5pt] (yd) at (0,-2.3) { };
        \node[anchor=north] (nyd) at (yd) { $y_\star$};
        \node at (-1,-0.5) {$ Y$ } ;
        \node at (0,0.5) {$ X$ } ;
        \draw [thick,black,->] (-0.5, -0.75) to (yd) ;       
        \draw [thick,black,<-] (-0.15, -0.75) to (yd) ;       
        \draw [thick,black,->] (0.25, -0.75) to (yd) ;       
        \draw [very thick, red,bend left=20,<-] (0.5, 0.5) to (yd) ;   \draw [gray, thin, ->] (-0.4, 0.4) to (-0.6,-0.2) ; 
        \draw [gray, thin, <-] (-0.2, 0.4) to (-0.4,-0.2) ; 
        \draw [gray, thin, ->] (-0.05, 0.35) to (-0.1,-0.3) ; 
        \draw [gray, thin, <-] (0.1, 0.35) to (0.1,-0.3) ; 
        \draw [gray, thin, ->] (0.2, 0.4) to (0.4,-0.3) ; 
        \draw[violet, thick, densely dotted, bend right=20,->] (-0.8,0.1) to (1,0.1) ;
        \node at (1.1,-0.1) {$\pi$};
        \end{scope}
        \begin{scope}[xshift=10cm]
        \filldraw[fill=blue!20!white, draw=black] (0,-3) ellipse [x radius=2, y radius=1] ;
        \node[inner sep=1.5pt,circle, draw=black, fill=black] (yu) at (0,0) { };
        \node[anchor=south] (nyu) at (yu) { $y^\star$};
        \node at (-1,-3.5) {$\bar Y$ } ;
        \draw [thick, black,<-] (-0.5, -2.3) to (yu) ;       
        \draw [thick, black,->] (-0.15, -2.3) to (yu) ;       
        \draw [thick, black,<-] (0.25, -2.3) to (yu) ;       
        \draw [very thick, red , bend left=20, <-] (yu) to (0.5,-2.7) ;            
        \draw[blue, very thick, densely dotted, bend right=15,->] (-0.6,-1) to (0.9,-1) ;
        \node at (1.1,-1.3) {$\pi_{G^Y}$};
        \end{scope}
        \node at (0,-4.5) { $a)$ } ;
        \node at (5,-4.5) { $b)$ } ;
        \node at (10,-4.5) { $c)$ } ;

    \end{tikzpicture}
    \caption{Down- and Up-Stitching. Figure $a)$ depicts a clamped graph $\bar G=(G,\pi,X)$ with a rooted cut induced by $Y$, a choice of $\pi(X)$ highlighted by a dotted oriented arrow, and a highlighted edge that is part of $\rho(X)\cap\rho(Y)$. Figure $b)$ depicts the down-stitch of $\bar G$ and $Y$ with down-stitch vertex $y_\star$. Figure $c)$ depicts the up-stitch of $\bar G$ and $Y$ with up-stitch vertex $y^\star$ fixing new roots $\pi_{G^Y}(y^*)$ on $\rho(y^*)$, highlighted by a dotted oriented arrow.}
    \label{fig:down-up-stitching}
\end{figure}

The following relations are readily extracted from the definition; recall \cref{lem:rooted_Eulerian_cuts_are_Eulerian}.
\begin{observation}\label{obs:stitching_fundamentals}
     Let~$\ell \geq 1$, let ~$\bar{G}=(G,\pi,X_1,\ldots,X_\ell)$ be a rooted Eulerian digraph and let~$Y$ induce a proper rooted cut with fixed orderings~$\pi(Y) =\pi(\bar{Y})$. Let~$\bar G_Y = \stitch(\bar{G};\pi,Y)$ and~$\bar G^Y = \stitch(\bar{G};\pi,\bar{Y})$ be the down- and up-stitches of~$G$ at~$Y$ with down- and up-stitch vertices~$y_\ast,y^\ast$ respectively. Then
     \begin{enumerate}
         \item $y_*$ and $y^*$ are Eulerian vertices and $G_Y$ and $G^Y$ are well-defined Eulerian digraphs up to $\bigcup_{X \in \insc(Y)} X$ and $\bigcup_{X \in \outc(Y)} X$ respectively. If $G$ is loopless, then $G_Y$ and $G^Y$ are loopless, and\label{obs:stitching_fundamentals:1}
         \item $Y \subset V(G_Y)$ and $\insc_G(Y) = \insc_{G_Y}(Y)$, as well as~$\bar{Y} \subset V(G^Y)$ and $\outc_G(Y) = \insc_{G^Y}(\bar Y)$. Additionally~$G_Y[Y] = G[Y]$ as well as~$G^Y[\bar{Y}] = G[\bar{Y}]$, where $E(G_Y) \cap E(G^Y) = \rho(Y)$, and\label{obs:stitching_fundamentals:2}
         \item $\rho_{G_Y}(y_\star) = \rho_{G_Y}(Y) = \rho_G(Y)$, and~$\rho_{G_Y}(X) = \rho_G(X)$ for every $X \in \insc(Y)$ whence $\bar G_Y$ is a well-defined rooted Eulerian digraph of index $\tau = \Abs{\insc(Y)} \geq 1$, and\label{obs:stitching_fundamentals:3}
         \item $\rho_{G^Y}(\bar{Y}) = \rho_{G^Y}(y^\ast) = \rho_G(\bar{Y})$, and~$\rho_{G^Y}(X) = \rho_G(X)$ for every $X \in \outc(Y)$ whence $\bar G^Y$ is a well-defined rooted Eulerian digraph of index $\tau - \ell +1 \geq 1$.\label{obs:stitching_fundamentals:4}
        
     \end{enumerate} 
\end{observation}
\begin{remark}
    Although technically the graphs obtained by stitching are new graphs, and the incidences of edges vary due to introducing $y_\ast$ and $y^\ast$, the edges themselves as objects remain the same. This way we may indeed write~$\rho_G(\bar{Y}) = \rho_{G^Y}(y^*)$ instead of carrying a bijection between these sets.
\end{remark}

We extend the definition of stitching to~$\Omega$-knitworks as follows.
\begin{definition}[Types of feasible linkages]\label{def:types_of_linkages_on_a_cut}
    Let~$\Omega$ be a well-quasi-order and~$\GGG$ an $\Omega$-knitwork. Let~$X \subset V(G)$ induce a Eulerian~$k$-cut in $G$ for some $k\in 2\N$. Let~$\LLL_\GGG(\rho,X)$ be the set of all $\m_G$-respecting strong~$(\rho^-(X),\rho^+(X))$-linkages~$\LLL$ such that for every $P \in \LLL$ it holds $V^\circ(P) \subseteq X$. We define~$\mathfrak{M}_\GGG(\rho,X) \coloneqq \{ \tau(\LLL) \mid \LLL \in \LLL_\GGG(\rho,X)\}$ to be the \emph{set of feasible~$(\rho,X)$-types}. 
\end{definition}
\begin{remark}
We may omit the subscript $\GGG$ and write $\LLL(\rho,X)$ as well as $\MM_\GGG(\rho,X)$ if clear from context.

    Recall that $V^\circ(P) \subseteq X$ implies that all the paths in $\LLL$ use only vertices in $X$ except for their endpoints which may be part of $\bar X$. Note further that while~$\rho(X) = \rho(\bar{X})$,~$\mathfrak{M}(\rho,X)$ and~$\mathfrak{M}(\rho,\bar X)$ may differ. We choose this notation to emphasise the importance of the set of vertices inducing the cut. 

\end{remark}
Note that $M \subseteq \operatorname{Match}(\rho^-(X),\rho^+(X))$ for all $M \in \mathfrak{M}(\rho,X)$  and, by definition,  $M' \in \mathfrak{M}(\rho,X)$ for every $M' \subseteq M$. In particular $\mathfrak{M}(\rho,X)$ is a reliable link by definition. Thus, the following is well-defined.

\begin{definition}\label{def:stitching_knitwork}
   Let~$\Omega$ be a well-quasi-order, $k \in 2\N$, and let~$\GGG = (\bar{G},\mu,\m,\Phi)$ be an~$\Omega$-knitwork. Let~$Y \subset V(G)$ induce a proper rooted~$k$-cut in~$\bar{G}$. Let~$\pi(Y) = (e_1,\ldots,e_k) = \pi(\bar{Y})$ be an ordering of~$\rho(Y)=\rho(\bar{Y})$. Then we define $\stitch(\GGG;\pi,Y) \coloneqq (\bar G_Y,\mu_Y,\m_Y,\Phi_Y)$ as follows:
   \begin{itemize}
       \item $\bar G_Y \coloneqq \stitch(\bar{G};\pi,Y) \text{ with down-stitch vertex } y_\ast$,
       \item$\mu_Y(v) \coloneqq \begin{cases}
           \mu(v), & v\in Y\\
           (e_1,\ldots,e_k), &v = y_\ast\end{cases},$
        \item$\m_Y(v)\coloneqq \begin{cases}
           \m(v), &\ v\in Y\\
           \mathfrak{M}(\rho,\bar{Y}), &v = y_\ast\end{cases}, \text{ and}$
       \item$ \Phi_Y(v) \coloneqq \Phi(v) \text{ for } v\in Y,$
   \end{itemize}
   and we do not specify~$\Phi_Y$ on~$y_\ast$.

    We further define $\stitch(\GGG;\pi,\bar{Y}) \coloneqq (\bar G^Y,\mu^Y, \m^Y,\Phi^Y)$ as follows:
   \begin{itemize}
       \item$\bar G^Y \coloneqq \stitch(\bar{G};\pi,\bar{Y}) \text{ with up-stitch vertex } y^\ast$,
      \item $\mu^Y(v) \coloneqq \mu(v)$ for $v\in \bar Y$
      \item $\m^Y(v)\coloneqq \m(v), \text{ for }\ v\in \bar Y, \text{ and}$
    \item$ \Phi^Y(v) \coloneqq \Phi(v) \text{ for } v\in \bar{Y}$,
   \end{itemize}
    and we do not specify~$\mu^Y,\m^Y$ and $\Phi^Y$ on~$y^\ast$. 
\end{definition}

Using the above definition together with the fact that $\mathfrak{M}(\rho,\bar Y)$ is a reliable link and \cref{obs:stitching_fundamentals}, the following is straightforward.

\begin{observation}\label{obs:stitching_fundamentals_knitworks}
     Let~$\Omega$ be a well-quasi-order and let~$\GGG = (\bar{G},\mu,\m,\Phi)$ be an~$\Omega$-knitwork. Let~$Y \subset V(G)$ induce a proper rooted~$k$-cut in~$\bar{G}$. Let~$\pi(Y) = (e_1,\ldots,e_k) = \pi(\bar{Y})$ be an ordering of~$\rho(Y)=\rho(\bar{Y})$. Let $(\bar G_Y,\mu_Y,\m_Y,\Phi_Y) = \stitch(\GGG;\pi,Y)$ as well as~ $(\bar G^Y,\mu^Y,\m^Y,\Phi^Y) = \stitch(\GGG;\pi,\bar{Y})$. Then
     \begin{enumerate}
     \item  $\stitch(\GGG;\pi,Y)$ and~$\stitch(\GGG;\pi,\bar{Y})$  are well-defined~$\Omega$-knitworks,\label{obs:stitching_fundamentals_knitworks:1}
         \item $\mu(x) = \mu_Y(x)$ for all~$x \in Y$ and~$\mu(x) = \mu^Y(x)$ for all~$x \in \bar{Y}$, and moreover~$\mu_Y(y_\ast) =\pi_G(Y) = \pi_{G^Y}(y^\ast) =\pi_G(\bar Y )=  \mu^Y(y^\ast)$,\label{obs:stitching_fundamentals_knitworks:2}
        \item $\m(x) = \m_Y(x)$ for all~$x \in Y$ and~$\m(x) = \m^Y(x)$ for all~$x \in \bar{Y}$, and\label{obs:stitching_fundamentals_knitworks:3}
         \item $\Phi(x) = \Phi_Y(x)$ for all~$x \in Y$ and~$\Phi(x) = \Phi^Y(x)$ for all~$x \in \bar{Y}$.\label{obs:stitching_fundamentals_knitworks:4}
     \end{enumerate}
\end{observation}

We further get the following useful corollary from \cref{def:stitching_knitwork} for up-stitches.

\begin{observation}\label{obs:up-stitch_of_well-linked_is_well-linked}
Let~$\Omega$ be a well-quasi-order and~$\GGG$ a well-linked~$\Omega$-knitwork for a rooted Eulerian digraph~$\bar G$. Let~$Y \subseteq V(G)$ be a proper rooted cut in~$\bar G$. Then~$\stitch(\GGG;\pi,\bar{Y})$ with up-stitch vertex~$y^*$ is a well-linked~$\Omega$-knitwork.
\end{observation}
\begin{proof}
    This follows immediately from the fact that~$y^* \notin \dom(\m^Y)$ and~$\m^Y(y) = \m(y)$ for all~$y \in Y$.
\end{proof}
\begin{remark}
    Note that \cref{obs:up-stitch_of_well-linked_is_well-linked} is \emph{not} generally true for down-stitches since $\MM(\rho,\bar Y)$ may not be a well-linked link.
\end{remark}

Finally, we define the \emph{torso} and \emph{pieces} of an $\Omega$-knitwork.

\begin{definition}[Torso and Pieces]\label{def:torso_and_pieces}
    Let $\Omega$ be a well-quasi-order and $\GGG = (\bar G, \mu, \m, \Phi)$ an $\Omega$-knitwork with rooted Eulerian digraph $\bar{G}=(G,\pi,X_1,\ldots,X_{\ell})$ for some $\ell \geq 1$. 
    We define the \emph{torso of $\GGG$}---written $\torso(\GGG)$---to be the $\Omega$-knitwork obtained by iteratively up-stitching the proper rooted cuts $X_1,\ldots,X_\ell$, introducing up-stitch vertices $x_1^*,\ldots,x_\ell^*$; we refer to these as the \emph{up-stitch vertices of $\torso(\GGG)$}. We denote the underlying rooted Eulerian digraph of $\torso(\GGG)$ by $\torso(\bar G)$.
    Let $1\leq i\leq \ell$. We refer to $\PPP_i \coloneqq \stitch(\GGG;\pi,X_i)$ as the \emph{$i$-th piece of $\GGG$ with down-stitch vertex $x^i_*$}, and to $\{\PPP_i\mid 1\leq i \leq \ell\}$ as the \emph{pieces of $\GGG$}. Let $\bar P_i$ be the underlying rooted digraph of $\PPP_i$, then we refer to it as the \emph{$i$-th piece of $\bar G$}

\end{definition}
\begin{remark}
    Note that iterative stitching is well-defined by \crefthm{obs:stitching_fundamentals}{2}; in particular $\torso(\bar G)$ is a controlled rooted Eulerian digraph and $\m_{\torso{(\GGG)}}$ is not defined on $x_1^*,\ldots,x_\ell^*$.

   Further, note that $x_1^*$, say, may have all its edges with vertices in $\{x_2^*,\ldots,x_\ell^*\}$ if $\rho(X_1) \subseteq \bigcup_{i=2}^\ell \rho(X_i)$, and the vertex set of the torso may consist solely of $\{x_1^*,x_2^*,\ldots,x_\ell^*\}$.

   Finally note that the way $\Omega$-knitworks are defined, the pieces have no inherent $\Omega$-sleeve (it is only defined on the down-stitch vertex) and may be arbitrary quasi-Eulerian digraphs rooted in a single cut; this will be mitigated later and in future work with what we call \emph{locations}.
\end{remark}
See \cref{fig:torso-stitch} for an illustration.
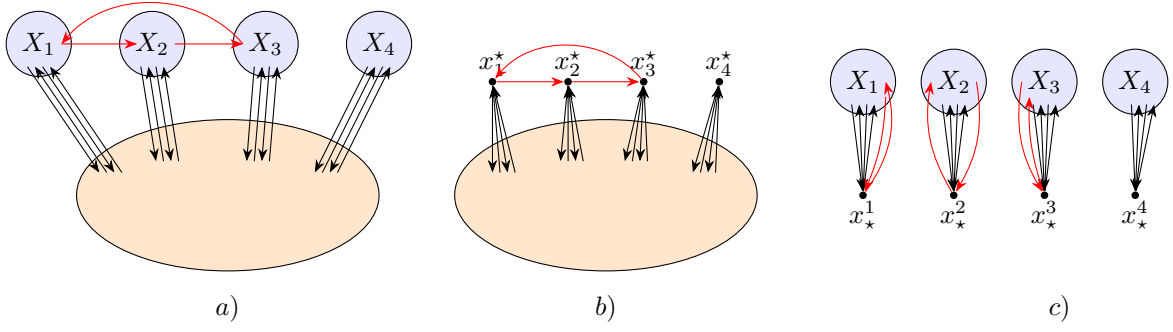
\begin{figure}
    \centering
    \begin{tikzpicture}[>=Stealth]
        \filldraw[fill=orange!20!white, draw=black, thin] (0,0) ellipse [x radius=2, y radius=1];
        \foreach \x in {1,2,3,4}
        {
        \node[circle,fill=blue!10!white,draw=black,thin] (X-\x) at (-4+1.5*\x,2) { $X_\x$ } ;
        };
        \foreach \x in {0,2}
        {
        \draw[<-] (-2.55+\x*0.1,1.7) to (-1.6+0.1*\x,0.3) ;
        \draw[->] (-2.65+\x*0.1,1.7) to (-1.7+0.1*\x,0.3) ;

        \draw[->] (-1.15+\x*0.1,1.7) to (-0.95+0.1*\x,0.45) ;
        \draw[<-] (-1.05+\x*0.1,1.7) to (-0.85+0.1*\x,0.45) ;
        \draw[->] (0.35+\x*0.1,1.7) to (0.25+0.1*\x,0.45) ;
        \draw[<-] (0.455+\x*0.1,1.7) to (0.35+0.1*\x,0.45) ;
        \draw[->] (1.85+\x*0.1,1.7) to (1.15+0.1*\x,0.3) ;
        \draw[<-] (1.95+\x*0.1,1.7) to (1.25+0.1*\x,0.3) ;

        };
        \draw[->,red] (-2.2,2) to (-1.2,2) ;
        \draw[<-, bend left=50,red] (-2.2,2) to (0.2,2) ;
        \draw[->,red] (-0.7,2) to (0.2,2) ;

        \begin{scope}[xshift=5cm]
        \filldraw[fill=orange!20!white, draw=black, thin] (0,0) ellipse [x radius=2, y radius=1];
        \foreach \x in {1,2,3,4}
        {
        \node[circle,fill=black, inner sep=1pt] (d-\x) at (-2.5+\x,1.5) {  } ;
        \node[inner sep=2pt, anchor=south] (nd-\x) at (-2.5+\x,1.5) { $x_\x^\star$ } ;
        };

        \foreach \x in {0,2}
        {
        \draw[->] (d-1) to (-1.5+0.1*\x,0.3) ;
        \draw[<-] (d-1) to (-1.4+0.1*\x,0.3) ;
        \draw[->] (d-2) to (-0.6+0.1*\x,0.45) ;
        \draw[<-] (d-2) to (-0.5+0.1*\x,0.45) ;
        \draw[->] (d-3) to (0.25+0.1*\x,0.45) ;
        \draw[<-] (d-3) to (0.35+0.1*\x,0.45) ;
        \draw[->] (d-4) to (1.15+0.1*\x,0.3) ;
        \draw[<-] (d-4) to (1.25+0.1*\x,0.3) ;
        };
        \draw[->,red] (d-1) to (d-2) ;
        \draw[->,red] (d-2) to (d-3) ;
        \draw[->,bend right=50,red] (d-3) to (d-1) ;

        \end{scope}
        \begin{scope}[xshift=11cm]
        \foreach \x in {1,2,3,4}
        {
        \node[circle,fill=blue!10!white,draw=black,thin] (x-\x) at (-3.8+1.2*\x,1.5) { $X_\x$ } ;
        \node[circle,fill=black, inner sep=1pt] (u-\x) at (-3.8+1.2*\x,0) { } ;
        \node[inner sep=2pt, anchor=north] (nu-\x) at (-3.8+1.2*+\x,0) { $x^\x_\star$} ;
        }
        \foreach \x in {0,2}
        {
        \draw[->] (-2.75+\x*0.1,1.2) to (u-1) ;
        \draw[<-] (-2.65+\x*0.1,1.2) to (u-1) ;
        \draw[->] (-1.55+\x*0.1,1.2) to (u-2) ;
        \draw[<-] (-1.45+\x*0.1,1.2) to (u-2) ;
        \draw[->] (-0.35+\x*0.1,1.2) to (u-3) ;
        \draw[<-] (-0.25+\x*0.1,1.2) to (u-3) ;
        \draw[->] (0.95+\x*0.1,1.2) to (u-4) ;
        \draw[<-] (1.05+\x*0.1,1.2) to (u-4) ;
        };
         \draw[->, bend left = 10,red] (-2.3,1.3) to (u-1) ;
        \draw[<-, bend left=25,red] (-2.3,1.5) to (u-1) ;
         \draw[<-, bend right=20,red] (u-2) to (-1.1,1.5) ;
         \draw[->, bend left=20,red] (u-2) to (-1.7,1.5) ;
         \draw[->, bend left=10,red] (u-3) to (-0.4,1.3) ;
         \draw[<-, bend left=20,red] (u-3) to (-0.5,1.5) ;

        \end{scope}
        \node at (0, -1.5) { $a)$ } ;
        \node at (5, -1.5) { $b)$ } ;
        \node at (11,-1.5) { $c)$ } ;
    \end{tikzpicture}
    \caption{An example of a Torso with respective Pieces. Figure $a)$ depicts a rooted Eulerian digraph $\bar G = (G, \pi, X_1, \dots, X_4)$. Figure $b)$ depicts its torso. Figure $c)$ depicts the $4$ down-stitches at the root-sets obtained from $\bar G$.}
    \label{fig:torso-stitch}
\end{figure}

The following two observations are imminent by repeated application of \cref{obs:stitching_fundamentals} and \cref{obs:stitching_fundamentals_knitworks}.

\begin{observation}\label{obs:torso}
     Let $\Omega$ be a well-quasi-order and $\GGG = (\bar G, \mu, \m, \Phi)$ an $\Omega$-knitwork with rooted Eulerian digraph $\bar{G}=(G,\pi,X_1,\ldots,X_{\ell})$ for some $\ell \geq 1$. Let $X \coloneqq \bigcup_{i=1}^\ell X_i$. Let $\torso(\GGG)$ be the torso of $\GGG$ with up-stitch vertices $x_1^*,\ldots,x_\ell^*$, underlying rooted Eulerian digraph $\torso(\bar G)=(G^*,\pi,x_1^*,\ldots,x_t^*)$ and $\Omega$-sleeve $(\mu_{\torso},\m_{\torso},\Phi_{\torso})$. Then the following hold.
     \begin{enumerate}
        \item $G^*$ is a Eulerian digraph where $V(G^*)$ is partitioned by $\bar X,\{x_1^*\},,\ldots,\{x_\ell^*\}$, in particular $V(G^*) \cap V(G) = \bar X$, and $G^*[\bar X] = G[X]$. Furthermore, $E(G^*) = E(G[X]) \cup \bigcup_{i=1}^\ell{\rho_G(X_i)}$.\label{obs:torso:1}
         \item It holds $\pi_{G^*}(x_i^*) = \pi_G(X_i)$ for every $1 \leq i \leq \ell$,\label{obs:torso:2}
         \item $\restr{\mu_{\torso}}{\bar X} = \mu$ and  and it is not specified on the up-stitch vertices.\label{obs:torso:3}
         \item $\restr{\m_{\torso}}{\bar X} = \m$ and it is not specified on the up-stitch vertices.\label{obs:torso:4}
         \item$\restr{\Phi_{\torso}}{\bar X} = \Phi$ and it is not specified on the up-stitch vertices.\label{obs:torso:5}
         \item If $\GGG$ is well-linked/inter-linked, then $\torso(\GGG)$ is well-linked/inter-linked.\label{obs:torso:6}
     \end{enumerate}
\end{observation}

\begin{observation}\label{obs:torso_and_pieces_fundamentals}
     Let  $\bar{G}=(G,\pi,X_1,\ldots,X_{\ell})$ be a rooted Eulerian digraph of index $\ell \geq 1$. Let $X \coloneqq \bigcup_{i=1}^\ell X_i$. Let $\bar G^* \coloneqq \torso(\bar G)$ be the torso of with up-stitch vertices $x_1^*,\ldots,x_\ell^*$. Let $\bar P_i$ be the $i$-th piece of $\bar G$ with down-stitch vertex $x^i_*$ for every $1 \leq i \leq \ell$. Then the following hold.
     \begin{enumerate}
       \item $V(G^*) \cap V^\circ(P_i) = \emptyset$ and $V^\circ(P_i) \cap V^\circ(P_j) = \emptyset$  for every pair of distinct $1 \leq i,j \leq \ell$. Further, $V(G) = \left(V(G^*) \cup \bigcup_{i =1}^\ell V^\circ(P_i)\right) \setminus \{x_1^*,x^*_1,\ldots,x_\ell^*,x^\ell_*\}$.\label{obs:torso_and_pieces_fundamentals:1}
       \item $E(G^*) \cap E(P_i) = \rho_G(x_i) = \rho_{G^*}(x_i^*)\cap \rho_{P_i}(x^i_*)$ and $E(P_i) \cap E(P_j) = \rho_G(X_i) \cap \rho_G(X_j)$ for every pair of distinct $1 \leq i,j \leq \ell$.\label{obs:torso_and_pieces_fundamentals:2}
     \end{enumerate}
\end{observation}

\subsection{Knitting}\label{subsec:knitting}
Finally we introduce \emph{knitting}---sort of the inverse operation to stitching---an operation that given two rooted Eulerian digraphs with distinguished ordered cuts produces a new rooted Eulerian digraph by ``knitting'' the graphs together at these cuts respecting their order. 

 The highlevel intuition behind knitting two rooted digraphs is as follows. Take two digraphs $G_1,G_2$ and two induced cuts $X_1,X_2$ of the same order in the respective digraphs. The idea is now to construct a new graph $G$ with vertex set $X_1\cup X_2$ that agrees with $G_1$ on $X_1$ and with $G_2$ on $X_2$ by identifying the edges $\rho_{G_1}(X_1)$ and $\rho_{G_2}(X_2)$\footnote{Intuitively speaking, knitting may be seen as the natural ``graph sum'' operation when dealing with digraphs, carvings and immersion, similar to ``clique-sums'' for undirected graphs, treewidth and minors.}. Clearly this does not create loops.

\begin{definition}[Knitting]\label{def:knitting_rooted_graphs}
   For $i \in\{ 1,2 \}$ let~$\bar{G}_i=(G_i,\pi_i,X_1^i,\ldots,X_{\ell_i}^i)$ be rooted Eulerian digraphs for some $\ell_i \geq 1$ and~$X_j^i \subset V(G_i)$ for every $1 \leq j \leq \ell_i$. Let $k \in 2\N$ and~$Y \subset V(G_1)$ induce a proper rooted~$k$-cut in~$G_1$. Let $X_2 \coloneqq X_j^2$  with $\delta(X_2) = k$ and $\pi_2(X_2)=(e_1^2,\ldots,e_k^2)$ the ordering of the respective cut for some $1 \leq j \leq \ell_2$. Let~$\pi_{G_1}(Y) \coloneqq (e_1^1,\ldots,e^1_k)$ be an ordering of~$\rho(Y)$ such that~$e_i^1 \in \rho^+(Y) \iff e_i^2 \in \rho^-(\bar X_2)$ for every~$1 \leq i \leq k$.

   We define~$\knit(G_2,G_1;\pi_2,X_2,\pi_1,Y) \coloneqq (V',E',\operatorname{inc}')$ via
   \begin{itemize}
       \item $ V' \coloneqq Y \cup \bar X_2,$
        \item $ E' \coloneqq E(G_1[Y]) \cup E(G_2[\bar X_2]) \cup \{e_1,\ldots,e_k\},$
        \item $(e,v) \in \operatorname{inc}' :\iff 
            \begin{cases} 
                &(e,v) \in \operatorname{inc}_{G_1} \text{ and } e \in E(G_1[Y]), v\in Y,\\
                &(e,v) \in \operatorname{inc}_{G_2} \text{ and } e \in E(G_2[X_2]), v\in \bar X_2,\\
                &e= e_i,\ v\in Y \text{ and } (e_i^1,v) \in \operatorname{inc}_{G_1},\\
                &e= e_i,\ v\in \bar X_2 \text{ and } (e_i^2,v) \in \operatorname{inc}_{G_2}
                \end{cases}$
         \item $(v,e) \in \operatorname{inc}' :\iff 
            \begin{cases} 
                &(v,e) \in \operatorname{inc}_{G_1} \text{ and } e \in E(G_1[Y]), v\in Y,\\
                &(v,e) \in \operatorname{inc}_{G_2} \text{ and } e \in E(G_2[{\bar X_2}]), v\in {\bar X_2},\\
                &e= e_i,\ v\in Y \text{ and } (v,e_i^1) \in \operatorname{inc}_{G_1},\\
                &e= e_i,\ v\in \bar X_2 \text{ and } (v, e_i^2) \in \operatorname{inc}_{G_2}.
             \end{cases}$
   \end{itemize}
    Further, let $\{X_1^{\insc},\ldots,X_\tau^{\insc}\} = \insc(Y)$ for $\tau \geq 1$. After renaming $e_i^1,e_i^2$ to $e_i$ in $G_1,G_2$ respectively for every $1 \leq i \leq k$ it holds $\rho_{G_1}(X_t^{\insc}) \subseteq E'$ for every $1 \leq t \leq \tau$ as well as $\rho_{G_2}(X_r^2) \subseteq E'$ for $1\leq r \leq \ell_2$. Thus, we may unambiguously set $\pi(X_p^{\insc}) \coloneqq \pi_1(X_p^{\insc})$ for every $1 \leq p \leq \tau$ and $\pi(X_q^2) \coloneqq \pi_2(X_q^2)$ for every $q\in \{1,\ldots,\ell_2\}\setminus\{j\}$, and define~$$\knit(\bar{G_2},\bar{G_1};\pi_2,X_2,\pi_1,Y) \coloneqq ((V',E',\operatorname{inc}'),\pi, X_1^{\insc},\ldots,X_\tau^{\insc},X_1^2,\ldots,X_{j-1}^2,X_{j+1}^2,\ldots,X_{\ell_2}^2).$$

We say that~$\knit(\bar{G_2},\bar{G_1};\pi_2,X_2,\pi_1,Y) $ is obtained by \emph{knitting~$\bar G_2$ along $\pi_{2}(X_2)$ to~$\bar G_1$ along~$\pi_{1}(Y)$}.
\end{definition}

Again, the following observations, although a bit tedious, are straightforward to verify from the definition.

\begin{observation}\label{obs:knitting_fundamentals}
    Let~$\bar{G}_i=(G_i,\pi_i,X_1^i,\ldots,X_{\ell_i}^i)$, $X_2$, and $Y$ be as in \cref{def:knitting_rooted_graphs}. Then
    \begin{enumerate}
        \item $\bar G' = \knit(\bar{G_2},\bar{G_1};\pi_2,X_2,\pi_1,Y)$ is a well-defined rooted Eulerian digraph, and if $ G_1, G_2$ are both loopless, then $ G'$ is loopless,\label{obs:knitting_fundamentals:1}
         \item the sets~$\bar X_2,Y \subset V(G')$  partition~$V(G')$ whence they induce the same cut~$\rho_{G'}(\bar X_2) = \rho_{G'}(Y)$ of size~$k\in 2\N$, \label{obs:knitting_fundamentals:2}
        \item $G'[Y] = G_1[Y]$ and~$G'[\bar X_2] = G_2[\bar X_2]$, and\label{obs:knitting_fundamentals:3}
        \item $\stitch(\bar{G_1};\pi_{G_1},\tilde X) \cong \stitch(\bar{G'};\pi_{G'},\tilde X)$ for every $\tilde X \in \insc(Y)$ as well as $\stitch(\bar{G_2};\pi_{G_1},\tilde X) \cong \stitch(\bar{G'};\pi_{G'},\tilde X)$ for every $\tilde X \in \{X_1^2,\ldots,X_{\ell_2}^2\}\setminus \{X_2\}$.\label{obs:knitting_fundamentals:4}
    \end{enumerate}
\end{observation}

The definition of knitting readily extends to~$\Omega$-knitworks.

\begin{definition}\label{def:knitting_knitworks}
     Let~$\Omega$ be a well-quasi-order and let~$\GGG_i = (\bar{G}_i,\mu_i,\m_i,\Phi_i)$ be~$\Omega$-knitworks with rooted Eulerian digraphs~$\bar{G_i}=(G_i,\pi_i,X_1^i,\ldots,X_{\ell_i}^i)$ for $\ell_i \geq 1$ and~$i \in \{1,2\}$. Let $X_2 \coloneqq X^2_j$ for some $1 \leq j \leq \ell_2$ with~$\pi_{2}(X_2) = (e^2_1,\ldots,e^2_k)$ for some~$k \in 2\N$. Let~$Y \subset V(G_1)$ induce a proper rooted~$k$-cut in~$G_1$. Let~$\pi_{G_1}(Y) \coloneqq (e_1^1,\ldots,e^1_k)$ be an ordering of~$\rho(Y)$ such that~$e_i^1 \in \rho^+(Y) \iff e_i^2 \in \rho^-(\bar X_2)$ for every~$i\in\{1,\ldots,k\}$. Let~$\bar G' \coloneqq \knit(\bar G_2, \bar G_1;\pi_2,X_2,\pi_1,Y)$. We define 
     \begin{itemize}
         \item $\mu' \coloneqq \begin{cases}
             \mu_1(x) \text{ renaming } e_i^1 \text{ to } e_i, &\text{ if } x \in Y\\
             \mu_2(x) \text{ renaming } e_i^2 \text{ to } e_i, &\text{ if } x \in \bar X_2\\
         \end{cases}$,
         
           \item $\m' \coloneqq \begin{cases}
             \m_1(x) \text{ renaming } e_i^1 \text{ to } e_i, &\text{ if } x \in Y\\
             \m_2(x) \text{ renaming } e_i^2 \text{ to } e_i, &\text{ if } x \in \bar X_2\\
         \end{cases}$,  and 
         
         \item $\Phi'(x) \coloneqq \begin{cases}
             \Phi_1(x), &\text{ if } x \in Y\\
             \Phi_2(x), &\text{ if } x\in \bar X_2.
         \end{cases}$
     \end{itemize}

  Finally we define~$\GGG'\coloneqq \knit(\GGG_2,\GGG_1;\pi_2,X_2,\pi_{1},Y) \coloneqq \big(\bar G', \mu', \m', \Phi'\big)$ and say that~$\GGG'$ is obtained by \emph{knitting~$\GGG_2$ along $\pi_2(X_2)$ to~$\GGG_1$ along~$\pi_{1}(Y)$}.
\end{definition}

Again, using the above definition and \cref{obs:knitting_fundamentals} the following is straightforward.

\begin{observation}\label{obs:knitting_fundamentals_knitworks}
      Let~$\Omega$ be a well-quasi-order. Let~$\GGG_i = (\bar{G}_i,\mu_i,\m_i,\Phi_i)$, $X_2$ and $Y$ be as in \cref{def:knitting_knitworks}.  Let~$\GGG' \coloneqq \knit(\GGG_2,\GGG_1;\pi_2,X_2,\pi_{1},Y) =\big(\bar G', \mu',\m', \Phi'\big)$. Then
     \begin{enumerate}
        \item  ~$\GGG'$ is a well-defined~$\Omega$-knitwork,
         \item $\mu'(x) = \mu_1(x)$ for all~$x \in Y$ and~$\mu'(x) = \mu_2(x)$ for all~$x \in \bar X_2$,
         \item $\m'(x) = \m_1(x)$ for all~$x \in Y$ and~$\m'(x) = \m_2(x)$ for all~$x \in \bar X_2$, and
         
         \item $\Phi'(x) = \Phi_1(x)$ for all~$x \in Y$ and~$\Phi'(x) = \Phi_2(x)$ for all~$x \in \bar X_2$.
     \end{enumerate}
\end{observation}

\section{Decomposition Theorems for Knitworks}\label{sec:decomposition_theorem}
In this section we will prove important theorems that will allow us to apply induction on the structure of Eulerian digraphs in the hopes to tackle a proof of \cref{conj:wqo_dreg,conj:wqo_gen} and related well-quasi-ordering results. Most notably, we prove two decomposition theorems in \cref{subsec:knitting_immersions,subsec:decomp_and_altering_kitworks}. In the last \cref{subsec:manipulate_knitworks} we present two more results that allow us to manipulate knitworks, one of which justifies why we may restrict to loopless digraphs. All of the following results hold true for stable and non-stable $\Omega$-knitwork immersion as will become apparent.

\subsection{Knitting strong immersions}
\label{subsec:knitting_immersions}
The definitions of stitching and knitting allow us to unambiguously decompose~$\Omega$-knitworks along proper rooted cuts and knit them back together in a unique and reversible way (up to isomorphisms). 
\begin{lemma}[Stitch-and-Knit]\label{lem:stitch-and-knit}
    Let~$\Omega$ be a well-quasi-order. Let~$\bar G$ be a loopless rooted Eulerian digraph. Let~$Y \subset V(G)$ induce a proper rooted~$k$-cut in~$\bar G$ for some~$k\in 2\N$ and let~$\pi_G(Y)=(e_1,\ldots,e_k)=\pi_G(\bar Y)$ be an ordering of~$\rho(Y)$. Let
    \begin{align*}
         &\text{$\bar G_Y \coloneqq \stitch(\bar G;\pi_G,Y)$ with down-stitch vertex $y_*$, and}\\
        &\text{$ \bar G^Y \coloneqq \stitch(\bar G;\pi_G,\bar{Y})$ with up-stitch vertex $y^*$.}
    \end{align*}
Let $(\mu_Y,\m_Y,\Phi_Y)$ be an $\Omega$-sleeve for $\bar G_Y$ and let $(\mu^Y,\m^Y,\Phi^Y)$ be an $\Omega$-sleeve for $G^Y$ and refer to the respective $\Omega$-knitwork as $\HHH_Y=(\bar G_Y,\mu_Y,\m_Y,\Phi_Y)$ and $\HHH^Y = (\bar G^Y, \mu^Y, \m^Y, \Phi^Y)$. 

Then, fixing the ordering~$\pi_{G_Y}(Y) =(e_1,\ldots,e_k)$ of $\rho_G(Y)$, it holds that $$\HHH \coloneqq \knit( \HHH^Y, \HHH_Y; \pi_{G^Y},y^*,\pi_{G_Y},Y) $$
is an $\Omega$-knitwork satisfying $H = G$ and $\restr{\mu_H}{Y} = \mu_Y$ as well as $\restr{\mu_H}{\bar Y} = \mu^Y$, and analogously for $\m_H$ and $\Phi_H$.
\end{lemma}
\begin{proof}

     By \crefthm{obs:stitching_fundamentals}{3} we have~$\rho_{G_Y}(Y) = \rho_G(Y) = \{e_1,\ldots,e_k\}$, whence the choice~$\pi_{G_Y}(Y) = \pi_G(Y)$ is well-defined. By \crefthm{obs:stitching_fundamentals}{4} note that~$\pi_{G}(\bar Y) = \pi_{G^Y}(y^*)$ and by assumption $\pi_G(\bar Y) = \pi_G(Y)$. In particular we have~$e_i \in \rho_{G_Y}^+(Y) \iff e_i \in \rho_{G^Y}^-(\bar Y)$ for~$i\in\{1,\ldots,k\}$, whence $\knit(\HHH^Y,\HHH_Y; \pi_{G^Y},y^*,\pi_{G_Y},Y)$ is well defined using \cref{def:knitting_rooted_graphs}. Note here that knitting loses the information of~$\mu_y,\m_Y$ on~$y_\ast$, and of ~$\mu^Y,\m^Y$ on~$y^\ast$ respectively (here they are not defined by definition of up-stitching), for they are not part of the resulting graph. 
    
    Let~$ \knit(\HHH^Y,\HHH_Y; \pi_{G^Y},y^*,\pi_{G_Y},{Y}) = \big(\knit(\bar G^Y,\bar G_Y; \pi_{G^Y},y^\ast,\pi_{G_Y},{Y}), \mu_H,\m_H,\Phi_H\big)$ with a respective $\Omega$-sleeve~$(\mu_H,\m_H,\Phi_H)$ as given by \cref{def:knitting_knitworks}.
    \begin{claim}
        $ \bar{G} = \knit(\bar G^Y,\bar G_Y; \pi_{G^Y},y^\ast,\pi_{G_Y},{Y})$.
    \end{claim}
    \begin{claimproof}
        Let~$\bar H = \knit(\bar G^Y,\bar G_Y; \pi_{G^Y},y^\ast,\pi_{G_Y},{Y})$. Then by Part 1 and 3 of \cref{obs:knitting_fundamentals} $\bar H$ is a well-defined rooted Eulerian digraph satisfying~$H[Y] = G_Y[Y]$ as well as~$H[\bar{Y}] = G^Y[\bar{Y}]$.
        
        By \crefthm{obs:stitching_fundamentals}{2}~$Y$ induces a proper rooted~$k$-cut in~$\bar G_Y$ and~$\bar{Y}$ induces a (not necessarily proper) rooted~$k$-cut in~$\bar G^Y$ such that~$G_Y[Y] = G[Y]$ and~$G^Y[\bar{Y}] = G[\bar Y ]$. In particular then~$H[Y] = G[Y]$ and $H[\bar{Y}] = G[\bar{Y}]$ by \cref{def:knitting_rooted_graphs}.

        By 3. of \cref{obs:stitching_fundamentals} together with our choice of~$\pi_{G_Y}(Y)$ we derive $\pi_{G^Y}(y^\ast) = (e_1,\ldots,e_k) = \pi_{G_Y}(Y) = \pi_G(Y)$. Thus, by \cref{def:knitting_rooted_graphs} we get~$\rho_H(Y) = \rho_G(Y)$ respecting the same incidences; the claim follows as we did not alter any other incidences of root edges of $\bar G$ during stitching (or knitting). 
    \end{claimproof}
    Recall the \cref{def:knitting_knitworks} of knitting $\Omega$-knitworks. Combining Parts 2, 3, and 4 of \cref{obs:stitching_fundamentals_knitworks} with Parts 2, 3, and 4 of \cref{obs:knitting_fundamentals_knitworks} we immediately get~$\mu_H(x) = \mu_{G_Y}(x)$,~$\m_H(x) = \m_{G_Y}(x)$ and~$\Phi_H(x) = \Phi_{G_Y}(x)$ for all~$x \in Y$, and analogously with $G^Y$ and $\bar Y$. This concludes the proof.
\end{proof}

In particular we derive the following.
\begin{corollary}\label{cor:stitch-and-knit}
    Let~$\Omega$ be a well-quasi-order. Let~$\GGG$ be an $\Omega$-knitwork. Let~$Y \subset V(G)$ induce a proper rooted~$k$-cut in~$\GGG$ for some~$k\in 2\N$ and let~$\pi_G(Y)=(e_1,\ldots,e_k)=\pi_G(\bar Y)$ be an ordering of~$\rho(Y)$. Let
    \begin{align*}
         &\text{$\GGG_Y \coloneqq \stitch(\GGG;\pi_G,Y)$ with down-stitch vertex $y_*$, and}\\
        &\text{$ \GGG^Y \coloneqq \stitch(\GGG;\pi_G,\bar{Y})$ with up-stitch vertex $y^*$.}
    \end{align*}
Then, fixing the ordering~$\pi_{G_Y}(Y) = (e_1,\ldots,e_k) = \pi_{G^Y}(y^*)$ of $\rho_G(Y)$, it holds that $$\knit( \GGG^Y, \GGG_Y; \pi_{G^Y},y^*,\pi_{G_Y},Y) = \GGG.$$
\end{corollary}
\begin{remark}
  Note that the corollary holds with equality instead of isomorphism due to the intrinsic identifications in all the definitions (renaming edges).
\end{remark}

Finally, we also need the following technical variation of \cref{lem:stitch-and-knit} regarding sub-knitworks.
\begin{corollary}\label{cor:stitch-and-knit_subknitwork}
    Let~$\Omega$ be a well-quasi-order. Let~$\GGG$ be an $\Omega$-knitwork. Let~$Y \subset V(G)$ induce a proper rooted~$k$-cut in~$\GGG$ for some~$k\in 2\N$ and let~$\pi_G(Y)=(e_1,\ldots,e_k)=\pi_G(\bar Y)$ be an ordering of~$\rho(Y)$. Let
    \begin{align*}
         &\text{$\GGG_Y \coloneqq \stitch(\GGG;\pi_G,Y)$ with down-stitch vertex $y_*$, and}\\
        &\text{$ \GGG^Y \coloneqq \stitch(\GGG;\pi_G,\bar{Y})$ with up-stitch vertex $y^*$.}
    \end{align*}

Let $\HHH_Y \subseteq \GGG^Y$ such that $y_* \in V(H_Y)$ and $\rho_{G_Y}(y_*) = \rho_{H_Y}(y_*)$. Let $Y' \coloneqq V(H_Y) \setminus \{y_*\}$. Let $\HHH^Y \subseteq \GGG^Y$, whence by definition $y^* \in V(H^Y)$ and $\rho_{G^Y}(y^*) = \rho_{H^Y}(y_*)$. Then, fixing the ordering~$\pi_{H_Y}(Y') = (e_1,\ldots,e_k) = \pi_{H^Y}(y^*)$ of $\rho_G(Y)$, it holds that $$\HHH \coloneqq \knit( \HHH^Y, \HHH_Y; \pi_{H^Y},y^*,\pi_{H_Y},Y') \subseteq \GGG.$$
\end{corollary}
\begin{proof}
    By \cref{def:knitting_knitworks} $\HHH$ is a well-defined $\Omega$-knitwork, and, similar to \cref{lem:stitch-and-knit} one easily verifies that $\HHH \subseteq \GGG$. This follows essentially from \cref{lem:stitch-and-knit} noting that $\HHH \subseteq \GGG$ can be obtained by deleting the Eulerian digraphs $\CCC_Y \coloneqq G_Y - H_Y$ as well as $\CCC^Y \coloneqq G^Y - H^Y$ as guaranteed by \cref{lem:subknitwork_difference_eulerian} and adapting the $\Omega$-sleeve accordingly.
\end{proof}

Similar to the Stitch-and-Knit \cref{lem:stitch-and-knit}, stitching and knitting can be used to reduce the question of whether a graph~$H$ strongly immerses into~$G$ by decomposing the two graphs at cuts via stitches and asking whether the resulting stitches can be immersed into each other; if so, we can knit the immersions back together to yield an immersion of~$H$ in~$G$. 

Before we state this formally, we observe the following which is imminent from the \cref{def:knitwork_immersion}.
\begin{lemma}\label{obs:knitwork_immersion_on_mu}
    Let~$\Omega$ be a well-quasi-order, let~$\GGG = (\bar{G},\mu,\m,\Phi)$ and~$\GGG'=(\bar{G}',\mu',\m',\Phi')$ be~$\Omega$-knitworks and let~$\gamma: \GGG\hookrightarrow^* \GGG'$ be an~$\Omega$-knitwork immersion. Let~$v\in \dom(\mu)$ and let~$(e_1,\ldots,e_k) = \mu(v)$. Further let~$(e_1',\ldots,e_k') = \mu(v')$ for~$v' = \gamma(v)$, in particular $\deg_{G}(v) = \deg_{G'}(v')$.
    Then~$e_i' \in \rho^+(v') \iff e_i \in \rho^+(v)$; in particular $e_i'$ is an end of $\gamma(e_i)$.
\end{lemma}
\begin{proof}
    By \cref{def:knitwork_immersion} of~$\Omega$-knitwork immersion~$\gamma$ respects~$\mu$ and $\mu'$, and since $v$ and $v'$ are of equal degree,  we derive that~$e_i' \in \gamma(e_i)$ for all~$1 \leq i \leq k$. Since~$\gamma$ is in particular an immersion of~$G$ into~$G'$, the edge~$e_i=(v,w) \in \rho^+(v)$ for some~$i \in \{1,\ldots,k\}$ is mapped to a path starting in~$v' = \gamma(v)$. If the path~$\gamma(e_i)$ does not start in~$e_i'$ then it must start in some other edge~$f \in N_{G'}(v')$, but clearly~$f$ cannot be a loop since we do not allow for loops; thus~$f = e_j$ for some~$j \in \{1,\ldots,k\}$ with~$i \neq k$. But then~$e_j  \in \gamma(e_i') \cap \gamma(e_j')$; a contradiction to the paths being edge-disjoint using the fact that~$\gamma$ is an immersion of rooted Eulerian digraphs.
\end{proof}

Complementing the above, note that due to the \cref{def:knitwork_immersion}, more specifically the fact that immersions need to respect $\mu$ and $\mu'$, no edge $e_i \in \rho(v)$ can be mapped to a path that visits $v$ twice, regardless whether or not the immersion is strong.

The following is the main result of this subsection. 

\begin{theorem}\label{thm:knitting_knitwork_immersion}
    Let~$\Omega=(V(\Omega),\preceq)$ be a well-quasi-order. Let~$\GGG=((G,\pi,X_1,\ldots,X_\ell),\mu,\m,\Phi)$ and~$\HHH=((H,\pi,A_1,\ldots,A_\ell),\nu,\n,\Psi)$ be~$\Omega$-knitworks of common index~$\ell \in 2\N$. Let~$Y \subset V(G)$ and~$B \subset V(H)$ induce proper rooted~$k$-cuts in~$\bar{G}$ and~$\bar{H}$ respectively for some~$k \in 2\N$ such that $\Abs{\insc(Y)} = \Abs{\insc(B)}$. Let~$\pi(Y)=(e_1^Y,\ldots,e_k^Y)=\pi(\bar Y)$ and~$\pi(B)=(e_1^B,\ldots,e_k^B) =\pi(\bar{B})$ be orderings of~$\rho_G(Y),\rho_H(B)$ respectively. Let
    \begin{align*}
        \GGG_Y &=  \stitch(\GGG;\pi,Y) \text{ with down-stitch vertex } y_\ast,\\
       \GGG^Y &=\stitch(\GGG;\pi,\bar{Y}) \text{ with up-stitch vertex } y^\ast,\\
        \HHH_B &=  \stitch(\HHH;\pi,B) \text{ with down-stitch vertex } b_\ast,\\
        \HHH^B &=\stitch(\HHH;\pi,\bar{B}) \text{ with up-stitch vertex } b^\ast, 
    \end{align*}

    respectively. Further let
    \begin{align*}
        \gamma_d:&  \HHH_B \hookrightarrow \GGG_Y, \text{ with } \gamma_d(b_\ast) = y_\ast\quad  \text{ and } \quad
        \gamma_u:  \HHH^B \hookrightarrow \GGG^Y , \text{ with } \gamma_u(b^\ast) = y^\ast,
    \end{align*}be strong~$\Omega$-knitwork immersions. 
    
    Then there exists a strong~$\Omega$-knitwork immersion~$\gamma:\HHH \hookrightarrow \GGG$ such that~$$\restr{\gamma}{H[B]} = \restr{\gamma_u}{H_B[B]} \quad \text{ and } \quad\restr{\gamma}{H[\bar{B}]} = \restr{\gamma_d}{H^B[\bar{B}]}.$$ The same holds true for weak immersion $\hookrightarrow^*$.
\end{theorem}
\begin{proof}
We prove the theorem for strong $\Omega$-knitwork immersion; the other case is analogous (we highlight the relevant differences where needed). First, note that it is enough to discuss \emph{stable} $\Omega$-knitwork immersion, as guaranteed by the following.
\begin{claim}\label{thm:knitting_knitwork_immersion_claim_stable}
    If the theorem holds for stable (strong) $\Omega$-knitwork immersion, then it holds for (strong) $\Omega$-knitwork immersion.
\end{claim}
\begin{claimproof}
Again, we give a proof for strong immersion, for the other case follows verbatim. Let $\HHH,\GGG$ as well as $\gamma_d:\HHH_B \hookrightarrow \GGG_Y$ and $\gamma_u : \HHH^B \hookrightarrow \GGG^Y$ as in the theorem. We are to show that $\gamma:\HHH \hookrightarrow \GGG$. Let $\GGG_d \subseteq \GGG_Y$ be the witness for $\gamma_d$ and $\GGG_u \subseteq \GGG^Y$ be the witness for $\gamma_u$, in particular $\gamma_d: \HHH_B \hookrightarrow \GGG_d$ and $\gamma_u:\HHH^Y \hookrightarrow \GGG_u$ are stable, and by assumption of the theorem we derive that $y_* \in V(G_d)$ and $y^* \in V(G_u)$. Since $\delta_G(Y) = \delta_H(B)$ we further derive that $\rho_{G_u}(y^*) = \rho_{G^Y}(y^*)$ as well as $\rho_{G_d}(y_*) = \rho_{G_Y}(y_*)$. Let $Y' \coloneqq V(G_d)\setminus \{y_*\}$, then by the above $\GGG' \coloneqq \knit(\GGG_u,\GGG_d; \pi_{\GGG_u},y^*,\pi_{\GGG_d},Y')$ is a well-defined $\Omega$-knitwork and by \cref{cor:stitch-and-knit_subknitwork} we have $\GGG' \subseteq \GGG$. By \cref{cor:stitch-and-knit} we derive that $\stitch(\GGG';\pi_{\GGG'},Y') = \GGG_d$ as well as  $\stitch(\GGG';\pi_{\GGG'},\bar Y') = \GGG_u$. Thus, by assumption of the claim, applying the theorem to $\GGG',\HHH$ as well as the stable strong $\Omega$-knitwork immersions $\gamma_d \colon \HHH_Y \hookrightarrow \GGG_d$ and $\gamma_u \colon \HHH^Y \hookrightarrow \GGG_u$, we derive the existence of a stable strong $\Omega$-knitwork immersion $\gamma: \HHH \hookrightarrow \GGG'$. Since $\GGG' \subseteq \GGG$, \cref{def:knitwork_immersion} implies that $\gamma: \HHH \hookrightarrow \GGG$ as desired, concluding the proof.
\end{claimproof}

Throughout the rest of the proof we assume that the immersions $\gamma_u$ and $\gamma^d$ are stable without further mention, implicitly using \cref{thm:knitting_knitwork_immersion_claim_stable}. In particular both immersions satisfy $(1)$-$(6)$ of \cref{def:knitwork_immersion}.

    \begin{claim}\label{claim:stitch_knit_corr_paths_under_partial_immersion_uno}
        Let~$e_i^Y \in \rho_G^+(Y)$ for some $i \in \{1,\ldots,k\}$. Then~$\gamma_d(e_i^B) \subset G_Y$ is a path ending in~$e_i^Y \in \rho_{G_Y}^+(Y)$ that is otherwise disjoint from~$\rho_{G_Y}(Y)$ and~$\gamma_u(e_i^B) \subset G^Y$ is a path starting in~$e_i^Y \in \rho_{G^Y}^-(\bar{Y})$ that is otherwise disjoint from~$\rho_{G^Y}(\bar{Y})$.
    \end{claim}
    \begin{claimproof}
       We start with a proof for~$\gamma_d$. To this extent note that by \crefthm{obs:stitching_fundamentals}{3}~$\rho_{H_B}(b_\ast) = \{e_1^B,\ldots,e_k^B\} = \rho_H(B)$ and~$\rho_{G_Y}(y_\ast) = \{e_1^Y,\ldots,e_k^Y\} = \rho_G(Y)$ and by \crefthm{obs:stitching_fundamentals_knitworks}{3} we know that~$\nu_B(b_\ast) = (e_1^B,\ldots,e_k^B)$ as well as~$\mu_Y(y_\ast) = (e_1^Y,\ldots,e_k^Y)$. By the assumption of the theorem we know that~$\gamma_d(b_\ast) = y_\ast$ and since~$\gamma_d$ is a stable strong~$\Omega$-knitwork immersion we derive from \crefdef{def:knitwork_immersion}{3} that for every~$1 \leq i \leq k$ the path~$\gamma_d(e_i^B)$ contains the edge~$e_i^Y$. Since all these paths are edge-disjoint the path~$\gamma_d(e_i^B)$ contains no other edge of~$\rho_{G_Y}(y_\ast)$. Further, by \cref{obs:knitwork_immersion_on_mu} (note that the lemma is true for strong and weak immersion, and its assumptions hold true since we assume the immersions to be stable), we derive that $e_i^B \in \rho_{H_B}(b_\ast)^- \iff e_i^Y \in \rho_{G_Y}(y_\ast)^-$. Further, using the fact that~$\rho_{G}^+(Y) = \rho_{G_Y}^-(y_\ast)$ as well as~$\rho_{H}^+(B) = \rho_{H_B}^-(b_\ast)$ (and analogously for switched signs), the path $\gamma_d(e_i^B)$ must end in~$e_i^Y$---for it uses exactly one edge incident to $y_\ast$---concluding the proof for this case.

       The proof for~$\gamma_u$ follows verbatim using~$\gamma_u(b^\ast) = y^\ast$. Note that by \cref{def:stitching_knitwork}, $y^\ast \notin \dom(\mu^Y)$ as well as $b^\ast \notin \dom(\nu^B)$, so we use $\pi_{G^Y}(y^\ast)$ as well as $\pi_{G^B}(b_\ast)$ in the proof instead, together with the \cref{def:rooted_immersion} of immersion (the arguments are identical). The only difference is that the path $\gamma_u(e_i^B)$ starts in the respective edge~$e_i^Y \in \rho^-(\bar{Y})$ by symmetry, using that $e_i^Y \in \rho_G^+(Y) \iff e_i^Y \in \rho_G^-(\bar Y)$.  
    \end{claimproof}

    Analogously by switching signs and using the symmetry for cut edges we get the following.
     \begin{claim}\label{claim:stitch_knit_corr_paths_under_partial_immersion_dos}
        Let~$e_i^Y \in \rho_G^-(Y)$ for some $i \in \{1,\ldots,k\}$. Then~$\gamma_d(e_i^B) \subset G_Y$ is a path starting in~$e_i^Y \in \rho_{G_Y}^-(Y)$  that is otherwise edge-disjoint from~$\rho_{G_Y}(Y)$ and~$\gamma_u(e_i^B) \subset G^Y$ is a path ending in~$e_i^Y \in \rho_{G^Y}^-(\bar{Y})$ that is otherwise edge-disjoint from~$\rho_{G^Y}(\bar{Y})$.
    \end{claim}
    For weak immersion both claims hold analogously.

    Combining Claims \ref{claim:stitch_knit_corr_paths_under_partial_immersion_uno} and \ref{claim:stitch_knit_corr_paths_under_partial_immersion_dos} together with \crefthm{obs:stitching_fundamentals}{2} we derive the following.
    \begin{claim}\label{claim:stitch_knit_effect_on_cut}
        For every~$1 \leq i \leq k$,~$e_i^Y$ is an end---once first and once last edge---to both otherwise vertex-disjoint paths~$\gamma_u(e_i^B) \subset G^Y$ and~$\gamma_d(e_i^B)\subset G_Y$; in particular $e_i^Y \in E(\gamma_u(e_i^B))\cap E(\gamma_d(e_i^B))$ and $V^\circ(\gamma_u(e_i^B))\cap V^\circ(\gamma_d(e_i^B)) = \emptyset$. Further~$e_i^Y \in \gamma_u(e)$ or $e_i^Y \in \gamma_d(e)$ if and only if~$e = e_i^B$.
    \end{claim}
    The same claim holds for weak immersion; in particular the paths are also vertex-disjoint.

    Finally, we use \cref{claim:stitch_knit_effect_on_cut} to ``knit'' both immersions~$\gamma_u,\gamma_d$ when knitting the stitched graphs back together. We define a strong immersion~$\gamma: \HHH \hookrightarrow \GGG$ as follows. Recall \cref{def:paths} and \cref{def:concatenation_of_paths}.

    \begin{align*}
        &\gamma\colon V(H)\cup E(H) \to G,\\
        &\restr{\gamma}{H[B]} \coloneqq \restr{\gamma_d}{H_B[B]},\\
        &\restr{\gamma}{H[\bar{B}]} \coloneqq \restr{\gamma_u}{H^B[\bar{B}]}, \text{ and }\\
        &\gamma(e_i^B) = \begin{cases}
            \gamma_u(e_i^B) \circ \gamma_d(e_i^B),& e_i^B \in \rho_H^+(B)\\
            \gamma_d(e_i^B) \circ \gamma_u(e_i^B),& e_i^B \in \rho_H^-(B).
        \end{cases}
    \end{align*}
    By construction the above is well-defined. To see this, it suffices to show that~$ \gamma_u(e_i^B) \circ \gamma_d(e_i^B)$ and~$\gamma_d(e_i^B) \circ \gamma_u(e_i^B)$ are paths in~$G$ for $e_i^B \in \rho_H^+(B)$ and $e_i^B \in \rho_H^-(B)$ respectively. By \cref{claim:stitch_knit_effect_on_cut}, the respective summands are internally edge-disjoint, and for~$\gamma_u(e_i^B)=(f_1,\ldots,f_\ell)$ with~$f_1,\ldots,f_\ell \in E(G_Y) \cap E(G)$ it holds that~$(f_1,\ldots,f_\ell)_G$ is a path in~$G$, in particular~$\gamma_u(e_i^B)$ is. Analogously, all summands are paths in~$G$ and the claim follows by \cref{claim:stitch_knit_effect_on_cut} and \cref{obs:concat_paths}. 
    
    It is now straightforward to verify that~$\gamma$ is a strong rooted immersion.
\begin{claim}
        $\gamma:\bar H \hookrightarrow \bar G$ is a strong rooted immersion. 
    \end{claim}
    \begin{claimproof}
        The fact that~$\gamma: V(H) \to V(G)$ is injective is clear by construction together with the fact that~$\gamma_u,\gamma_d$ are strong immersions for two disjoint sets~$B,\bar{B}$ where~$V(H) = B \cup \bar{B}$ and $V(G) = Y \cup \bar Y$ by \crefthm{obs:knitting_fundamentals}{2}. Further, since $E(H[B]) \cap E(H[\bar{B}]) = \emptyset$, it follows by definition of $\gamma$ together with the \cref{def:immersion} of strong immersion, and \cref{claim:stitch_knit_effect_on_cut}, that every~$e \in E(H) \setminus \rho(B)$ is mapped to a unique path~$\gamma(e)$ in~$G$. In addition, the respective paths are edge-disjoint where no such path contains any vertex of~$\gamma(V(G))$ as an internal vertex---since both linkages $\gamma_u(E(H_B))$ and $\gamma_d(E(H^B))$ are strong---and no path contains an edge of~$\rho(Y)$. In particular~$\gamma(e)$ is either completely contained in~$G[Y]$ or in~$G[\bar{Y}]$. Finally, for~$e_i^B \in \rho(B)$ we have seen that~$\gamma(e_i^B)$ is a path in~$G$ containing~$e_i^Y$, for every~$1 \leq i \leq k$; note that~$E(\gamma_u(e_i^B)\cap E(\gamma_d(e_i^B)) =\{e_i^Y\}$ and~$V^\circ(\gamma_u(e_i^B)) \cap V^\circ(\gamma_d(e_i^B)) = \emptyset$ and thus $\gamma(e_i^B)$ contains no vertex of $\gamma(V(G))$ as an internal vertex (since $\gamma_u(e_i^B)$ and $\gamma_d(e_i^B)$ do not by assumption), concluding the proof that $\gamma$ is a strong immersion. (Again the proof for weak immersion follows verbatim with the only difference that $\gamma(e_i^B)$ may admits vertices of $V(G)$ as internal vertices depending on whether $\gamma_u(e_i^B)$ or $\gamma_d(e_i^B)$ did).

        To see that it is a \emph{rooted} strong immersion, note that~$\gamma_u,\gamma_d$ are rooted immersion (see \cref{def:rooted_immersion}) which implies that they respect the order of the roots. Note that $E(H) = E(H_B) \cup E(H^B)$ and $E(G) = E(G_Y) \cup E(G^Y)$ by \cref{cor:stitch-and-knit} and $\gamma$ agrees with $\gamma_u,\gamma_d$ away from $\rho(B)$ and extends them at $\rho(B)$. The claim follows using Parts 2, 3, and 4. of \cref{obs:stitching_fundamentals}. 
    \end{claimproof}

Finally, we derive the following.
    \begin{claim}
        $\gamma:\big(\bar H, \nu,\n,\Psi\big) \hookrightarrow \big(\bar G,\mu,\m,\Phi\big)$ is a strong~$\Omega$-knitwork immersion. 
    \end{claim}
    \begin{claimproof}
        By the previous claim we are left to verify the conditions of \cref{def:knitwork_immersion} on~$\nu,\mu$ as well as~$\n,\m$ and~$\Psi,\Phi$. But these are imminent from the definition:~$\Psi(v) \preceq \Phi(\gamma(v))$ using the fact that for~$v \in B$ we have~$\gamma(v) \in Y$ and~$\Psi(v) = \Psi_{B}(v)$ as well as~$\Phi(\gamma(v)) = \Phi_{Y}(\gamma_d(v))$ whence the claim follows from the fact that~$\Phi_{B}(v) \preceq \Phi_{Y}(\gamma_d(v))$ using that~$\gamma_d$ is an~$\Omega$-knitwork immersion, and analogously for~$v \in \bar{Y}$ implying~$\gamma(v) \in \bar{B}$. 

        The claims for~$\nu,\mu$ and~$\n,\m$ follow similarly from \cref{cor:stitch-and-knit} using the fact that~$\gamma_u,\gamma_d$ are~$\Omega$-knitwork immersions together with \cref{claim:stitch_knit_effect_on_cut} and Parts 2 and 3 of \cref{obs:knitting_fundamentals_knitworks} (for~$\mu$ and~$\m$ respectively) and the fact that~$\{B,\bar{B}\}$ and~$\{Y,\bar{Y}\}$ partition the vertex sets~$V(H)$ and~$V(G)$ respectively.  
    \end{claimproof}
This concludes the proof; recall that we implicitly used \cref{thm:knitting_knitwork_immersion_claim_stable}.
\end{proof}
\begin{remark}
    Note that for both strong and weak immersion the above knitting of immersion never produces ``new vertices'' for which a path self-intersects, simply because we never reroute at vertices and concatenate paths at common edges. Thus, this result would equally well hold when defining $\Omega$-knitwork immersion via linear paths instead. Compare with \cref{lem:knitwork_immersion_weak_vs_strong} and the preceding discussion. 
\end{remark}
    
Using \cref{lem:stitch-and-knit} we derive the following more general version of \cref{thm:knitting_knitwork_immersion}.

\begin{theorem}\label{thm:knitting_knitwork_immersion_general}
    Let~$\Omega=(V(\Omega),\preceq)$ be a well-quasi-order and let~$\bar G=(G,\pi,X_1,\ldots,X_\ell)$ and~$\bar H=(H,\pi,A_1,\ldots,A_\ell)$ be loopless rooted Eulerian digraphs of common index~$\ell \in 2\N$. Let~$Y \subset V(G)$ and~$B \subset V(H)$ induce proper rooted~$k$-cuts in~$\bar{G}$ and~$\bar{H}$ respectively for some~$k \in 2\N$ such that $\Abs{\insc(Y)} = \Abs{\insc{B}}$. Let~$\pi(Y)=(e_1^Y,\ldots,e_k^Y)=\pi(\bar Y)$ and~$\pi(B)=(e_1^B,\ldots,e_k^B) =\pi(\bar{B})$ be orderings of~$\rho_G(Y),\rho_H(B)$ respectively. Let
    \begin{align*}
        \bar G_Y &=  \stitch(\bar G;\pi,Y) \text{ with down-stitch vertex } y_\ast,\\
       \bar G^Y &=\stitch(\bar G;\pi,\bar{Y}) \text{ with up-stitch vertex } y^\ast,\\
        \bar H_B &=  \stitch(\bar H;\pi,B) \text{ with down-stitch vertex } b_\ast,\\
        \bar H^B &=\stitch(\bar H;\pi,\bar{B}) \text{ with up-stitch vertex } b^\ast, 
    \end{align*}

    respectively. Let $(\mu_Y,\m_Y,\Phi_Y), (\mu^Y,\m^Y,\Phi^Y)$ be a $\Omega$-sleeves for $\bar G_Y$ and $\bar G^Y$ respectively, and let $(\nu_B,\n_B,\Psi_B)$, $(\nu^B,\n^B,\Psi^B)$ be $\Omega$-sleeves for $\bar H_B$ and $\bar H^B$ respectively; denote the respective $\Omega$-knitworks by $\GGG_Y,\GGG^Y$ as well as $\HHH_B,\HHH^B$. Further let
    \begin{align*}
        \gamma_d:&  \HHH_B \hookrightarrow \GGG_Y, \text{ with } \gamma_d(b_\ast) = y_\ast,\text{ and}\\
        \gamma_u:&  \HHH^B \hookrightarrow \GGG^Y , \text{ with } \gamma_u(b^\ast) = y^\ast,
    \end{align*}be strong~$\Omega$-knitwork immersions. Fix orderings~$\pi_{G_Y}(Y) \coloneqq \pi_{G^Y}(y^*)$ of $\rho_G(Y)$ as well as~$\pi_{H_B}(B) \coloneqq \pi_{H^B}(b^*)$ of $\rho_H(B)$. Let $\GGG \coloneqq \knit( \GGG^Y, \GGG_Y; \pi_{G^Y}, y^*,\pi_{G_Y},Y)$ and $\HHH \coloneqq \knit( \HHH^B, \HHH_B; \pi_{H^B},b^*,\pi_{H_B},B)$.
    
    Then there exists a strong~$\Omega$-knitwork immersion~$\gamma:\HHH \hookrightarrow \GGG$ such that~$\restr{\gamma}{H[B]} = \restr{\gamma_u}{H_B[B]}$ and~$\restr{\gamma}{H[\bar{B}]} = \restr{\gamma_d}{H^B[\bar{B}]}$. The same holds true for weak immersion $\hookrightarrow^*$.
\end{theorem}
\begin{proof}
    By \cref{obs:stitching_fundamentals_knitworks} and \cref{lem:stitch-and-knit}  we derive that $\stitch(\HHH;\pi,B) = \HHH_B$ and $\stitch(\HHH;\pi,\bar B) = \HHH^B$ as well as  $\stitch(\GGG;\pi,Y) = \GGG_Y$ and $\stitch(\GGG;\pi,\bar Y) = \GGG^Y$. The claim then follows from \cref{thm:knitting_knitwork_immersion} applied to $\HHH$ and $\GGG$ respectively.
\end{proof}

Let $\Omega$ be a well-quasi-order and let~$\GGG$ and~$\HHH$ be clamped $\Omega$-knitworks with underlying rooted Eulerian digraphs $\bar G = (G,\pi,X)$ and $\bar H = (H,\pi,A)$. Let $Y \subset V(G)$ induce a proper rooted~$k$-cut for some $k \in 2\N$ in $\bar G$.
The following is imminent by \cref{def:rooted_cut} and \cref{obs:stitching_fundamentals,obs:stitching_fundamentals_knitworks}. 
\begin{observation}
   It holds $\insc(Y) = \{X\}$ and $\outc(Y) = \emptyset$. Further, let $\pi(Y)=\pi(\bar Y)$ be orderings of $\rho(Y)$, then $\stitch(\GGG;\pi,Y)$ and $\stitch(\GGG; \pi,\bar Y))$ are clamped $\Omega$-knitworks.
\end{observation}

Assume further that $\delta(A) = k$. Then we write $\knit(\HHH,\GGG;\pi,Y)$ instead of $\knit(\HHH,\GGG;\pi_H,A,\pi_G,Y)$ for simplicity, as we may only knit $\HHH$ to $\GGG$ along $\pi(A)$ by \cref{def:knitting_knitworks,def:knitting_rooted_graphs}.

\cref{thm:knitting_knitwork_immersion} allows us to make the following simplification regarding clamped~$\Omega$-knitworks.

\begin{corollary}\label{cor:immersion_of_stitches_yields_immersion}
    Let~$\GGG= \big((G,\pi,X),\mu,\m,\Phi\big)$ and~$\HHH= \big((H,\pi,A),\nu,\n,\Psi\big)$ be clamped~$\Omega$-knitworks. Let~$\GGG_X \coloneqq \stitch(\GGG;\pi,X)$ and $\GGG^X \coloneqq \stitch(\GGG;\pi,\bar{X})$, as well as~$\HHH_A \coloneqq \stitch(\HHH;\pi,A)$ and~$\HHH^A \coloneqq \stitch(\HHH;\pi,\bar{A})$. Then~$$\HHH \hookrightarrow \GGG \text{ if and only if } \HHH_A \hookrightarrow \GGG_X \text{ and } \HHH^A \hookrightarrow \GGG^X,$$ 
   and the same holds true for immersion $\hookrightarrow^*$.
\end{corollary}
\begin{proof}
    The forward implication is clear and the backward implication follows from \cref{thm:knitting_knitwork_immersion}, noting that the additional assumptions on $\gamma_d$ and $\gamma_u$ are trivially satisfied.
\end{proof}

\subsection{Decomposing Knitworks} 
\label{subsec:decomp_and_altering_kitworks}
\cref{thm:knitting_knitwork_immersion} is a powerful tool, that allows us to simplify the question whether given two $\Omega$-knitworks $\HHH$ and $\GGG$ it holds $\HHH \hookrightarrow \GGG$, by decomposing $\HHH$ and $\GGG$ into smaller pieces via stitching, inductively proving that the stitches can be immersed into each other, and then knitting the pieces back together using \cref{thm:knitting_knitwork_immersion} to prove $\HHH \hookrightarrow \GGG$. On a high level it is the main tool that allows us to inductively prove well-quasi-order results by decomposing our graphs into pieces for which we have already proved well-quasi-ordering. In this section, we provide a first ``meta theorem'' in that direction.

Recall the \cref{def:torso_and_pieces} of pieces of $\GGG$.
\begin{definition}\label{def:stitch_class}
    Let $\Omega$ be a well-quasi-order. Let $\mathbf{G}(\ell;\Omega)$ be a class of $\Omega$-knitworks of index $\ell$. We define $\stitch(\mathbf{G}(\ell;\Omega))$\index{stitch!of $\mathbf{G}(\ell;\Omega)$} to be the class of $\Omega$-knitworks obtained by adding for every $\GGG \in \mathbf{G}(\ell;\Omega)$ the pieces of $\GGG$.
    
    We define $\torso(\mathbf{G}(\ell;\Omega))$\index{torso!of $\mathbf{G}(\ell;\Omega)$} to be the class of $\Omega$-knitworks obtained by adding $\torso(\GGG)$ for every $\GGG \in \mathbf{G}(\ell;\Omega)$.
\end{definition}

\begin{remark}
    Recall that if $\GGG$ is clamped, then $\bar G = (G, \pi,X)$ for some $X \subseteq V(G)$, and we have $\torso(\GGG) = \stitch(\GGG;\pi,\bar X)$ by definition.
\end{remark}

The following is an easy observation from the \cref{def:rooted_immersion} of rooted immersion, where the first part follows from \crefdef{def:rooted_immersion}{2}, noting that the down-stitch vertex is the only non-root vertex, and the second part uses that torsos are controlled.
\begin{observation}\label{obs:immersion_on_torso_stitches_respects_stitchvertices}
    Let $\Omega$ be a well-quasi-ordered. Let $\mathbf{G}(\ell;\Omega)$ be a class of $\Omega$-knitworks of index $\ell$.
    \begin{enumerate}
        \item Let $\GGG_1,\GGG_2 \in \stitch(\mathbf{G}(\ell;\Omega))$ with down-stitch vertices $x^1_*$ and $x^2_*$ respectively.
    Let $\gamma: \GGG_1 \hookrightarrow^* \GGG_2$, then $\gamma(x^1_*)=x^2_*$ and $\gamma$ is stable.\label{obs:immersion_on_torso_stitches_respects_stitchvertices:1}
    \item Let $\GGG,\HHH \in \torso(\mathbf{G}(\ell;\Omega))$ with up-stitch vertices $(x_1^*,\ldots,x_\ell^*)$ and $(y_1^*,\ldots, y_\ell^*)$ respectively.
    Let $\gamma: \GGG \hookrightarrow^* \HHH$, then $\gamma(x_i^*)=y_i^*$ for every $1 \leq i \leq \ell$.\label{obs:immersion_on_torso_stitches_respects_stitchvertices:2}
    \end{enumerate}
\end{observation}
\begin{remark}
    Note that, regarding $(1)$ of the observation, the $\Omega$-sleeve of $\GGG_i$ is only defined for $x_i^*$ whence the map is easily seen to be stable.
\end{remark}

 We have the following decomposition theorem.
 
\begin{theorem}\label{lem:knitting_immersions_decomposition}
    Let $\Omega=(V_\Omega,\preceq)$ be a well-quasi-order and let $\ell \geq 1$. Let $\mathbf{G}(\ell)$ be a class of $\Omega$-knitworks of index $\ell$.
    If $\torso(\mathbf{G}(\ell))$ and $\stitch(\mathbf{G}(\ell))$ are well-quasi-ordered by (strong) $\Omega$-knitwork immersion, then $\mathbf{G}(\ell)$ is well-quasi-ordered by (strong) $\Omega$-knitwork immersion.
\end{theorem}

We outsource the most crucial step of the proof of \cref{lem:knitting_immersions_decomposition} as an independent lemma.

\begin{lemma}\label{lem:knitting_immersions_decomposition_handson}
    Let $\Omega=(V_\Omega,\preceq)$ be a well-quasi-order and let $\ell \geq 1$. For $i=1,2$ let $\GGG_i$ be $\Omega$-knitworks of index $\ell$ rooted in $(\pi_i,X_1^i,\ldots,X_\ell^i)$ respectively. Let $\HHH_i \coloneqq \torso(\GGG_i)$ with up-stitch vertices $(h_1^i,\ldots,h_\ell^i)$. Let $\PPP_i^j \coloneqq \stitch(\GGG_i;\pi,X_i^j)$ with down-stitch vertex $x_i^j$ for every $1 \leq j \leq \ell$, and let $\bar P_i^j$ be its underlying rooted Eulerian digraph. Let $\gamma:\HHH_1 \hookrightarrow \HHH_2$ and $\eta_j:\bar P_1^j\hookrightarrow \bar P_2^j$ for every $1 \leq j \leq \ell$. Then there exists $\gamma^*:\GGG_1 \hookrightarrow \GGG_2$ that agrees with the maps $\gamma,\eta_1,\ldots,\eta_\ell$ where defined.

    The same holds true for weak-immersion, replacing $\hookrightarrow$ with $\hookrightarrow^*$.
\end{lemma}
\begin{proof}
The proofs for $\hookrightarrow$ and $\hookrightarrow^*$ are analogous; we provide a proof for strong $\Omega$-knitwork immersion. 
Note that by assumption of the theorem the maps $\eta_j$ are independent of $\Omega$-sleeves, i.e., it is only defined on the rooted digraphs $\bar P_i^j = (P_i^j,\pi_{\PPP_i},X_i)$ for $1 \leq j \leq \ell$ and $i=1,2$. A

Let $\tilde\PPP_i^j = \left(\bar P_i^j,\emptyset, \emptyset,\emptyset)\right)$ where $\emptyset$ means the maps are nowhere defined. Then the following is true by construction and the definition of pieces.
\begin{claim}\label{lem:knitting_immersions_decomposition_handson-claimweird}
    For every $i=1,2$ and $1 \leq i \leq j$, the $\Omega$-sleeve of $\tilde\PPP_i^j$ and the one of $\PPP_i^j$ agree everywhere except for the down-stitch vertex $x_i^j$.
\end{claim}

First note that the theorem holds for $\ell=1$ by \cref{thm:knitting_knitwork_immersion_general} and \cref{lem:stitch-and-knit} noting that for a clamped $\Omega$-knitwork $\GGG$ with underlying rooted Eulerian digraph $\bar G = (G,\pi,X_1)$ we have $\torso(\GGG) = \stitch(\GGG;\pi,\bar X_1)$ where we set $\pi(\bar X_1) \coloneqq \pi(X_1)$ as usual, together with \cref{obs:immersion_on_torso_stitches_respects_stitchvertices}, and using \cref{lem:knitting_immersions_decomposition_handson-claimweird}. 
    
 Henceforth, assume that $\ell \geq 2$. The proof is by iterative application of \cref{thm:knitting_knitwork_immersion_general} and \cref{lem:stitch-and-knit}, knitting the pieces back to the torsos (where we have the option to arbitrarily define the sleeve at the down-stitch vertices). Observe that \crefthm{obs:immersion_on_torso_stitches_respects_stitchvertices}{1} implies that $x_1^j$ is mapped to $x_2^j$ under the respective immersion $\eta_j:\bar P_1^j \hookrightarrow \bar P_2^j$ for every $1 \leq j \leq \ell$. Similarly, \crefthm{obs:immersion_on_torso_stitches_respects_stitchvertices}{2} implies that $\gamma: \HHH_1 \hookrightarrow \HHH_2$ satisfies $\gamma(h_1^p)= h_2^p$ for every $1 \leq p \leq \ell$. Let $X_i \coloneqq V(G_i) \setminus\bigcup_{j=1}^{\ell}X_i^j$ for $i=1,2$.
Inductively define $\GGG_i^0 \coloneqq \HHH_i$ and $\GGG_i^{p+1} \coloneqq \knit\hspace*{-0.7ex}\big(\GGG_i^p,\tilde\PPP_i^{p+1}; \pi_{\GGG_i^p},h_i^{p+1},\pi_{\PPP_i^{p+1}},X_i^{p+1}\big)$ for all $p<\ell$ and both $i=1,2$. Let $i \in \{1,2\}$.
    \begin{claim}\label{claim:stitching_back_the_pieces}
        $\GGG_i^p$ is a well-defined $\Omega$-knitwork for every $0 \leq p \leq \ell$ and it satisfies the following conditions:
    \begin{enumerate}[label=(\roman*)]
        \item for every $p < q \leq \ell$, $h_i^q \in V(G_i^p)$ and $\rho_{G_i^p}(h_i^q) = \rho_{H_i}(h_i^q)$ as well as $\pi_{G_i^p}(h_i^q)= \pi_{H_i}(h_i^q)$; essentially we do not touch the vertices or incidences of $h_i^q$.
        \item $V(G_i^p)$ is partitioned into $X_i$, $X_i^1,\ldots,X_i^p$, $\{h_i^j \mid  p+1\leq j \leq \ell\}$.
        \item Let $V_i^p \coloneqq X_i \cup \bigcup_{j=1}^p X_i^j$. Then $G_i[V_i^p] = G_i^p[V_i^p]$ and their $\Omega$-sleeves agree on $V_i^p$.
    \end{enumerate}
    In particular $\GGG_i^\ell = \GGG_i$.
    \end{claim}
    \begin{claimproof}
    For $p=0$ the claim is clear. Thus, assume $(i)-(iii)$ to be true for some $p \geq 0$; we give a proof for $p+1\leq \ell$.
     Note that $X_i^{p+1}$ induces a proper rooted cut in $\tilde\PPP_i^{p+1}$ and, by \cref{def:stitching_std} of stitches and \crefthm{obs:torso_and_pieces_fundamentals}{2}, it holds $\rho_{P_i^{p+1}}(X_i^{p+1}) = \rho_{P_i^{p+1}}(x_i^{p+1})= \rho_{H_i}(h_i^{p+1})$ (see also \cref{obs:stitching_fundamentals}). Again, simply by \cref{def:stitching_std} of up- and down-stitches, for $\pi_{H_i}(h_i^{p+1})=(e_1,\ldots,e_t)$, say, we have $\pi_{P_i^{p+1}}(X_i^{p+1})=(e_1,\ldots,e_t)$ for the same edges. In particular the edges satisfy $e_k\in \rho_{P_i^{p+1}}^+(X_i^{p+1}) \iff e_t\in \rho^-_{H_i}(h_i^{p+1})$ for $1 \leq k \leq t$; this is essentially a special case of \cref{obs:knitwork_immersion_on_mu}. By our assumption $(i)$ we derive that $$e_k\in \rho_{P_i^{p+1}}^+(X_i^{p+1}) \iff e_k\in \rho^-_{G_i^p}(h_i^{p+1}).$$  Thus, since $\GGG_i^p$ is a well-defined $\Omega$-knitwork by assumption, $\GGG_i^{p+1}$ is a well-defined $\Omega$-knitwork by \cref{obs:knitting_fundamentals_knitworks} as all the assumptions to apply the \cref{def:knitting_knitworks} of knitting are given.

     We are left to prove that $(ii)$ and $(iii)$ remain satisfied. For $(ii)$ this is a straightforward consequence of the fact that $X_i^j \cap X_i^t = \emptyset$ for any two distinct $j,t \in \{1,\ldots,\ell\}$, as imposed by
     \cref{def:rooted_graph} of rooted digraphs, as well as $(ii)$ for $G_i^p$ and \crefthm{obs:knitting_fundamentals}{2}. 

     Finally, $(iii)$ follows from \crefthm{obs:knitting_fundamentals}{3}, \cref{obs:knitting_fundamentals_knitworks} and \cref{lem:knitting_immersions_decomposition_handson-claimweird}, where the details are exactly as in the proof of \cref{lem:stitch-and-knit}. Note again that by definition, the resulting graph is in fact \emph{equal} and not only isomorphic to the original graph . This concludes the proof of the claim.
    \end{claimproof}

    With \cref{claim:stitching_back_the_pieces} we are left to iteratively apply \cref{thm:knitting_knitwork_immersion}. 

    \begin{claim}
        For every $0 \leq p \leq \ell$ there is a strong immersion $\xi_p:\GGG_1^p\hookrightarrow \GGG_2^p$ that agrees with $\gamma,\eta_1,\ldots,\eta_p$ where defined.
    \end{claim}
    \begin{claimproof}
        The claim is true for $p=0$ by assumption of the theorem and setting $\xi_p = \gamma$. 
       Suppose that the claim was false and let $1 \leq p \leq \ell$ be minimal with $\GGG_1^p \not\hookrightarrow \GGG_2^p$ with a respective strong immersion agreeing with $\gamma,\eta_1,\ldots,\eta_p$ where defined. Henceforth, let $1 \leq p\leq \ell$ be minimal refuting the claim. In particular there is no such $\xi_p:\GGG_1^p \hookrightarrow \GGG_2^p$. 
      For $t =1,2$ let $\AAA_t^{p} = \stitch(\GGG_t^p;\pi,X_t^{p})$ with down-stitch vertex $x_t^1$ and let $\BBB_t^p=\stitch(\GGG_t^p;\pi,\bar X_t^p)$ with up-stitch vertex $h_t^p$. By construction (and \cref{lem:stitch-and-knit}) we derive that $\AAA_t^p = \PPP_t^p$ (and not $\tilde \PPP_t^p$) and $\BBB_t^p=\GGG_t^{p-1}$. 
        
    By \cref{thm:knitting_knitwork_immersion_general} using \cref{lem:knitting_immersions_decomposition_handson-claimweird} we derive that we can ``knit'' $\eta_p:\tilde\PPP_1^p  \hookrightarrow \tilde\PPP_2^p$ and $\xi_{p-1}:\GGG_1^{p-1} \hookrightarrow \GGG_2^{p-1}$---which exists by minimality assumption on $p$---accordingly; a contradiction. Thus there is no such $p$ and the claim follows.
    \end{claimproof}
    
    By the previous claim, we have $\xi_\ell:\GGG_1^\ell \hookrightarrow \GGG_2^\ell$, which with \cref{claim:stitching_back_the_pieces} implies that there is a map $\gamma^*:\GGG_1 \hookrightarrow \GGG_2$ as desired, concluding the proof.
\end{proof}

With \cref{lem:knitting_immersions_decomposition_handson}, we have gathered all the tools to prove \cref{lem:knitting_immersions_decomposition}

\begin{proof}[Proof of \cref{lem:knitting_immersions_decomposition}]
    As before, we settle for a proof of the strong version without loss of generality. 

    Let $(\GGG_i)_{i \in \N}$ be a sequence of $\Omega$-knitworks $\GGG_i \in \mathbf{G}(\ell)$ with $\bar G_i = (G_i,\pi,X_i^1,\ldots,X_i^\ell)$ for every $i \in \N$. To prove the theorem we need to show the following. 
    \begin{itemize}
        \item[$(\star)$]There exist $i,j \in \N$ with $i<j$ such that $\GGG_i \hookrightarrow \GGG_j$. 
    \end{itemize} Let $X_i \coloneqq V(G_i) \setminus\bigcup_{j=1}^{\ell}X_i^j$. Let $\HHH_i \coloneqq \torso(\GGG_i)$ with up-stitch vertices $(h_{i}^1,\ldots,h_{i}^\ell)$ and let $\PPP_i^j \coloneqq \stitch(\GGG_i;\pi,X_i^j)$ with down-stitch vertex $x_i^j$ for every $1 \leq j \leq \ell$ and every $i \in \N$. By definition, for every $i \in \N$ and every $1 \leq j \leq \ell$ it holds $\HHH_i \in \torso(\mathbf{G}(\ell))$ and $P_i^j \in \stitch(\mathbf{G}(\ell))$. 

    Since $\stitch(\mathbf{G}(\ell))$ is well-quasi-ordered by assumption of the theorem, \cref{obs:wqo_of_tuples} implies that the class of $\ell$-tuples of elements in $\stitch(\mathbf{G}(\ell))$---we denote it by $\stitch(\mathbf{G}(\ell))^\ell$---is well-quasi-ordered by~$\preceq^\ell$, where we define $(a_1,\ldots,a_\ell) \preceq^\ell (b_1,\ldots,b_\ell)$ if and only if $a_j \preceq b_j$ for every $1 \leq j \leq \ell$. Let $\PPP_i \coloneqq (\PPP_i^1,\ldots,\PPP_i^\ell) \in \stitch(\mathbf{G}(\ell))^\ell$ for every $i \in \N$. Then \cref{obs:wqo_yields_infinite_chain} implies the existence of an infinite index set $I \subseteq \N$ such that $(\PPP_i)_{i \in I}$ is a chain with respect to $\preceq^\ell$. In particular, for $i,j \in I$ with $i<j$ we have $\PPP_i \hookrightarrow \PPP_j$, which implies that $\PPP_i^k \hookrightarrow \PPP_j^k$ for every $1 \leq k \leq \ell$. 

    Similarly, again applying \cref{obs:wqo_yields_infinite_chain}, there exists an infinite index set $J \subseteq I$ such that $(\HHH_i)_{i\in J}$ is a chain with respect to $\preceq$. Let $i,j \in J$ with $i<j$, then $\HHH_i \hookrightarrow \HHH_j$ as well as $\PPP_i^t \hookrightarrow \PPP_j^t$ for every $1 \leq t \leq \ell$ by construction of $J$. By \cref{lem:knitting_immersions_decomposition_handson} we derive that $\GGG_i \hookrightarrow \GGG_j$, proving $(\star)$.
\end{proof}

The following is a generalisation of \cref{lem:knitting_immersions_decomposition}.

\begin{theorem}\label{thm:knitting_immersions_decomposition}
    Let $\Omega$ be a well-quasi-order and let $\Omega^*$ be an extension of $\Omega$. Let $\ell \geq 1$ and let $\mathbf{G}(\ell;\Omega)$ be a class of $\Omega$-knitworks of index $\ell$. Let $\mathbf{G^*}(\Omega^*)$ be a class of $\Omega^*$-knitworks well-quasi-ordered by (strong) $\Omega^*$-knitwork immersion.

    If $\stitch(\mathbf{G}(\ell;\Omega)),\torso(\mathbf{G}(\ell;\Omega)) \subseteq \mathbf{G^*}(\Omega^*)$, then  $\mathbf{G}(\ell;\Omega)$ is well-quasi-ordered by (strong) $\Omega$-knitwork immersion.
\end{theorem}
\begin{proof}
    This follows at once from \cref{lem:knitting_immersions_decomposition} and the fact that subclasses of well-quasi-ordered classes remain well-quasi-ordered.
\end{proof}

By our definition of $\Omega$-knitwork, its pieces have almost nowhere defined $\Omega$-sleeves. We can prove a stronger, more technical version of \cref{lem:knitting_immersions_decomposition_handson} useful for inductive purposes in future work as well as for the inductive proof of \cref{thm:intro}; noting that in the above results the restriction that their sleeves were almost nowhere defined was irrelevant.

\begin{definition}\label{def:focus}
    Let $\Omega$ be a well-quasi-order. Let $\GGG=((G,\pi,X),\mu,\m,\Phi)$ be a clamped $\Omega$-knitwork. Let $\ell \geq 1$ and let $\LL^- \coloneqq \left((X_1,\ldots X_\ell),\pi')\right)$ such that $X_i,X_j \subset V(G)$ are pairwise disjoint for $1\leq i<j\leq \ell$ and such that $\bar X_i$ induces a proper rooted cut in $\GGG$. Further $\pi'(X_i) = \pi'(\bar X_i)$ fixes some linear order on $\rho(X_i)$ for every $1 \leq i \leq \ell$. Then we call $\LL^-$ a \emph{location} in $\GGG$. Let $Y \coloneqq \bar X \cap \bar X_1 \cap \ldots\cap \bar X_\ell$.

    We define $\GGG(X_1,\ldots,X_\ell) \coloneqq ((G,\pi^*,X, X_1,\ldots,X_\ell),\mu',\m',\Phi')$ to be the $\Omega$-knitwork obtained by letting $f' \coloneqq \restr{f}{Y}$ for all $f' \in \{\mu,\m,\Phi\}$ and letting $\pi^*(X) \coloneqq \pi(X)\eqqcolon \pi^*(\bar X)$ as well as $\pi^*(X^i)\coloneqq \pi(X^i)\eqqcolon \pi^*(\bar X^i)$ for every $1 \leq i \leq \ell$. 
    Define $\GGG^{\bar X_i} \coloneqq \stitch(\GGG;\pi^*,X_i)$ with up-stitch vertex $x_i^*$. 

    Then we call $\GGG(X_1,\ldots,X_\ell)$ the \emph{$\LL^-$-focus of $\GGG$}, and $\GGG^{\bar X_i}$ its \emph{$i$-th piece}.
\end{definition}
\begin{remark}
    By definition $\GGG(X_1,\ldots,X_\ell)$ is an $\Omega$-knitwork (the sleeve is no longer defined on the new root sets). Note that the focus depends on a choice of $\pi'$.
\end{remark}

\begin{lemma}\label{lem:knitting_immersions_decomposition_handson_located}
    Let $\Omega=(V_\Omega,\preceq)$ be a well-quasi-order and let $\ell \geq 1$. For $i=1,2$ let $\GGG_i$ be clamped $\Omega$-knitworks rooted in $(\pi_i,X^i)$ respectively. Let $\LL_i^-=((X_1^i,\ldots,X_\ell^i),\pi_i')$ be a location in $\GGG_i$. Let $\HHH_i \coloneqq \torso(\GGG_i(X_1^i,\ldots,X_\ell^i))$, let $\PPP_i \coloneqq \stitch(\GGG_i(X_1^i,\ldots,X_\ell^i); \pi_i,X^i)$ with down-stitch vertex $x^i_*$ and let $\GGG_i^j \coloneqq \stitch(\GGG_i;\pi_i',X_i^j)$ with up-stitch vertex $x_i^j$ for every $1 \leq j \leq \ell$. Assume that $\gamma:\HHH_1 \hookrightarrow \HHH_2$, $\eta:\PPP_1\hookrightarrow \PPP_2$, and $\eta_j:\GGG_1^j \hookrightarrow \GGG_2^j$ for every $1 \leq j \leq \ell$. Then there exists $\gamma^*:\GGG_1 \hookrightarrow \GGG_2$ that agrees with the maps $\gamma,\eta,\eta_1,\ldots,\eta_\ell$ where defined.

    The same holds true for weak-immersion, replacing $\hookrightarrow$ with $\hookrightarrow^*$.
\end{lemma}
\begin{proof}
    Let $\GGG_i^* \coloneqq \GGG_i(X_1^i,\ldots,X_\ell^i)$ and let $\PPP_i^j \coloneqq \stitch(\GGG_i^*;\pi_i',X_i^j)$ be the respective down-stitches with down stitch vertex $a_i^j$ for every $1 \leq j \leq \ell$ and $i=1,2$.  Note that $\PPP_{i} \cong \stitch(\GGG_i;\pi_i,X^i))$, i.e., it is a piece in the sense of \cref{def:torso_and_pieces} up-to renaming the down-stitch vertex.

    By definition it turns out that for every $i=1,2$ and $1 \leq j \leq \ell$ the underlying digraphs $G_i^j \cong P_i^j$ are isomorphic by construction. Further, $\pi_{\GGG_i^j}(x_i^j) = \pi_{\PPP_i^j}(a_i^j)$ again by construction whence one easily verifies that $\eta_j:\bar P_1^j\hookrightarrow \bar P_2^j$ is a strong rooted immersion.

    By \cref{lem:knitting_immersions_decomposition_handson} we derive the existence of a strong $\Omega$-knitwork immersion $$\gamma^*:\GGG_1(X^1,X_1^1,\ldots,X_\ell^1) \hookrightarrow \GGG_2(X^2,X_1^2,\ldots,X_\ell^2)$$ that agrees with $\gamma,\eta,\eta_1,\ldots,\eta_\ell$ where defined, where $\left(\dom(\gamma),\dom(\eta),\dom(\eta_1),\ldots,\dom(\eta_\ell)\right)$ is a partition of $V(G_1)$ by \cref{def:focus}. Further note that for $i=1,2$ and every $v \in V(G_i)\setminus X$ there is $\DDD \in \{\HHH_i,\PPP_i,\GGG_i^1,\ldots,\GGG_i^\ell\}$ such that $v \in V(D)$ and $\mu_\DDD(v) = \mu_{\GGG_i}(v)$, as well as  $\m_\DDD(v) = \m_{\GGG_i}(v)$ and $\Phi_\DDD(v) = \Phi_{\GGG_i}(v)$ (or the maps are not defined on $v$ in both $\Omega$-knitworks). In particular, one easily verifies that $\gamma^*(E(G_1))$ is an $\m_{\GGG_2}$-respecting linkage and further $\left(\gamma^*(V(G_1)),\gamma^*(E(H))\right)$ is a strong $\Omega$-knitwork immersion model of $\GGG_1$ in $\GGG_2$, concluding the proof.  
\end{proof}

\cref{thm:knitting_immersions_decomposition} is a first version of our main tool to reduce general results about $\Omega$-knitworks to controlled and clamped $\Omega$-knitworks which will be of importance for future research. Although it does not explicitly read that way, it is worth pointing out that this theorem roughly suggests that if we manage to nicely decompose graphs in a tree-like fashion such that the torso of each bag is an $\Omega$-knitwork, then it suffices to prove that the classes of torsos ``induced'' by the bags are well-quasi-ordered. A full proof of this claim would need a stronger version of \cref{thm:knitting_immersions_decomposition} relying on \cref{lem:knitting_immersions_decomposition_handson_located}, more definitions and another argument regarding linkedness, which we omit here as it does not fit the scope of this paper. We defer this to future work.

\subsection{Manipulating Knitworks}
\label{subsec:manipulate_knitworks}
We discuss two more easy but helpful results that allow us to alter $\Omega$-knitworks. The first one verifies that a well-quasi-ordered class of $\Omega$-knitworks stays well-quasi-ordered if one omits some roots.

\begin{lemma}\label{lem:wqo_loosing_roots}
    Let $\Omega$ be a well-quasi-order and let $\ell \geq 1$. Let $\mathbf{G}(\ell; \Omega)$ be a class of $\Omega$-knitworks of index $\ell$ well-quasi-ordered by (strong) $\Omega$-knitwork immersion. Let $\ell' < \ell$ and let $\mathbf{G}(\ell';\Omega)$ be a class of $\Omega$-knitworks such that for every $\GGG= \big((G,\pi,X_1,\ldots,X_{\ell'}), \mu, \m ,\Phi) \in \mathbf{G}(\ell'; \Omega) $ there exists an $\Omega$-knitwork $\GGG^*=(\bar G^*, \mu, \m, \Phi) \in \mathbf{G}(\ell; \Omega)$ where $\bar G^* = (G, \pi,Y_1,\ldots,Y_\ell)$ and there are $1 \leq i_1<\ldots <i_{\ell'} \leq \ell$ with $Y_{i_j} = X_j$ for $1 \leq j \leq \ell'$.

    Then $\mathbf{G}(\ell';\Omega)$ is well-quasi-ordered by (strong) $\Omega$-knitwork immersion.
\end{lemma}
\begin{proof}
    We give a proof for strong $\Omega$-knitwork immersion, the other case follows verbatim. Let $(\GGG_i)_{i \in \N}$ be an infinite sequence of elements in $\mathbf{G}(\ell';\Omega)$ such that $\bar G_i = (G_i, \pi,X_1^i,\ldots,X_{\ell'}^i)$ for every $i \in \N$. Let $(\GGG_i^*)_{i \in \N}$ be a respective sequence of elements in $\mathbf{G}(\ell;\Omega)$ such that $\GGG_i^*$ and $\GGG_i$ share a common underlying quasi-Eulerian digraph, the same $\Omega$-sleeve and $\bar G_i^* = (G_i,\pi,Y_1^i,\ldots,Y_\ell^i)$ such that there are $1 \leq j_1^i,\ldots,j_{\ell'}^i \leq \ell$ with $X_j^i = Y_{j_i}^i$ for every $i \in \N$. By the pigeonhole principle there is an infinite index set $I \subseteq \N$ such that $j_t^p = j_t^q$ for every $p,q \in I$ and $1 \leq t \leq \ell'$; let them be defined via $j_t \coloneqq j_t^p$ for simplicity.

    Since $\mathbf{G}(\ell;\Omega)$ is well-quasi-ordered by strong $\Omega$-knitwork immersion, there is $\gamma: \GGG_i^* \hookrightarrow \GGG_j^*$ for some $i < j $ with $i,j \in J$. By \cref{def:rooted_immersion} and \cref{def:knitwork_immersion} it follows that $\gamma : \GGG_i \hookrightarrow \GGG_j$ concluding the proof.
\end{proof}
\begin{remark}
    Note that simply omitting roots without further assumptions is not possible as we need the graph to be Eulerian away from the vertices in the root sets.
\end{remark}
The second one, as alluded to in \cref{subsec:knitting_immersions}, justifies why we can restrict to loopless graphs. We fix the following.

\begin{definition}[$\Omega$-independence]\label{def:Omega_independent}
   Let $\Omega$ be a well-quasi-order and let $\mathbf{G}(\Omega)$ be a class of $\Omega$-knitworks. We call $\mathbf{G}(\Omega)$ \index{knitwork!$\Omega$-independent}\index{OMEGA@$\Omega$-independent}\emph{$\Omega$-independent} if for every $\GGG =( \bar G, \mu, \m ,\Phi) \in \mathbf{G}(\Omega)$ and every partial function $\Psi : V(G) \to V(\Omega)$ it holds $(\bar G, \mu,\m,\Psi) \in \mathbf{G}(\Omega)$.

\end{definition}

 We further define the \emph{core} of a class of $\Omega$-knitworks.
\begin{definition}
    Let $\Omega$ be a well-quasi-order and let $\mathbf{G}(\Omega)$ be a class of $\Omega$-knitworks. Let $\mathbf{G}$ be the class of rooted digraphs obtained from $\mathbf{G}(\Omega)$ by adding $\bar G$ to $\mathbf{G}$ for every $\GGG \in \mathbf{G}(\Omega)$. Then we call $\mathbf{G}$ the \emph{core}\index{knitworks!core}\index{core} of $\mathbf{G}(\Omega)$.
\end{definition}

Recall that by \cref{obs:wqo_of_tuples} given a well-quasi-order $\Omega$ as well as the natural well-quasi-order $(\N,\leq )$ we have that $\Omega \times (\N,\leq)$ is a well-quasi-order. One easily sees that $\Omega \cong \Omega \times \{0\}$, thus we may see $\Omega \times (\N,\leq)$ as an extension of $\Omega$ in the obvious way.

\begin{lemma}\label{lem:wqo_assume_no_loops}
    Let $\mathbf{G}$ be a class of planted Eulerian digraphs, possibly with loops. Let $\mathbf{G^*}$ be obtained from $\mathbf{G}$ by deleting the loops for every graph.
    Let $\Omega$ be any well-quasi-order and let $\Omega^* \coloneqq \Omega \times (\N,\leq)$. Let $\mathbf{G^*}(\Omega^*)$ be an $\Omega^*$-independent class of planted $\Omega^*$-knitworks with core $\mathbf{G^*}$. 

    If $\mathbf{G^*}(\Omega^*)$ is well-quasi-ordered by (strong) $\Omega^*$-knitwork immersion, then $\mathbf{G}$ is well-quasi-ordered by (strong) rooted immersion.
\end{lemma}
\begin{proof}
    We prove it for strong immersion; the other case is analogous. 
    
    Let $(\bar G_i)_{i\in \N}$ be a sequence of rooted Eulerian digraphs in $\mathbf{G}$ with roots $(\pi_i,x_i)$. Since $(\N,\preceq)$ is a well-quasi-order, \cref{obs:wqo_yields_infinite_chain} yields the existence of $I\subseteq \N$ such that $\loops(x_i) \leq \loops(x_j)$ for every $1 \leq i \leq j$ with $i,j \in I$; identify $I$ with $\N$. For every $i \in \N$ construct an $\Omega^*$-knitwork $\GGG_i^*$ from $\bar G_i$ as follows. For every $i \in \N$, let $\bar G_i^*$ be obtained from $\bar G_i$ by deleting the loops for every $v \in V(G_i)$ (note here that loops are of no relevance to the roots, and thus $\bar G_i^*$ and $\bar G_i$ share the same roots). Fix some element $\star \in V(\Omega)$. We define a sleeve $(\mu_i, \m_i, \Psi_i)$ for $\bar G_i^*$ where $\Psi_i(v) \coloneqq (\star, \Abs{\loops_{G_i}(v)})$ for every $v \in V(G_i)\setminus \{x_i\}$. Since $\mathbf{G^*}( \Omega^*)$ is $\Omega^*$-independent it holds that $\GGG_i^*=(\bar G_i^*,\mu_i,\m_i,\Psi_i) \in \mathbf{G^*}(\Omega^*)$ for every $i \in \N$ by definition. In particular there exists $1 \leq i < j $ such that $\gamma^* \colon\GGG_i^* \hookrightarrow \GGG_j^*$ by assumption of the lemma. We claim that there is a rooted immersion $\gamma \colon \bar G_i \hookrightarrow \bar G_j$ agreeing with $\gamma^*$ where defined. Note that $\gamma^*: \bar G_i^* \hookrightarrow \bar G_j^*$ is a strong rooted immersion by \cref{def:knitwork_immersion}, and thus it suffices to argue that we can extend $\gamma^*$ for loops (since $G_i^* \subseteq G_i$ and $E(G_i) \setminus E(G_i^*) = \loops(G_i)$).
    
   Note that \crefdef{def:knitwork_immersion}{6} for $\Omega^*$-knitwork immersion implies that $\Abs{\loops_{G_i}(v)}\leq \Abs{\loops_{G_j}(\gamma^*(v))}$, which in turn implies that for every $v \in V(G_i)$---including $v=x_i$ due to our first step---there is an injection $\ell: \loops_{G_i}(v) \to \loops_{G_j}(\gamma^*(v))$. Thus, we may easily extend $\gamma^*$ on $e \in \loops_{G_i}(v)$ by setting $\gamma(e) \coloneqq \ell(e) \in \loops_{G_j}(\gamma(v))$. This proves the lemma.
\end{proof}

\section{A Framework for ordering Graphs encoded by Carvings}
\label{sec:tree_lemma}
\subsection{Carvings and Tree Decompositions}

We assume the reader to be familiar with the definition of \emph{tree and branch decompositions} of undirected graphs as well as \emph{treewidth} and \emph{branchwidth}; see \cite{Diestel2017} for details. Throughout this exposition, whenever we talk about the treewidth of a digraph $G$---written $\tw{G}$---we implicitly mean the undirected treewidth of the underlying undirected graph of $G$ unless explicitly stated otherwise.

In this subsection we discuss the concept of \emph{carving width} introduced and analysed in \cite{carvingwidth}\footnote{Here the parameter is referred to as \emph{congestion}.} as well as \cite{RobertsonST1994}, a notion similar to branch width tailored towards induced cuts of a graph and suited for immersions. 

\begin{definition}[Carving]
    \label{def:carving}
    Let~$G$ be a (possibly not weakly connected) digraph. A~\emph{carving} of~$G$ is given by a pair~$\TTT \coloneqq (T,\ell)$ where
    \begin{enumerate}[label=(\arabic*)]
        \item $T$ is a cubic undirected tree, and 
        \item $\ell: V(G) \to \leaves{T}$ is an injective map labelling part of the leaves.
    \end{enumerate}

    Let~$e=\{t_1,t_2\} \in E(T)$ and let~$T_1,T_2$ be the components of~$T-\{t_1,t_2\}$ containing~$t_1$ and~$t_2$ respectively. Let $X_e \coloneqq\ell^{-1}(\leaves{T_1})$  where for a leaf~$t$ not in the image of~$\ell$ we define~$\ell^{-1}(t) \coloneqq \emptyset$. Then we call $\{X_e,\bar X_e\}$ the \emph{partition induced by $e$}. We define~$\epsilon_{\TTT,G}:E(T) \to 2^{E(G)}$ via~$\epsilon_{\TTT,G}(\{t_1,t_2\})\coloneqq \rho_G(X_e)$. We may omit the subscripts $\TTT,G$ and simply write~$\epsilon$ if the carving and graph are clear from context.
    
     We define the \emph{width} of~$\{t_1,t_2\}$ through~$\w_G(\{t_1,t_2\}) \coloneqq \Abs{\epsilon(\{t_1,t_2\})}$. Finally, we define the width of the decomposition via~$\w_G(\TTT) \coloneqq \max\{\w(e) \mid e \in E(T)\}$. Again we may omit the subscript $G$ if the graph is clear from context.

    The \emph{carving width of~$G$} is defined via~$$\cw{G} \coloneqq \min\{\w_G(\TTT) \sth \TTT \text{ is a carving of } G\}.$$
\end{definition}
\begin{remark}
    We allow for leafs~$t$ with~$\ell^{-1}(t) = \emptyset$ so that a branch decomposition always exists (since~$T$ needs to be cubic this is needed for~$\Abs{V(G)} \leq 2$. Otherwise, it may be neglected.
    Note that the same carving $(T,\ell)$ may be a carving of different graphs $G_1,G_2$ (for example if $G_2$ is obtained from $G_1$ by adding some edges) which is why we introduce the definition of the width with a subscript. This will ease the notation in later sections.
\end{remark}

The next lemma guarantees that for any $d\geq 1$, the class of Eulerian digraphs of carving width at most $d$ is closed under immersions. While this is again folklore, we add a proof for completeness.

\begin{lemma}[Homogeneity for Immersions]\label{lem:cw_closed_under_immersion}
    Let~$G$ and~$H$ be Eulerian digraphs. If $G$ immerses $H$ then $\cw{H} \leq \cw{G}$.
\end{lemma}
\begin{proof}
    Let~$H,G$ be Eulerian digraphs such that~$G$ immerses $H$. Let~$\TTT=(T,\ell)$ be a carving of~$G$ witnessing its carving width. Since~$G$ contains~$H$ as an immersion, \cref{obs:immersion_robust_under_splitting_off} implies the existence of an Eulerian subgraph~$G' \subseteq G$ such that $H$ can be obtained from $G$ by sequentially splitting of edge-pairs at vertices.  Thus, it suffices to prove that carving width is closed under taking Eulerian subgraphs and splitting off edge-pairs at vertices. 
    
    \begin{claim}
        Let~$G' \subseteq G$ be Eulerian. Then~$\cw{G'} \leq \cw{G}$.
    \end{claim}
    \begin{claimproof}
        Since $G,G'$ are Eulerian, the graph $G'' \coloneqq G-G'$ is Eulerian (not necessarily connected). By \cref{obs:covering_eulerian_digraphs} there exists a cycle cover $\CCC$ of $G''$. In particular $G'$ can be obtained from $G$ by subsequently subtracting $E(C)$ for $C \in \CCC$ from $G$ and deleting isolated vertices. Thus, it suffices to prove the claim for~$G' \coloneqq G-E(C)$ for some fixed cycle~$C$ in $G$ (the claim is trivial if $G'$ is obtained by deleting an isolated vertex). But clearly~$\TTT$ is still a valid~carving of~$G'$ (where $\ell_{G'}^{-1}(t) = \emptyset$ if $\ell_G^{-1}(t)=v$ for some $v$ that became isolated). Since we only deleted edges, the width of any edge in~$E(T)$ cannot increase; this concludes the proof of the claim.
    \end{claimproof}

    \begin{claim}
        Let~$G'$ be obtained from~$G$ by splitting off $(e_1,e_2)$ at~$v \in V(G)$ for $e_1,e_2 \in E(G)$. Then~$\cw{G'} \leq \cw{G}$.
    \end{claim}
    \begin{claimproof}
        Let $v_l,v_r \in V(G)$ such that~$e_1=(v_l,v)$ and $e_2=(v,v_r)$. Let $(G',e') = \spl(G;(e_1,e_2))$. We claim that~$\TTT$ still witnesses that~$G'$ satisfies~$\cw{G'} \leq \cw{G}$. To see this note that the resulting edge~$e'$ appears in~$\epsilon(\{t_1,t_2\})$ if and only if $v_l$ and~$v_r$ lie on opposite sides of~$T \setminus \{t_1,t_2\}$. But then~$\epsilon(\{t_1,t_2\}) \cap \{e_1,e_2\} \neq \emptyset$ and thus the width of the edge cannot increase: for we delete the edges $e_1,e_2$ in $G$ and add $e'$.
    \end{claimproof}
This concludes the proof.
\end{proof}
\begin{remark}
    Note that if a loopless directed graph~$G$ contains a vertex~$v$ of degree~$k \in \N$, then~$\cw{G} \geq k$ simply because the edge incident to the leaf representing~$v$ in the cubic tree must have width~$k$. 
\end{remark}

Finally, there is a nice qualitative relation between carving width and treewidth due to Bienstock \cite[Theorem 1 and preceding discussion]{carvingwidth}. 

\begin{lemma}[$ ${\cite[Theorem 1]{carvingwidth}}]
\label{thm:qualitative_equivalence_of_tw_and_cw}
    Let~$\Delta \in \N$ and~$G$ be a digraph of maximum degree~$\Delta$. Then~$\frac{2}{3}\tw{G} \leq \cw{G} \leq \Delta \cdot \tw{G}$. 
\end{lemma}
\subsection{Carving Knitworks}

We extend the \cref{def:carving} of carvings of digraphs to carvings of rooted digraphs as follows.

\begin{definition} \label{def:knitworks_carving}
    Let $t \geq 1$ and~$\bar G = (G,\pi,X_1,\ldots,X_t)$ be a rooted Eulerian digraph. 
    
    Let $\torso(\bar G) = (G^\star, \pi,x_1^*,\ldots,x_t^*)$ for respective up-stitch vertices $x_1^*,\ldots,x_t^*$. Let~$(T,\ell)$ be a carving of~$G^*$ and let~$e_1,\ldots,e_t\in E(T)$ be the edges incident to $\ell(x_1^*),\ldots,\ell(x_t^*)$ respectively. We define~$\rt(T) \coloneqq (e_1,\ldots,e_t)$ and call~$(T,\ell; \rt(T))$ a \emph{rooted carving of~$\bar G$}. The width is defined via~$\w(T,\ell;\rt(T)) \coloneqq \operatorname{max}(\w(T,\ell), \Abs{X_1}, \ldots, \Abs{X_t})$ and the carving width~$\cw{\bar G}$ of the rooted Eulerian digraph is defined as the minimum width over all rooted carvings of~$\bar G$. 
\end{definition}
\begin{remark}
    If $\bar G$ is clamped, whence $\rt(T)=(e_1)$, we write $\rt(T) = e_1$ for simplicity.

    \cref{def:knitworks_carving} may seem odd at first: We want to point out that we could loosen the above definition to define the carving of $\GGG$ to be the minimum over the carving width of its torso and pieces. Although all the relevant results would still hold, we can derive these versions from the above definition by simply combining the carvings.
\end{remark}

We lift the definition to $\Omega$-knitworks as usual, making it independent of a choice of $\Omega$-sleeve.
\begin{definition}
    Let~$\Omega$ be a well-quasi-order and let~$\GGG=(\bar{G},\mu,\m,\Phi)$ be an~$\Omega$-knitwork. We define the \index{knitwork!carving width}\emph{carving width} of~$\GGG$, denoted $\cw{\GGG}$, as the carving width of~$\bar{G}$, i.e.,~$\cw{\GGG} \coloneqq \cw{\bar{G}}$. 
\end{definition}

Carving width for rooted and unrooted Eulerian digraphs is closely related; in particular if the rooted digraphs are controlled it is equal.

\begin{lemma}\label{lem:cw_of_rooted_vs_unrooted}
    Let~$\bar G = (G,\pi,X_1,\ldots,X_t)$ be a rooted Eulerian digraph of index $t \geq 1$. Let $k \in 2\N$ such that $\Abs{X_i} \leq k$ and $\delta(X_i) \leq k$ for every $1 \leq i \leq t$. Then~$\cw{\bar G} \leq \cw{G} + t\cdot k$. 
    
    If $\bar G$ is controlled, then $\cw{\bar G} = \cw{G}$.
\end{lemma}
\begin{proof}
    For the upper bound we construct a rooted carving witnessing said width as follows. Let~$X_0 \coloneqq V(G) \setminus \bigcup_{i=1}^t X_i$ possibly empty. Let~$\torso(\bar G) = (G^\star, \pi,x_1^*,\ldots,x_t^*)$. By \cref{obs:torso} $V(G^\star)$ is partitioned by $X_0,\{x_1^*\},\ldots,\{x_t^*\}$.
    Let $(T_0,\ell_0)$ be a carving of~$G[X_0]$ witnessing its carving width (if $X_0 = \emptyset$ let $T$ be given by a cubic tree on $4$ vertices). Let $T$ be obtained from $T_0$ by subdividing any edge $e \in E(T)$ $t$ times, introducing vertices $\tilde v_1,\ldots,\tilde v_t$ and adding new vertices $v_1,\ldots,v_t$ and  for every $1 \leq i \leq t$ adding the edge $e_i\coloneqq \{v_i,\tilde v_i\}$. Then $T$ is a tree satisfying $\leaves{T} = \leaves{T_0} \cup \{v_1,\ldots,v_t\}$. Define $\ell: V(G^\star) \to \leaves{T}$ via $\restr{\ell}{X_0} = \ell_0$ and $\ell(x_i^*) = v_i$ for every $1 \leq i \leq t$. Then $(T,\ell;(e_1,\ldots,e_t))$ is a rooted carving of $\torso(\bar G)$ satisfying $$\w(T,\ell;(e_1,\ldots,e_t)) \leq \w(T_0,\ell_0) + \delta(x_1^*) + \ldots \delta(x_t^*) \leq \cw{G} + t\cdot k.$$ 

    For the second part of the statement, note that if $\bar G$ is controlled then $\torso(\bar G) \cong \bar G$ and the claim follows. 
\end{proof}

We define the first distinct class of $\Omega$-knitworks of interest for this paper.

\begin{definition}\label{def:class_of_bounded_cw}
    Let $\Omega$ be a well-quasi-order and fix~$k, \ell \geq 1$. We define the class $\mathbf{G}(k, \ell; \Omega)$ of well-linked $\Omega$-knitworks of index $\ell$ admitting carving-width at most~$k$.
\end{definition}
\begin{remark}
    The class $\mathbf{G}(k,\ell;\Omega)$ is $\Omega$-independent by definition.
\end{remark}

Our first main result reads as follows. 

\begin{theorem}[Well-Quasi-Order for bounded Carving Width]\label{thm:wqo_bounded_carvingwidth_knitworks_gen}
    Let $\Omega$ be a well-quasi-order and fix~$k, \ell \geq~1$. Then $\mathbf{G}(k, \ell; \Omega)$ is well-quasi-ordered by strong $\Omega$-knitwork immersion.
\end{theorem}
\begin{remark}
    We want to emphasise that it seems very plausible that a version of the theorem for maximum degree four (omitting the maps $\mu$ and $\m$ and in sleeves) could be directly derived from the respective well-quasi-ordering result of bounded rankwidth graphs (specifically circle graphs) via the pivot-minor relation; see \cite{wqo_rankwidth} for details. 
\end{remark}

In fact we can prove the following ``stronger'' version.
\begin{theorem}[Well-Quasi-Order for bounded Carving Width]\label{thm:wqo_bounded_carvingwidth_knitworks}
    Let $\Omega$ be a well-quasi-order and fix~$k, \ell \geq~1$. Then $\mathbf{G}(k, \ell; \Omega)$ is well-quasi-ordered by stable strong $\Omega$-knitwork immersion.
\end{theorem}

Note that given a class $\CCC(\Omega)$ of $\Omega$-knitworks, if $\CCC(\Omega)$ is well-quasi-ordered by stable (strong) $\Omega$-knitwork immersion, then it is also well-quasi-ordered by (strong) $\Omega$-knitwork immersion. 
\begin{observation}\label{obs:stable_wqo_implies_gen}
    Let $\Omega$ be a well-quasi-order. Let $\CCC(\Omega)$ be a class of $\Omega$-knitworks. If $\CCC(\Omega)$ is well-quasi-ordered by stable (strong) $\Omega$-knitwork immersion, then $\CCC(\Omega)$ is well-quasi-ordered by (strong) $\Omega$-knitwork immersion respectively.
\end{observation}

Combining \cref{thm:wqo_bounded_carvingwidth_knitworks} and \cref{obs:stable_wqo_implies_gen} we derive  \cref{thm:wqo_bounded_carvingwidth_knitworks_gen}; hence we are left to prove \cref{thm:wqo_bounded_carvingwidth_knitworks}.

\medskip

 It is an easy exercise to verify that \cref{thm:wqo_bounded_carvingwidth_knitworks} (even for weak $\Omega$-knitwork immersion) and \cref{thm:wqo_bounded_carvingwidth_knitworks_gen} are in fact equivalent. One direction is given by \cref{obs:stable_wqo_implies_gen}. The other follows from the fact that bounded carving width implies bounded degree, and the fact that $\mathbf{G}(k,\ell;\Omega)$ is $\Omega$-independent: We may ``add colours'' that encode the degrees of the vertices and thus make sure that the resulting ``weak'' $\Omega$-knitwork immersion is in fact forced to be strong and stable, since all the vertices are ``fully used''.

\subsection{A Lemma on Trees}

To prove \cref{thm:wqo_bounded_carvingwidth_knitworks} we need some results from the literature.
The following results are originally due to Robertson and Seymour \cite{GMIV} and have been adapted to a setting akin to ours by Geelen, Gerards, and Whittle \cite{Gee02}. We start with some definitions that we transcribe from \cite{Gee02} to our setting.

\begin{definition}[Labeled Rooted Trees]
\label{def:labeled_trees}
    Let~$T$ be a tree and~$k \in \N$. Let~$r \in V(T)$ be some vertex and direct all the edges of~$T$ away from~$r$, such that~$r$ has in-degree~$0$ and every other vertex has in-degree exactly~$1$. Then we call~$(T,r)$ a \emph{tree rooted at~$r$} or simply \emph{a rooted tree}. Let $P$ be a path starting in $r$ and ending in a leaf of $T$ then we call $P$ a \emph{root-to-leaf} path.
    
    Let~$\omega: E(T) \to \{1,\ldots,k\}$ be some labelling function. Then we call~$(T,\omega)$ a~\emph{$k$-labelled tree with label~$\omega$}, and similarly we call~$(T,r,\omega)$ a~\emph{$k$-labelled rooted tree}.
\end{definition}
Given labelled trees, they define \emph{linkedness} as follows.

\begin{definition}[Linked labeled Trees]
\label{def_linked_trees}
    Let~$(T,r,\omega)$ be a $k$-labelled rooted tree. 
    Let~$e_1,e_2$ be two edges in~$E(T)$ with~$\omega(e_1) = \omega(e_2) = \rho$ such that there is a directed root-to-leaf path~$P$ in~$(T,r)$ visiting~$e_1$ prior to~$e_2$, and let~$P^1_2 \subseteq P$ be the sub-path with first edge~$e_1$ and last edge~$e_2$. If~$\omega(e) \geq \rho$ for every~$e \in E(P^1_2)$, then we say that~$(e_1,e_2)$ is \emph{$\omega$-linked in $(T,r,\omega)$} or~\emph{$e_2$ is~$\omega$-linked to~$e_1$}. 
\end{definition}
\begin{remark}
Clearly~$\omega$-linkedness is a transitive relation, i.e., if~$(f,f')$ and~$(e,f)$ are~$\omega$-linked, then~$(e,f')$ is~$\omega$-linked.
\end{remark}

The above definition is easily extended to rooted and labelled forests, where a rooted forest is a (possibly infinite) set of disjoint rooted trees with a common label~$\omega$ obtained by combining the labels of each tree to a single one. (Here we assume all the trees to have pairwise disjoint vertex and edge sets).

Let~$\FFF$ be a rooted forest and let~$S \subseteq E(\FFF)$ be some set of edges. Then we define~$u_\FFF(S)$ to be the set of those edges in~$\FFF$ whose tail is a head of an edge in~$S$; in a sense the ``child-edges'' with~$u$ referring to ``under''.

The following is the version of the Tree Lemma \cite[Lemma 2.2]{GMIV} due to Robertson and Seymour, in the setting presented by Geelen, Gerards, and Whittle; in their setting an antichain given  a quasi-order $(V_\Omega,\preceq)$ is a set $\AAA \subseteq V_\Omega$ such that element sin $\AAA$ are pairwise incomparable. This is trivially seen to be equivalent to our definition.

\begin{lemma}[Lemma on Trees; Theorem 3.1 in \cite{Gee02}]
\label{tree-lemma}
    Let~$\mathcal{F}$ be a rooted forest of~$k$-labelled rooted trees with common label~$\omega$ for some~$k \in \N$. Let~$\preceq$ denote a quasi-order on the edges of~$\mathcal{F}$ with no infinite strictly descending sequence and such that~$e \preceq f$ whenever~$(f,e)$ is~$\omega$-linked. If the edges of~$\FFF$ are not well-quasi-ordered by~$\preceq$ then there exists an infinite antichain~$\AAA$ of edges of~$\FFF$ such that~$(u_\FFF(\AAA),\preceq)$ is a well-quasi-order.
\end{lemma}
\begin{remark}
    Think of the quasi-order on the edges as~$e \preceq e'$ if and only if the rooted tree below~$e$ can be ``(strongly) immersed'' in the rooted tree below~$e'$ in some sense that needs to be made precise. 
\end{remark}

Geelen, Gerards, and Whittle extracted a very useful corollary to \cref{tree-lemma} suitable for cubic trees; we introduce the needed notation before stating their result. A \emph{leaf edge} of a forest~$\FFF$ is an edge~$e \in E(\FFF)$ that is incident to some~$v \in \leaves{\FFF}$, \emph{root-edges} and \emph{non-leaf edges} are defined in the obvious way.

\begin{definition}[Binary Forest]
    A \emph{($k$-labelled) binary forest}~$(\FFF,\mathrm{left},\mathrm{right})$ is a ($k$-labelled) rooted forest~$\FFF$ where the roots of the rooted trees have out-degree~$1$, together with maps~$\mathrm{left},\mathrm{right}$ which are defined on the non-leaf edges such that the head of each non-leaf edge~$e \in E(\FFF)$ has exactly two distinct out-going edges~$\lenks{e},\riets{e}$.
\end{definition}

\begin{lemma}[Lemma on Cubic Trees, Theorem 3.2 in \cite{Gee02}]
\label{cubic-tree-lemma}
    Let~$k \in \N$ and let~$(\FFF,\mathrm{left},\mathrm{right})$ be an infinite~$k$-labelled binary forest with label~$\omega$. Let~$\preceq$ denote a quasi-order on~$E(\FFF)$ with no infinite strictly decreasing sequences, such that~$e\preceq f$ whenever~$(f,e)$ is~$\omega$-linked. If the leaf edges of~$\FFF$ are well-quasi-ordered by~$\preceq$ but the root edges of~$\FFF$ are not, then~$\FFF$ contains an infinite sequence~$(e_i)_{i \in \N}$ of non-leaf edges such that:
    \begin{enumerate}
        \item $(e_i)_{i \in \N}$ is an antichain with respect to~$\preceq$,
        \item $(\lenks{e_i})_{i \in \N}$ is a chain with respect to $\preceq$, i.e., $\lenks{e_i} \preceq \lenks{e_j}$ for every $0 \leq i < j $,
        \item $(\riets{e_i})_{i \in \N}$ is a chain with respect to $\preceq$, i.e., $\riets{e_i} \preceq \riets{e_j}$ for every $0 \leq i < j $.
    \end{enumerate}
\end{lemma}

To make use of \cref{cubic-tree-lemma}, and inspired by the above \cref{def_linked_trees}, we define \emph{linked} carvings as follows.

\begin{definition}[Linked carving]
\label{def:linked_carving}
    Let~$G$ be an Eulerian digraph and let~$(T,\ell)$ be a carving of~$G$ of width~$k \in \N$. Let~$e_1,e_2 \in E(T)$ and let~$T_1,T_2$ be subtrees of~$T$ such that~$T_i$ is the component of~$T-e_i$ not containing $e_{3-i}$ for~$i=1,2$. Let~$P$ be a shortest path in~$T$ containing both~$e_1$ and $e_2$. We call~$\{e_1,e_2\}$ \emph{linked} if the minimum width of an edge in~$P$ equals~$\delta(V(T_1),V(T_2))$. We call~$(T,\ell)$ \emph{linked}, if every pair of edges is linked.
\end{definition}

The following is the aforementioned result due to Geelen, Gerards, and Whittle, where the definition of ``branch width'' depends on the symmetric submodular function.

\begin{theorem}[Theorem 2.1 in \cite{Gee02}]
    An integer-valued symmetric submodular function with branch width $n$ has a linked branch decomposition of width~$n$.
\end{theorem}

We immediately derive the following.
\begin{corollary}
\label{cor:linked_ebw}
    Let~$G$ be an Eulerian digraph with carving width~$n \in \N$. Then~$G$ admits a linked carving of width~$n$.
\end{corollary}

With \cref{cubic-tree-lemma} and \cref{cor:linked_ebw} we have gathered all the tools needed from the literature to prove \cref{thm:wqo_bounded_carvingwidth_knitworks}.

\section{Ordering Eulerian digraphs of Bounded Carving Width}
\label{sec:carving_width}

The main goal of this section is to prove \cref{thm:wqo_bounded_carvingwidth_knitworks}. In light of \cref{obs:stable_wqo_implies_gen} we will present all the results in this section in the setting of stable (strong) $\Omega$-knitwork immersion.

As a first step, we prove that the class of Eulerian digraphs on a fixed set of vertices is well-quasi-ordered, noting that the carving width of the class may be unbounded.
\begin{lemma}\label{lem:wqo_finitely_many_vertices}
    Let $\Omega=(V(\Omega),\preceq)$ be a well-quasi-order and let $k, \ell, \Delta \geq 1$. Let $\mathbf{C}(k,\ell,\Delta)$ be a class of $\Omega$-knitworks of index $\ell$ such that for $\GGG \in \mathbf{C}(k,\ell,\Delta)$ it holds $\Abs{V(G)} \leq k$ and for every $v \in V(G)$ with $\deg(v) > \Delta$ it holds $v \notin \dom(\mu_G)$. Then $\mathbf{C}(k,\ell,\Delta)$ is well-quasi-ordered by stable strong $\Omega$-knitwork immersion.
\end{lemma}
\begin{proof}
 Let $(\GGG_i)_{i \in \N}$ be a sequence of $\Omega$-knitworks in $\mathbf{C}(k,\ell,\Delta)$ where $\GGG_i = \big(( G_i,\pi,X_1^i,\ldots,X_\ell^i),\mu_i,\m_i,\Phi_i)$. By moving to a suitable subsequence indexed by some infinite set $I \subseteq \N$ we may assume that~$\Abs{V(G_i)} =t+1 \leq k$ for some $t \in \N$ and by using a suited bijection we may assume that $V(H_i) = V=\{v_0,\ldots,v_t\}$ for all $i \in I$. Similarly, using the pigeonhole principle, we may assume for every $i,j \in I$ that $\dom(\mu_{H_i}) = \dom(\mu_{H_j})$ as well as $\dom(\m_{H_i}) = \dom(\m_{H_j})$ and $\dom(\Phi_{H_i}) = \dom(\Phi_{H_j})$. 

It is easily verified that there are only finitely many loopless non-isomorphic rooted digraphs on $t+1$ vertices without parallel edges. Thus, again using the pigeonhole principle, there is an infinite $I' \subseteq I$ such that all digraphs $H_i$ indexed by $I'$ are isomorphic (up to parallel edges), i.e., for every $i,j \in I'$ there is a map $\xi_i^j:  V(H_i) \to V(H_j)$ such that $(u,v) \in E(H_i)$ if and only if $(\xi_i^j(u),\xi_i^j(v)) \in E(H_j)$. Note that there are only $2\binom{t+1}{2}$ possible ``types'' of edges on $t+1$ vertices, and for every pair $0 \leq p , q \leq t$ of distinct $p,q$ let $\alpha^i_{p,q} \in \N$ be the number of $(v_p,v_q)$ edges in $H_i$. Since $(\N, \leq)$ is a well-quasi-order, by applying the pigeonhole principle at most $2\binom{t+1}{2}$ times, we find an infinite index set $J \subseteq I'$ such that for every $i < j $ with $i,j \in J$ it holds $\alpha^i_{p,q} \leq \alpha^j_{p,q}$ for every pair of distinct $0 \leq p,q \leq t$. In particular $H_i \subseteq H_j$ for every $i,j \in J$ with $i<j$.

Let $i \in J$ be fixed. By \cref{def:rooted_graph} of rooted digraphs, $X_p^i \cap X_q^i = \emptyset$ for every $1 \leq p < q \leq \ell$, implying that there are only finitely many distinct partitions of $V$ into $\ell$ sets. Thus, again applying the pigeonhole principle there is an infinite index set $J' \subseteq J$ such that $X_p^i = X_p^j$ for every $1 \leq p \leq \ell$ and every $i,j \in J'$. Further note that for every $j \in J'$ and $v\in V$ with $v \in \dom(\mu_{H_j})$ it holds $\delta(v) \leq \Delta$ and thus by repeated use of the pigeonhole principle---filtering once for each vertex---there is an infinite index set $J^* \subseteq J'$ such that $\mu_i(v) = \mu_j(v)$ as well as $\m_i(v) = \m_j(v)$ (after relabelling the edges) for every $v\in V$ and every $i,j \in J^*$. (Note that $\mu_i(v)$ may contain several edges of the same ``type'' as discussed above). 

Finally, again filtering $t$ times (once for each vertex) and using the fact that $\Omega$ is a well-quasi-order, there is an infinite index set $I^* \subseteq J^*$ such that $\Phi_i(v) \preceq \Phi_j(v)$ for every $i,j \in I^*$. This concludes the proof, since $H_i \subseteq H_j$ by construction, with agreeing $\Omega$-sleeves.
\end{proof}

As alluded to in \cref{sec:knitworks,sec:decomposition_theorem} one can reduce \cref{thm:wqo_bounded_carvingwidth_knitworks} to the setting of clamped knitworks; we may even reduce it to planted knitworks. Thus, we prove the following theorem which easily implies the former as we will see later.

\begin{theorem}\label{thm:wqo_bounded_carvingwidth_knitworks_red}
     Let~$k \in \N$ and let~$\Omega=(V(\Omega),\ll)$ be a well-quasi-order. Let~$(\GGG_i)_{i\in \N}$ be a sequence of well-linked planted~$\Omega$-knitworks~$\GGG_i=((G_i,\pi_i,x_i),\mu_i,\m_i,\Phi_i)$, and such that~$\cw{\GGG_i} \leq k$ for all~$i \in \N$. Then there exist~$j>i \geq 1$ such that~$\GGG_i \hookrightarrow \GGG_j$ by stable (strong) $\Omega$-knitwork immersion.
\end{theorem}
\begin{proof}
We give a proof for stable strong immersion since the general case follows from it. We will omit to say ``stable'' and tacitly assume all immersions to satisfy $(1)-(6)$ of \cref{def:knitwork_immersion} for notational simplicity. Towards a contradiction, assume that the theorem is false; let~$(\GGG_i)_{i \in \N}$ be a bad sequence of well-linked planted~$\Omega$-knitworks with~$\cw{\GGG_i} \leq k $ for all~$i \in \N$ with respect to strong~$\Omega$-knitwork immersion.

Let~$\GGG_i=(\bar{G_i},\mu_i,\m_i,\Phi_i)$ where~$\bar G_i = (G_i,\pi_{G_i},x_i)$ for some~$x_i \in V(G_i)$ and recall that~$x_i \notin \dom(\Phi_i) \cup \dom(\mu_i)$ by \cref{def:knitwork}, whence $x_i \notin \dom(\m_i)$, for every~$i \in \N$. Note that since $\bar G_i$ is controlled, $G_i$ is Eulerian (in particular its defect is $0$). \cref{lem:cw_of_rooted_vs_unrooted} yields $\cw{\GGG_i} = \cw{G_i}$, hence~$\delta(x_i) \leq k$ for every~$i \in \N$. By a variation of the pigeonhole principle, there exists an infinite index set $I\subseteq \N$ such that~$\delta(x_i) = \delta(x_j)$ for every~$i,j \in I$. Let $\delta \coloneqq \delta(x_i)$ for $i \in I$ and identify $I\cong \N$ for simplicity.

 By \cref{cor:linked_ebw} there exists a linked carving~$(T_i,\ell_i)$ for~$G_i$ and every~$i \in \N$ witnessing its carving width and in particular witnessing~$\cw{\GGG_i}$. For every~$i \in \N$, let~$e_i \in E(T_i)$ be the unique edge incident to~$r_i \coloneqq \ell_i(x_i)$, then~$(T_i,\ell_i;e_i)$ is a rooted carving of~$\bar G_i$ and we root the cubic trees~$T_i$ at~$r_i$. Thus the edges of the rooted tree~$(T_i,r_i)$ are directed away from $r_i$ from here on, whereas for $(T_i,\ell_i;e_i)$ we may technically view $T_i$ as an undirected tree. Let~$\epsilon_i \coloneqq \epsilon_{(T_i,\ell_i)}$ and define~$\omega_i(e) \coloneqq \Abs{\epsilon_i(e)}$ for every~$e \in E(T_i)$ and every~$i \in \N$. Then, for every $i\in \N$,~$\omega_i\colon E(T_i) \to \{1,\ldots,k\}$ is a labelling function and~$(T_i,r_i,\omega_i)$ is a~$k$-labelled rooted tree with label~$\omega_i$. 
 
 Fix $i\in \N$ and let~$e=(u,v) \in E(T_i)$ be some directed edge of the rooted tree $(T_i,r_i)$. Let~$T_i^u$,~$T_i^v$ be the two unique components of~$T_i-e$, where~$T_i^u$ contains~$u$ and~$T_i^v$ contains~$v$. We define~$X_e \coloneqq \ell_i^{-1}(\leaves{T_i^u})$ and~$X^e \coloneqq \ell_i^{-1}(\leaves{T_i^v})$. Clearly~$X^e = \bar{X_e}$ by definition. We derive the following.

\begin{claim}\label{claim:thm_wqo_bounded_carvingwidth_knitworks_red_uno}
    $X_e$ induces a proper rooted cut in~$\bar G_i$ and~$\rho(X_e) = \epsilon_i(e) = \rho(X^e)$.
\end{claim}
\begin{claimproof}
    Since~$(T_i,r_i)$ is a rooted tree with~$r_i =\ell(x_i)$, there is a root-to-leaf path $P$ in $T_i$ that contains $(u,v)$. Thus $r_i 
\in T_i^u$ implying~$x_i \in X_e$. Therefore~$X_e$ induces a proper rooted cut; the rest follows trivially by definition.
\end{claimproof}

 Assuming that~$(V(T_i)\cup E(T_i)) \cap (V(T_j)\cup E(T_j)) = \emptyset$ for distinct~$ i,j \in \N$---otherwise rename the trees accordingly---we let~$(\FFF,\mathrm{left},\mathrm{right})$ be the binary forest obtained by the union of all the rooted trees~$(T_i,r_i)$ after a choice of~$\mathrm{left}$ and~$\mathrm{right}$ maps for each rooted tree which is feasible since they are cubic. Let~$\omega:E(\FFF) \to \{1,\ldots,k\}$ be the respective labelling function of~$\FFF$ obtained by setting~$\restr{\omega}{E(T_i)} = \omega_i$ for every $i\in \N$. 
 
 We continue by defining a quasi-order~$\preceq$ on $E(\FFF)$ that respects~$\omega$-linkedness. To this extent, we need the following technical result that bridges the \cref{def:linked_carving} of linked carvings to $\omega$-linkedness. Intuitively speaking the following claim states that, given a pair $(e,f)$ of $\omega$-linked edges in $T_i$, the cuts (of equal order $t\in 2\N$) that they induce in $G_i$ can be connected by a $t$-linkage $\LLL_i$ such that $\LLL_i$ respects the restrictions imposed by $\m_i$, i.e., it is an $\m_i$-respecting linkage.  Recall \cref{def:types_of_linkages_on_a_cut}.

 \begin{claim}\label{claim:thm_wqo_bounded_carvingwidth_knitworks_red_feasible_linkages}
      Let~$i \in \N$ be fixed and~$e,f \in E(T_i)$ be two edges such that~$(e,f)$ is~$\omega$-linked and let~$t \coloneqq \omega(e) = \omega(f)$. Let $X \coloneqq X_f\cap X^e$. Then there exists an $\m_i$-respecting strong~$\{\rho(X^f),\rho(X_e)\}$-linkage~$\LLL_i \in \LLL(\rho,X)$ of order~$t$. Further every path~$P \in \LLL_i$ has one endpoint in~$X^f$ and one endpoint in~$X_e$.
 \end{claim}
 \begin{claimproof}
    Note that $X_f$ induces a proper rooted cut and $X^e$ induces a rooted cut ($\bar X^e$ induces a proper rooted cut) with $X_f\cap X^e = \emptyset$. Since the carving~$(T_i,\ell_i)$ is linked, using the \cref{def:linked_carving} and \cref{lem:Menger_for_Euler}, the claim follows from \cref{lem:from_linear_to_respecting_linkage_well-linked} for all assumptions are satisfied.
 \end{claimproof}

Recall that~$X^f = \bar{X_f}$.
  \begin{claim}\label{claim:thm_wqo_bounded_carvingwidth_knitworks_red_linked}
        Let~$i \in \N$ be fixed and~$e,f \in E(T_i)$ be two edges such that~$(e,f)$ is~$\omega$-linked. Let~$\pi_{G_i}(X_e)=\pi_{G_i}(X^e)$ be an ordering of~$\rho_{G_i}(X_e)$. Then there exists an ordering~$\pi_{G_i}(X_f)=\pi_{G_i}(X^f)$ of~$\rho_{G_i}(X_f)$ such that~$\stitch(\GGG_i;\pi_{G_i},X^f) \hookrightarrow \stitch(\GGG_i;\pi_{G_i},X^{e})$ by strong $\Omega$-knitwork immersion.
    \end{claim}
    \begin{claimproof}
      For notational convenience we omit the script $i$ throughout the proof of the claim, i.e., we write~$\GGG= \GGG_i$ and $T = T_i$. Let~$P$ be the shortest directed path in~$T$ visiting the edges~$e$ and~$f$ in this order. Then, since~$(e,f)$ is~$\omega$-linked, we know that~$\omega(e) = t = \omega(f)$ and~$\omega(\eta) \geq t$ for all~$\eta \in E(P)$ and some~$t \in 2\N$. In particular, since the carving is linked, we know that~$\Abs{\epsilon(e)} = \Abs{\epsilon(f)}$ and by \cref{def:linked_carving} of linkedness there is no $\{X^f,X_e\}$-cut in~$G$ of lower order. By \cref{claim:thm_wqo_bounded_carvingwidth_knitworks_red_feasible_linkages} there exists a strong~$\{\rho(X^f),\rho(X_e)\}$-linkage~$\LLL=\{P_1,\ldots,P_t\}$ of order~$t$ such that~$M_{\LLL}(v) \in \m(v)$ for every~$v \in \dom(\m) \cap X_f \cap X^e$ and such that~$\bigcup_{i=1}^t V^\circ(P_i)  \cap (X^f \cup X_e) = \emptyset$ (since $\LLL \in \LLL(\rho,X_f \cap X^e)$).

      \smallskip
      
      Let~$\pi_{G}(X_e) = (e_1,\ldots,e_t)$ be the ordering as given by the claim and without loss of generality assume that~$e_j$ is an end of~$P_j$ for every~$1 \leq j \leq t$ (else relabel the edges of the paths accordingly). 
        Let~$\rho(X^f) = \{f_1,\ldots,f_t\}= \rho(X_f) $ and without loss of generality assume that~$f_j$ is the other end of~$P_j$---note that~$f_j = e_j$ is possible---for every~$1\leq j \leq t$, else rename the edges accordingly. We set~$\pi_{G}(X_f) \coloneqq (f_1,\ldots,f_t)\eqqcolon \pi_G(X^f)$. Finally, let~$\GGG^{X_f}\coloneqq (\bar G^{X_f}, \mu^{X_f},\m^{X_f}.\Phi^{X_f}) =\stitch(\GGG;\pi_{G},X^f)$ and~$\GGG^{X_e} \coloneqq (\bar G^{X_e}, \mu^{X_e},\m^{X_e},\Phi^{X_e}) =\stitch(\GGG;\pi_{G},X^e)$ be the respective up-stitches with up-stitch vertices~$f^*$ and~$e^*$. Note that by \cref{claim:thm_wqo_bounded_carvingwidth_knitworks_red_uno}~$X_e,X_f$ are proper rooted cuts with $X^f = \bar X_f$ as well as $X^e = \bar X_e$ whence the above defined up-stitches are well-defined. Further recall that $V(G^{X_f}) = X^f \cup \{f^*\}$ by \cref{def:stitching_std}.
      
      We claim that~$\GGG^{X_f} \hookrightarrow \GGG^{X_e}$ by strong $\Omega$-knitwork immersion. To see this, define~$\gamma\colon V(G^{X_f}) \cup E(G^{X_f}) \to G^{X_e}$ to be the identity on~$X^f \subseteq X^e$  and set~$\gamma(f^*) \coloneqq e^*$. Further let~$\gamma$ be the identity on~$E(G[X^f]) \subseteq E(G[X^e])$ and set~$\gamma(f_j) \coloneqq P_j$ for every~$1 \leq j \leq t$. We claim that~$\gamma$ yields a strong~$\Omega$-knitwork immersion. We verify Conditions (1)-(6) of \cref{def:knitwork_immersion}.
      \begin{itemize}
      \item[(1)] By construction~$\gamma\colon G^{X_f} \hookrightarrow G^{X_e}$ is a strong immersion where further~$\gamma\big((f_1,\ldots,f_t)\big) = (e_1,\ldots,e_t)$ since~$e_j \in P_j = \gamma(f_j)$, implying that~$\gamma\colon \bar G^{X_f} \hookrightarrow \bar G^{X_e}$ is a strong rooted immersion; Condition (1) is satisfied. 
      \item[(2)] For $\m^{X_f}$ and $\m^{X_e}$ this is trivial since $\restr{\gamma}{X^f}$ is the identity and they are both not defined on $f^*$ and $e^*$ respectively. The same holds true for $\mu^{X_f}$ and $\mu^{X_e}$; Condition (2) is satisfied. 
      \item[(3)] The claim is trivially true for all~$v \in X^f$ with~$\rho(v) \subset E(G[X^f])$ again since~$\restr{\gamma}{E(G[X^f])}$ is the identity and thus~$\mu^{X_f}(v) = \mu^{X_e}(v)$ and in particular Condition (3) holds. Recall that~$f^* \notin \dom(\mu^{X_f})$ by definition. Thus let~$v \in X^f$ be a vertex with~$\rho_{G^{X_f}}(v) \cap \{f_1,\ldots,f_t\} \neq \emptyset$. By \cref{def:stitching_std,def:stitching_knitwork}~$\rho_{G^{X_f}}(v) = \rho_{G^{X_e}}(v) $ and~$\mu_{G^{X_f}}(v) = (g_1,\ldots,g_p) = \mu_{G^{X_e}}(v)$ for some~$p \leq k$. Again Condition (3) is trivially satisfied for all~$1 \leq j \leq p$ with~$g_j \notin \{f_1,\ldots,f_t\}$ since~$\gamma(g_j) = g_j$, and for~$g_j = f_\ell$ for some~$1 \leq \ell \leq t$ it holds~$\gamma(g_j)=P_\ell \subset G^{X_e}$ which is a path starting in~$f_\ell = g_j$; Condition (3) is satisfied. 
      \item[(4)] Note that~$f^* \notin \dom(\m^{X_f})$ and~$e^* \notin \dom(\m^{X_e})$ by \cref{def:stitching_knitwork} for they are up-stitch vertices, thus Condition (4) does not apply to~$e^*$. Finally, for~$v' \in \dom(\m^{X_e}) \setminus \gamma(\dom(\m^{X_f}))$ with~$v' \notin X_f \cap X^e$ the claim follows once again since~$\gamma$ is the identity (in particular if~$v' \in \dom(\m^{X_e})$ then~$v' = \gamma(v')$ whence~$v' \in \gamma(\dom(\m^{X_f}))$ and Condition (4) does not apply).
      
      If~$v' \in X_f \cap X^e$, then by our choice of~$\LLL_i$ satisfying \cref{claim:thm_wqo_bounded_carvingwidth_knitworks_red_feasible_linkages} and by construction of~$\gamma$,~$M_\gamma(v') = M_{\LLL_i}(v') \in \m^{X_e}(v')$; Condition (4) is satisfied. 

      \item[(5)] Since by definition~$\m^{X_e}$ and~$\m^{X_f}$ are not defined on~$e^*$ and~$f^*$ respectively, Condition (5) follows again by~$\restr{\gamma}{X^f}$ being the identity. 
      
      \item[(6)] Since~$f^* \notin \dom(\Phi^{X_f})$ and $e^* \notin \dom(\Phi^{X_e})$ by \cref{def:stitching_knitwork}, and~$\restr{\gamma}{X^f}$ is the identity, Condition (6) is satisfied.
      \end{itemize}
    Thus the claim follows.
    \end{claimproof}

We derive the following.
\begin{claim}\label{claim:thm_wqo_bounded_carvingwidth_knitworks_red_pi_i}
    Let~$i \in \N$ be fixed. Then there is a map~$\pi_i$ which for each edge~$e \in E(T_i)$ fixes an ordering~$\pi_i(X_e) = \pi_i(X^e)$ of the cut~$\rho_{G_i}(X_e)$, such that if~$(e,f)$ is~$\omega_i$-linked for edges~$e,f \in E(T_i)$, then $\stitch(\GGG_i;\pi_i,X^f) \hookrightarrow \stitch(\GGG_i;\pi_i,X^e)$ by strong $\Omega$-knitwork immersion. Further~$\pi_i(X_{e_i}) = \pi_{G_i}(x_i) = \pi_i(X^{e_i})$.
\end{claim}
\begin{claimproof}

    The proof is via induction by pushing orders starting from the root edge~$e_i$ as follows. Let~$F_i(e_i) \subseteq E(T_i)$ be the maximal set of edges containing all~$f \in E(T_i)$ with~$(e_i,f)$ being~$\omega_i$-linked such that there is no other edge~$f' \in E(T_i)$ on the unique directed~$(e_i,f)$-path in~$T_i$ with~$f' \in F_i(e)$. For each such edge $f \in F_i$ define~$\pi_i(X_f)$ via \cref{claim:thm_wqo_bounded_carvingwidth_knitworks_red_linked}. Let~$T_i^1,T_i^2$ be the sub-trees rooted in~$\lenks{e_i}$ and $\riets{e_i}$ respectively. If~$\lenks{e_i}\in F_i(e)$ then~$\pi_i(X_{\lenks{e_i}})$ was already defined and we may continue inductively as above by defining $F_i(\lenks{e_i})$ analogously. Otherwise, choose an arbitrary ordering~$\pi_i(X_{\lenks{e_i}})$ and again continue as above by defining $F_i(\lenks{e_i})$. Repeat this procedure iteratively until~$\pi_i$ has been defined on all of $E(T_i)$.

    This procedure clearly terminates since every tree has finitely many edges, defining~$\pi_i$ for every edge of~$E(T_i)$; recall that if~$(f,f')$ and~$(e,f)$ are~$\omega_i$-linked then so is~$(e,f')$, i.e., the relation is transitive. Hence, we do not ``doubly define"~$\pi_i$.

    Finally, the proof of the claim follows from the transitivity of strong~$\Omega$-knitwork immersion. To see this let~$(e,f')$ be~$\omega_i$-linked and assume that~$\stitch(\GGG_i;\pi_i,X^{f'})$ does not strongly immerse in $\stitch(\GGG_i;\pi_i,X^e)$. If there is no~$f' \in E(T_i)$ distinct from~$e$ and~$f$ such that~$(e,f)$ and~$(f,f')$ are~$\omega_i$-linked, then the definition of~$\pi_i$ and \cref{claim:thm_wqo_bounded_carvingwidth_knitworks_red_linked} yield a contradiction. Thus let~$e,f'$ be chosen so that the maximal set~$F \subseteq E(T_i)$ with~$f \in F$ if and only if~$(e,f),(f,f')$ are~$\omega_i$-linked has minimal cardinality; by the above~$\Abs{F} \geq 1$. Let~$f \in F$, then~$(e,f),(f,f')$ are~$\omega_i$-linked and by the minimality choice of~$e,f'$ it follows that~$\stitch(\GGG_i;\pi_i,X^{f'}) \hookrightarrow \stitch(\GGG_i;\pi_i,X^f)$ and~$\stitch(\GGG_i;\pi_iX^f) \hookrightarrow \stitch(\GGG_i;\pi_i,X^e)$. By \cref{lem:knitwork_imm_is_quasi_order} strong~$\Omega$-knitwork immersion is transitive, and thus $\stitch(\GGG_i;\pi_i,X^{f'}) \hookrightarrow \stitch(\GGG_i;\pi_i,X^e)$; a contradiction to the choice of~$e,f'$. The claim follows.
\end{claimproof}

By collecting all of the maps~$\pi_i$ of \cref{claim:thm_wqo_bounded_carvingwidth_knitworks_red_pi_i} we may define~$\pi$ for~$E(\FFF)$ via~$\restr{\pi}{E(T_i)} = \pi_i$. Finally, in light of the \cref{cubic-tree-lemma} on Cubic Trees, we define a quasi-order~$\preceq$ on~$E(\FFF)$ as follows.

 \begin{itemize}
        \item[$\preceq:$] Let~$e=(u,v),e'=(u',v') \in E(\FFF)$ with~$e \in E(T_i)$ and~$e' \in E(T_j)$ for some~$i,j \in \N$. Then we define~$e\preceq e'$ if and only if~$\stitch(\GGG_i;\pi,X^e) \hookrightarrow \stitch(\GGG_j;\pi,X^{e'})$ by strong~$\Omega$-knitwork immersion.
\end{itemize}

We next verify that the above-defined objects satisfy the conditions of the \cref{cubic-tree-lemma} on Cubic Trees. By definition of~$\preceq$ and the choice of~$\pi$ respecting \cref{claim:thm_wqo_bounded_carvingwidth_knitworks_red_pi_i} we derive that~$e \preceq f$ whenever~$(e,f)$ is~$\omega$-linked (in particular $e,f$ are edges of the same tree in $\FFF$).

We start with the root edges.
    \begin{claim}
        The root edges are not well-quasi-ordered by $\preceq$.
    \end{claim}
    \begin{claimproof}
        If the root edges were well-quasi-ordered, then this would imply the existence of~$1\leq i<j$ such that $\stitch(\GGG_i; \pi,X^{e_i}) \hookrightarrow \stitch(\GGG_j; \pi,X^{e_j})$ by strong $\Omega$-knitwork immersion, where~$\pi(X^{e_i}) = \pi_{G_i}(x_i)$ and~$\pi(X^{e_j}) = \pi_{G_j}(x_j)$ by definition. By the \cref{def:stitching_knitwork} of up-stitches and the fact that~$x_i \notin \dom(\Phi_i) \cup \dom(\mu_i)$, we derive that~$\stitch(\GGG_i; \pi,X^{e_i}) \cong \GGG_i$, as well as $\stitch(\GGG_j;\pi,X^{e_j}) \cong \GGG_j$. This contradicts the assumption that the initial sequence $(\GGG_i)_{i\in \N}$ was bad.
    \end{claimproof}

 Next, we deal with the leaf edges.
    \begin{claim}\label{claim:thm_wqo_bounded_carvingwidth_knitworks_red_leafs}
        The leaf edges are well-quasi-ordered by $\preceq$.
    \end{claim}
    \begin{claimproof}
        For any $i \in \N$ and any leaf edge~$e=(t,t')\in E(T_i)$ note that the graph~$G_i^e\coloneqq G_i^{X_e}$ of~$\stitch(\GGG_i;\pi,X^e)$ consists of exactly two vertices: namely the \emph{central vertex}~$v \coloneqq \ell_i^{-1}(t')$ and the newly introduced up-stitch vertex~$x^*$ say. The claim follows by \cref{lem:wqo_finitely_many_vertices}, noting that all the degrees are bounded by $k$.
    \end{claimproof}

We prove the last missing ingredient to apply the \cref{cubic-tree-lemma} on Cubic Trees.

\begin{claim}
    There is no infinite strictly decreasing sequence~$(f_i)_{i \in \N}$ for $\preceq$.
\end{claim}
\begin{claimproof}
    Assume the contrary; let~$\GGG_i = (G_i,\pi,x_i), \mu_i, \m_i, \Phi_i)$ be such that $f_i \in E(G_i)$ for every $i \in \N$ and without loss of generality (by switching to a respective infinite subsequence) assume that for every pair of distinct $i,j \in \N$, $G_i \neq G_j$. For every $i \in \N$ let $\HHH_i \coloneqq \stitch(\GGG_i;\pi,X^{f_i})$ be the respective up-stitch with up-stitch vertex $h_i^*$ and write $\HHH_i = \left((H_i,\pi,h_i^*), \nu_i, \n_i, \Psi_i\right)$ for convenience. Let~$m\coloneqq \Abs{H_1} = \Abs{E(H_1)} + \Abs{V(H_1)}$. Then, by assumption on $(f_i)_{i\in \N}$ being strictly decreasing with respect to $\preceq$, for every~$ j>1$ we derive~$\HHH_j \hookrightarrow \HHH_1$ by definition of $\preceq$ implying~$\Abs{H_j} \leq m$. Since there are (up to isomorphism) only finitely many Eulerian digraphs~$G$ satisfying~$\Abs{G} \leq m$, by the pigeonhole principle there is an infinite subsequence indexed by~$I \subseteq \N$ such that~$H_i \cong H_j$ for every~$i \in I$. Let the vertices of~$H_i$ be given by~$\{u_1^i,\ldots,u_t^i\}$ with the respective isomorphisms mapping~$u_\ell^i$ to~$u_\ell^j$ for every~$i,j \in I$ and every~$1 \leq \ell \leq t$. Similarly, by the pigeonhole principle, there is~$I_2 \subseteq I$ such that~$\dom(\nu_i) = \dom(\nu_j)$, $\dom(\n_i) = \dom(\n_j)$ and~$\dom(\Psi_i) = \dom(\Psi_j)$ (using the obvious identifications given the bijection on the vertices) where further~$\mu_i(u_\ell^i) = \mu_j(u_\ell^j)$---if defined---and $\n_i(u_\ell^i) = \n_j(u^j_\ell)$---if defined---for every~$i,j \in I_2$ and every~$1 \leq \ell \leq t$. Thus, since the sequence is strictly decreasing with respect to~$\preceq$, there must exist~$1 \leq \ell \leq t$ such that~$\Psi_j(u_\ell^j) \ll \Psi_i(u_\ell^i)$ for every~$j > i$ with~$i,j \in I_2$. Then this witnesses an infinite strictly decreasing sequence of elements in~$\Omega$ with respect to $\ll$; a contradiction to~$\Omega$ being a well-quasi-order. The claim follows.
\end{claimproof}

Combining the previous claims, the \cref{cubic-tree-lemma} on Cubic Trees implies the existence of an infinite sequence~$(f_i)_{i\in\N}$ of non-leaf edges such that the sequence forms an antichain with respect to~$\preceq$ whereas~$(\lenks{f_i})_{i\in \N}$ and~$(\riets{f_i})_{i\in \N}$ form chains with respect to~$\preceq$.  By switching to an infinite subsequence similar to above, we may assume that~$f_i \in E(T_i)$ for every~$i \in \N$. We are left to refute the existence of said antichain by proving that we find a pair~$f_i,f_j$ such that we can ``knit'' the~$\Omega$-knitwork immersions given on their left and right children to derive that~$f_i \preceq f_j$ as a contradiction.
\smallskip

We could again provide the respective construction by hand using \cref{subsec:knitting_immersions}, but it is easily derived from the decomposition \cref{lem:knitting_immersions_decomposition_handson_located} as follows. By construction for every $i\in \N$ we have that $\GGG_i^f$ consists of the rooted Eulerian digraph $\bar G_i^f = (G_i^f,\pi,f_i^*)$ with sleeve $(\mu_i^{X_{f_i}}, \m_i^{X_{f_i}}, \Phi_i^{X_{f_i}})$. Recall \cref{def:focus}. For every $i \in \N$ we define a new $\Omega$-knitwork $\GGG_i^{f,*} = \GGG_i^f(X^{\lenks{f_i}}, X^{\riets{f_i}})$, i.e., the $(X^{\lenks{f_i}}, X^{\riets{f_i}})$-focus. Finally, note that $\torso(\bar G_i^{f,*})$ consists of at most three vertices, and $\stitch(\bar G_i^{f,*};\pi,f^*)$ consists of at most two vertices. 

Thus using \cref{lem:wqo_finitely_many_vertices} and applying \cref{obs:wqo_yields_infinite_chain} twice, there exists an infinite index set $J\subseteq \N$ such that $(\torso(\bar G_i^{f,*}))_{i \in J}$  as well as $(\stitch(\bar G_i^{f,*};\pi,f^*))_{i \in J}$ are chains with respect to strong $\Omega$-knitwork immersion.

Further, by assumption of the Cubic Tree \cref{cubic-tree-lemma} and applying \cref{obs:wqo_yields_infinite_chain} twice, there is an infinite index set $I \subseteq J$ such that the up-stitches of the left child edges $\left(\stitch(\GGG_i^{f,*};\pi,X^{\lenks{f_i}})\right)_{i \in I} $ as well as the up-stitches of the right child edges $\left(\stitch(\GGG_i^{f,*};\pi,X^{\riets{f_i}})\right)_{i \in I}$, are chains with respect to strong $\Omega$-knitwork immersion. By \cref{lem:knitting_immersions_decomposition_handson_located} we derive that $(\GGG_i^f)_{i \in I}$ is a chain with respect to strong $\Omega$-knitwork immersion which in turn implies that $(f_i)_{i \in I}$ is a chain with respect to $\preceq$ contradicting the assumption that~$(f_i)_{i\in\N}$ was an antichain; this concludes the proof.
\end{proof}

We are ready to derive \cref{thm:wqo_bounded_carvingwidth_knitworks} from \cref{thm:wqo_bounded_carvingwidth_knitworks_red}.

\begin{proof}[Proof of \cref{thm:wqo_bounded_carvingwidth_knitworks}]

Let $\mathbf{G}(k,\ell;\Omega)$ be as in the \cref{thm:wqo_bounded_carvingwidth_knitworks} and fix $\Omega=(V(\Omega),\preceq)$. We derive the proof from \cref{thm:wqo_bounded_carvingwidth_knitworks_red} as follows. Throughout this proof, whenever we write $\hookrightarrow$ we mean the $\Omega$-knitwork immersion to be strong and stable.

\begin{claim}
    The class $\stitch(\mathbf{G}(k, \ell; \Omega))$ is well-quasi-ordered by stable strong $\Omega$-knitwork immersion.
\end{claim}
\begin{claimproof}
    Let $\HHH \in \stitch{\mathbf{G}(k, \ell; \Omega)}$. By \cref{def:stitch_class} of the class there exists an $\Omega$-knitwork $\GGG \in \mathbf{G}(k, \ell; \Omega)$ with $\bar G = (G,\pi,X_1,\ldots,X_\ell)$ and some $1 \leq j \leq \ell$ such that $\HHH \cong \stitch(\GGG;\pi,X_j)$. In particular, since $\cw{\bar G} \leq k$ it follows $\Abs{X_j}\leq k$ implying $\Abs{V(H)} \leq k+1$. Note that $H$ may have degree $1$ vertices (it is quasi-Eulerian). The claim now follows from \cref{lem:wqo_finitely_many_vertices}, noting that only the down-stitch vertex may be contained in $\dom(\mu_\HHH)$, whence $\Delta = k$ does the trick.
\end{claimproof}

\begin{claim}
    $\torso(\mathbf{G}(k, \ell; \Omega))$ is well-quasi-ordered by stable strong $\Omega$-knitwork immersion.
\end{claim}
\begin{claimproof}
    If $\ell=1$ then the claim follows from \cref{obs:torso} and \cref{thm:wqo_bounded_carvingwidth_knitworks_red}; assume $\ell \geq 2$.
    
    Let $\HHH \in \torso(\mathbf{G}(k, \ell; \Omega))$ with $\HHH=\big((H,\pi,x_1^*,\ldots,x_\ell^*),\mu,\m,\Phi\big)$, then $\HHH =\torso(\GGG)$ with up-stitch vertices $x_1^*,\ldots,x_\ell^*$ say, for some $\GGG \in \mathbf{G}(k, \ell; \Omega)$. 
    By repeated application of \cref{obs:up-stitch_of_well-linked_is_well-linked} the torso is again well-linked and $x^* \notin \dom(\Phi_H) \cup \dom(\mu_H)$ by definition. 
    
    Let $\Omega_\ell = (\{0,\ldots,\ell\},=)$ be the well-quasi-order on $\ell$ elements where no two distinct elements are comparable. Let $\Omega^+ = \Omega \times \Omega_\ell$  then $\Omega^+$ is again a well-quasi-order by \cref{obs:wqo_of_tuples}. Fix some element $
\star \in V(\Omega)$. We define the planted $\Omega^+$-knitwork $\HHH^+ \coloneqq  \big((H,\pi,x_1^*),\mu,\m,\Psi\big)$ via $\Psi(x_i^*) = (\star,i)$ for every $2 \leq i \leq \ell$ and $\Psi(v) = (\Phi(v),0)$ otherwise if  and only if $v \in \dom(\Phi)$ (in particular $\Psi$ is not defined on $x_1^*$). 

Let $\HHH_1,\HHH_2 \in \torso(\mathbf{G}(k, \ell; \Omega))$ with up-stitch vertices $x_1^*,\ldots,x_\ell^*$ and $y_1^*,\ldots,y_\ell^*$ respectively. One easily verifies that $\gamma:\HHH_1^+\hookrightarrow \HHH_2^+$ implies $\gamma:\HHH_1 \hookrightarrow \HHH_2$. To see this note that $\mu(x_i^*) = \pi(x_i^*)$, and $x_i^* \notin \dom(\Phi)$ for every $1 \leq i \leq \ell$ whence we do not lose any information (and equivalently for $y_i^*$). Due to $\Psi(x_i^*)=(\star,i) = \Psi(y_i^*)$ and the nature of $\Omega_\ell$, \crefdef{def:knitwork_immersion}{6} makes sure that the $i$-th up-stitch vertex of $\HHH_1$ is mapped to the $i$-th up-stitch vertex of $\HHH_2$ for every $2 \leq i \leq \ell$. Similarly, \crefdef{def:knitwork_immersion}{3} makes sure that $\pi(x_i^*)$ is mapped to $\pi(y_i^*)$, i.e., $\gamma$ satisfies \crefdef{def:rooted_immersion}{3} for $\bar H_1 \hookrightarrow \bar H_2$.

Finally, let $(\HHH_i)_{i \in \N}$ be a sequence of $\Omega$-knitworks with $\HHH_i \in \torso(\mathbf{G}(k, \ell; \Omega))$ for every $i \in \N$. Then $(\HHH_i^+)_{i \in \N}$ is a sequence of planted $\Omega^+$-knitworks with $\HHH_i^+ = \big((H_i,\pi,x_i^*),\mu_i,\m_i,\Psi_i\big)$ such that $x_i^* \notin \dom(\mu_i) \cup \dom(\Psi_i)$ by construction for every $i \in \N$. Then by \cref{thm:wqo_bounded_carvingwidth_knitworks_red} (for $\Omega^+$) there exists $1 \leq i < j$ such that $\HHH_i^+ \hookrightarrow \HHH_j^+$. By the above then $\HHH_i \hookrightarrow \HHH_j$ concluding the proof of the claim.
\end{claimproof}

The theorem follows by \cref{lem:knitting_immersions_decomposition}.
\end{proof}

Finally, note that \cref{thm:wqo_bounded_carvingwidth_knitworks_gen} (and in turn \cref{thm:wqo_bounded_carvingwidth_knitworks}) can be seen as a ``meta theorem'' in the following sense: Intuitively speaking, given a class $\mathbf{G}(\ell)$ of quasi-Eulerian digraphs such that there exists $k\in \N$ such that for every $G \in \mathbf{G}$ there exists a linked ``carving'' $(T,\ell)$ of $G$ into ``pieces''---here leaves $\ell^{-1}(t)$ for a leaf $t \in V(T)$ should be seen as quasi-Eulerian subdigraphs---for which we have already proved that they are well-quasi-ordered by (strong) immersion, and such that the cut-edges are well-linked inside the pieces. Then each such piece is an element of a respective well-quasi-order $\Omega$. In particular one can construct an $\Omega$-knitwork from $G$ of carving-width $\leq k$, by carefully replacing each carved out piece by a single vertex $v$ of degree $\leq k$ where each such vertex comes with a predefined ordering of $\mu(v)$---representing the carved cut---and a well-linked link of feasible linkages $\m(v)$ representing the possible routings between the cut-edges in said piece. Then we immediately derive that $\mathbf{G}(\ell)$ is well-quasi-ordered from \cref{thm:wqo_bounded_carvingwidth_knitworks}. We will provide a stronger version of \cref{thm:wqo_bounded_carvingwidth_knitworks} in future work where we discuss this further (for note that as it is stated here, the pieces do not admit $\Omega$-sleeves; compare with \cref{lem:knitting_immersions_decomposition_handson_located}).

\smallskip

We conclude the section with a proof of \cref{thm:intro}, which we derive as a corollary to the above; recall that we assumed our digraphs to be weakly connected, an assumption that we eliminate next.

\begin{corollary}\label{cor:bd_cw_wqo}
    Let $k \in \N$. The class $\mathbf{G}(k)$ of (possibly disconnected) Eulerian digraphs of carving width at most $k$ is well-quasi-ordered by strong immersion.
\end{corollary}
\begin{proof}
   Note that graphs in $\mathbf{G}(k)$ may admit loops and may consist of disjoint connected components. The claim then follows by combining \cref{thm:wqo_bounded_carvingwidth_knitworks_gen} together with \cref{lem:wqo_assume_no_loops}--- noting that $\mathbf{G}(k,\ell;\Omega)$ is $\Omega$-independent for any choice of well-quasi-order $\Omega$---and applying \cref{thm:higman} to the set of components of the respective Eulerian digraphs. This concludes the proof.
\end{proof}

\paragraph{The Importance of Well-linkedness.} It turns out that the \emph{well-linkedness} of the~$\Omega$-knitworks is crucial for a proof of \cref{thm:wqo_bounded_carvingwidth_knitworks}. Without the assumption, ``linkedness'' in \cref{cor:linked_ebw} can not directly be used to ``lift''~$\Omega$-knitwork immersions from children to parents through knitting at cuts. This is due to the restrictions imposed by the feasible linkages encoded by $\m$. To see this let~$Y$ be a rooted cut in some~$\Omega$-knitwork~$\GGG=((G,\pi,X), \mu ,\m,\Phi)$ and let~$\LLL$ be a~$\{\rho(X),\rho(Y)\}$-linkage in~$G$. One would hope to conclude~$\stitch(\GGG;\pi,\bar Y) \hookrightarrow \GGG$ for a suited order~$\pi(\bar Y)$. Unfortunately, it may happen that two edge-disjoint paths~$L_1,L_2 \in \LLL$ ``cross'' in a common vertex~$v \in \dom(\m)$ which implies that we may not use these paths to define the natural immersion if~$\{\tau(L_1),\tau(L_2)\}$ does not appear in a matching of $\m(v)$ (this would violate (4) of \cref{def:knitwork_immersion}). Moreover, given arbitrary~$\Omega$-knitworks it may happen that there exists no feasible linkage that can be used to lift immersions. This is in stark contrast with the undirected case using standard minors. Further, it is also highly different to the undirected case for immersion, since in the undirected setting edge-disjoint paths may cross freely in vertices which would rather correspond to~$\m(v) = \operatorname{Match}(\rho^-(v),\rho^+(v))$, i.e., there are no further restrictions on the linkages, which is a much easier case. The following is an easy ``counterexample'' to the version of \cref{thm:wqo_bounded_carvingwidth_knitworks_red} where we do \emph{not} impose well-linkedness similar to \cref{thm:antichain}. In essence, we encode alternating paths in an $\Omega$-knitwork by imposing restrictions on $\m$.

\begin{example}
    Let $V_i = \{v_0,\ldots,v_i\}$ for every $i \in \N$. Let $G_i$ be the graph on $V_i$ where we add back and forth edges between subsequent vertices $v_j,v_{j+1}$ for $1 \leq j < i$. Finally let $\Omega=(\{0,1\},=)$ be the well-quasi-order on two labels that are not comparable. Let $\Phi_i(v_0) = 0$ and $\Phi_i(v_i) = 1$. Let $\mu_i$ be arbitrarily defined for every vertex of $V(G_i)$ and set $\mu_i(v_j)$ to be the two-paths that start and end in $v_{j-1}$ and start and end in $v_{j+1}$ respectively. In particular we cannot pass through vertices with immersions; the links are not well-linked. One easily verifies that $(\GGG_i)_{i\geq 3}$ is an antichain. While this example may be slightly trivial as we basically encode our graphs to be disconnected (or behave like alternating paths), one easily extends this idea to more interesting examples (even on higher degree).
\end{example}

\section{Allowing a vertex of large degree}
\label{sec:bded_treewidth}

As seen in \cref{cor:antichain_tw}, the class of Eulerian digraphs of bounded treewidth is \emph{not} well-quasi-ordered by strong immersion. If we additionally require the maximum degree of the class to be bounded, \cref{thm:qualitative_equivalence_of_tw_and_cw} together with \cref{cor:bd_cw_wqo} imply that it \emph{is} well-quasi-ordered by strong immersion, proving \cref{thm:intro_2}. Thus, any bad sequence with respect to \emph{strong} immersion of Eulerian digraphs of bounded treewidth must admit unbounded maximum degree. In this section we complement \cref{cor:antichain_tw},  proving that the antichain given in the proof of \cref{thm:antichain} is \emph{not} an antichain for weak immersion, marking a first step towards a proof of the following.

\begin{conjecture}\label{thm:wqo_bounded_tw}
    Let $k \in \N$. The class of Eulerian digraphs of treewidth at most $k$ is well-quasi-ordered by immersion.
\end{conjecture}

It turns out that \cref{thm:wqo_bounded_tw} follows as a corollary from a different result we are pursing in future work, which is why we do not provide a full proof for it here and leave it as a conjecture. It turns out that a full proof of \cref{thm:wqo_bounded_tw} needs a non-trivial amount of work, exceeding the scope of this paper, where most of the techniques we developed can be used to prove a stronger result and thus we will restrict our attention in this section to a small subclass of Eulerian digraphs of unbounded degree. This brings to light a nice trick, linking bounded treewidth to bounded carving width in a way that may be of independent interest.

\smallskip

Recall the \cref{def:knitwork,def:knitwork_immersion} of $\Omega$-knitwork and $\Omega$-knitwork immersion. We slightly adjust these definitions to accommodate for ``large'' degrees.
\begin{definition}[$(\Omega,\Delta)$-knitworks]\label{def:knitwork_large_degree}
    Let $\Omega$ be a well-quasi-order and let $\Delta \in 2\N$. An $(\Omega,\Delta)$-knitwork is an $\Omega$-knitwork $\GGG= \big((G,\pi,X_1,\ldots,X_\ell), \mu, \m,\Phi\big)$ for some $\ell \geq 1$ satisfying the following. Let $Y \subseteq V(G)$ be the set of vertices of degree greater than $\Delta$, then $\dom(\mu) \cap Y = \emptyset$ and $Y \cap \bigcup_{i=1}^\ell X_i= \emptyset$.
\end{definition}
\begin{remark}
Note that $\dom(\m)\cap Y = \emptyset$ by \cref{def:knitwork}.

    If we allow for unbounded degree in our class of interest, then \crefdef{def:knitwork_immersion}{3} can in general not be guaranteed (see also the remark following said definition); when dealing with $(\Omega,\Delta)$-knitworks this problem is mitigated by ``ignoring'' high degree vertices. We want to emphasise that although this definition may seem strange at first, it makes intuitive sense: note that the reason for why we need $\mu$ and $\m$ stems from the fact that we want to cut graphs along cuts and replace pieces by placeholder vertices keeping track of their respective immersion-type and the ways we are allowed to route through them. Until now it seems plausible that we will only cut our graphs along cuts of a predetermined size and thus the degree of vertices for which we will define $\mu$ and $\m$ will be bounded. We leave this for future work.
\end{remark}

The following is a rephrasing of \cref{lem:wqo_finitely_many_vertices}.

\begin{observation}
    Let $\Omega$ be a well-quasi-order and let $\Delta \in 2\N$. Let $k,\ell \geq 1$ and let $\CCC(\ell,k,\Delta)$ be a class of $(\Omega,\Delta)$-knitworks of index $\ell$ and on $\leq k$ vertices. Then $\CCC(\ell,k,\Delta)$ is well-quasi-ordered by strong $\Omega$-knitwork immersion.
\end{observation}

Next we generalise the definition of treewidth to $\Omega$-knitworks in the obvious way.

\begin{definition}[Treewidth for Knitworks]\label{def:knitworks_treewidth}
    Let $\bar G=(G,\pi,X_1,\ldots,X_\ell)$ be a rooted digraph for some $\ell \geq 1$ and $\torso(\bar G)$ its torso. Let $\omega$ be the treewidth of the undirected underlying graph of  $\torso(\bar G)$. Then we define the treewidth of $\bar G$ via $\tw{\bar G} \coloneqq \max(\omega,\Abs{X_1},\ldots,\Abs{X_\ell})$. Given an $\Omega$-knitwork $\GGG$ we define $\tw{\GGG} \coloneqq \tw{G}$.
\end{definition}

Finally, note that the \cref{def:knitworks_carving,def:knitworks_treewidth} together with \cref{thm:qualitative_equivalence_of_tw_and_cw} lift to the following qualitative equivalence.
\begin{lemma}\label{lem:qualitative_equivalence_of_tw_and_cw_knitworks}
Let $\Omega$ be a well-quasi-order and $\Delta \in \N$. Let $\GGG$ be an $\Omega$-knitwork such that $G$ is of degree at most $\Delta$. Then $\frac{2}{3} \tw{\GGG} \leq \cw{\GGG} \leq \Delta \cdot \tw{\GGG}$.
\end{lemma}
\begin{proof}
   One easily verifies that $\cw{\GGG} = \max(\cw{\torso(\bar G)},\Abs{X_1},\ldots,\Abs{X_\ell})$. By definition, the underlying graph of $\torso(\bar G)$ is Eulerian. Thus $\frac{2}{3} \tw{\torso(\bar G)} \leq \cw{\torso(\bar G)} \leq \Delta \cdot \tw{\torso(\bar G)}$ by \cref{thm:qualitative_equivalence_of_tw_and_cw}. Let $\omega \coloneqq \tw{\torso(\bar G)}$. We derive that $$\frac23 \max(\omega,\Abs{X_1},\ldots,\Abs{X_\ell}) \leq \max(\cw{\torso(\bar G)},\Abs{X_1},\ldots,\Abs{X_\ell}) \leq \Delta \cdot \max(\omega,\Abs{X_1},\ldots,\Abs{X_\ell}),$$
    concluding the proof.
\end{proof}

The base case for \cref{thm:wqo_bounded_tw} extending \cref{thm:wqo_bounded_carvingwidth_knitworks} is given by Eulerian digraphs of bounded treewidth admitting only a single vertex of ``large'' degree. In particular, that vertex may be an obstruction to the carving width being bounded; compare with \cref{thm:qualitative_equivalence_of_tw_and_cw}. To this extent, we define the following.
\begin{definition}
    Let $k,\ell \geq 1$, $d \in \N$ and $\Delta \in 2\N$. Let $\Omega$ be a well-quasi-order. We define $\mathbf{T}(k,d,\Delta,\ell;\Omega)$ to be the class of well-linked $(\Omega,\Delta)$-knitworks of index $\ell$ of treewidth at most $k$ admitting at most $d$ vertices of degree greater than $\Delta$. 
\end{definition}

We highlight the following observation following from the above \cref{def:knitworks_treewidth}.
\begin{observation}\label{obs:bounded_tw_class_then_torso_stitch_bounded_tw}
    Let $k,\ell \geq 1$, $d \in \N$ and $\Delta \in 2\N$. Let $\Omega$ be a well-quasi-order. Let $\mathbf{G}(\ell;\Omega)$ be a class of $\Omega$-knitworks of index $\ell$ such that for every $\GGG \in \mathbf{G}(\ell;\Omega)$ it holds $\tw{\GGG} \leq k$. Then 
    \begin{enumerate}
        \item for every $\HHH \in \stitch(\mathbf{G}(\ell;\Omega)) \cup \torso(\mathbf{G}(\ell;\Omega))$ it holds $\tw{\HHH} \leq k$, and\label{obs:bounded_tw_class_then_torso_stitch_bounded_tw:1}
        \item $\torso(\mathbf{T}(k,\ell,d,\Delta;\Omega)) \subseteq \mathbf{T}(k,\ell,d,\Delta;\Omega)$.\label{obs:bounded_tw_class_then_torso_stitch_bounded_tw:2}
    \end{enumerate}
\end{observation}
\begin{remark}
    Note that the torsos remain well-linked due to \cref{obs:up-stitch_of_well-linked_is_well-linked}, whereas stitches may \emph{not} be well-linked in general.
\end{remark}

Note first that \cref{thm:wqo_bounded_carvingwidth_knitworks} together with \cref{thm:qualitative_equivalence_of_tw_and_cw} imply the following.

\begin{corollary}\label{cor:wqo_bounded_cw_Delta_version}
    Let $k,\ell \geq 1$ and $\Delta \in 2\N$. Let $\Omega$ be a well-quasi-order. Then $\mathbf{T}(k,0,\Delta,\ell;\Omega)$ is well-quasi-ordered by \emph{strong} $\Omega$-knitwork immersion.
\end{corollary}  
In contrast to \cref{cor:wqo_bounded_cw_Delta_version}, when allowing for a single vertex of unbounded degree, the class is not well-quasi-ordered by strong immersion anymore. In particular, \cref{cor:antichain_tw} implies the following.

\begin{corollary}\label{cor:one_vertex_not_wqo_strong_Delta}
    Let $k \geq 2$, $d,\ell\geq 1$, $\Delta\geq 4$. Let $\Omega$ be a well-quasi-order. Then $\mathbf{T}(k,d,\Delta,\ell;\Omega)$ is \emph{not} well-quasi-ordered by \emph{strong} $\Omega$-knitwork immersion.
\end{corollary}
\begin{remark}
    \cref{cor:one_vertex_not_wqo_strong_Delta} is best possible in the sense that, non-trivial connected Eulerian digraphs admit undirected treewidth at least $2$.
\end{remark}

As claimed above, when omitting the assumption on the immersions to be strong, we conjecture that a positive result complementing \cref{cor:one_vertex_not_wqo_strong_Delta} holds true. 

\begin{conjecture}\label{thm:bounded_tw_basecase}
    Let $k,d,\ell \geq 1$, $\Delta \in 2\N$ and let $\Omega=(V(\Omega),\preceq)$ be a well-quasi-order. Then $\mathbf{T}(k,d,\Delta,\ell;\Omega)$ is well-quasi-ordered by $\Omega$-knitwork immersion.
\end{conjecture}

We only settle the base case of a single vertex in this section, i.e., the following is the main theorem of this section.

\begin{theorem}\label{thm:bounded_tw_basecase_1}
    Let $k,\ell \geq 1$, $d\leq 1$, and $\Delta \in 2\N$. Let $\Omega=(V(\Omega),\preceq)$ be a well-quasi-order. The class~$\mathbf{T}(k,d,\Delta,\ell;\Omega)$ is well-quasi-ordered by $\Omega$-knitwork immersion.
\end{theorem}

While this may seem to be a trivial extension at first glance, it does not follow from the bounded carving width result effortlessly. The following technical lemma lies at the heart of the proof and may be of its own interest. Given two digraphs $H,G$ with $V(H) \subseteq V(G)$ we define $G+H\coloneqq (V(G),E(G)\cup E(H))$ possibly introducing parallel edges.

\begin{lemma}\label{lem:matching_degree_ones_low_cw}
    Let $G$ be a quasi-Eulerian digraph and let $(T,\ell)$ be a carving of $G$ witnessing its carving width $k \in \N$. Then there exists a perfect matching $M \in \operatorname{Match}(V^-(G),V^+(G))$ such that $(T,\ell)$ is a carving of $G+M$ of width $\leq 2k$.  
\end{lemma}
\begin{proof}
Let $G$ be a quasi-Eulerian digraph of defect $\tau \in 2\N$  with $\cw{G} \leq k$ and let $(T,\ell)$ be a carving witnessing said width. First note that if $\tau \leq 2k$ then the lemma is trivially satisfied by choosing an arbitrary perfect matching $M \in \operatorname{Match}(V^-(G),V^+(G))$ as we only add $k$ edges, hence increasing the width of the carving by at most $k$. Henceforth, assume that $\tau \geq 2k+2$. Let $V \subseteq V(G)$ be the set of degree-one vertices in $G$. Let $G^*$ be obtained from $G$ by identifying $V$ to a single vertex $v^*$. Since $\Abs{V^+(G)} = \Abs{V^-(G)}$, the graph $G^*$ is Eulerian. Let $\CCC$ be a circle cover of $G^*$ which exists by \cref{obs:covering_eulerian_digraphs} and \cref{obs:linkage_gives_linear_linkage}. Then $\Abs{\CCC} \geq \frac{\deg(v^*)}{2}$, since every circle visits the same vertex at most once; let $\CCC^* \subseteq \CCC$ be minimal such that every edge incident to $v^*$ is part of a cycle in $\CCC^*$. Clearly $\Abs{\CCC^*} = \frac{\deg(v^*)}{2}$. Note further that $E(G^*) = E(G)$ by construction, therefore every circle $C \in \CCC^*$ corresponds to a unique $(V^+(G),V^-(G))$-path in $G$ in the obvious way: starting with the unique edge in $\rho^+(v^*)\cap E(C)$ and ending with the unique edge in $\rho^-(v) \cap E(C)$. Thus, $\CCC^*$ unambiguously defines a strong linear~$(V^+(G),V^-(G))$-linkage~$\LLL^*$ of order~$\frac{\deg(v^*)}{2}$ in~$G$. 

 \begin{claim}\label{claim:matching_degree_ones_low_cw}
     Let $e \in E(T)$ and let $\{X_e,\bar X_e\}$ denote the partition induced by $e$. Then $\operatorname{bal}_G(X_e) \leq k$.
 \end{claim}
 \begin{claimproof}
     Assume the contrary and without loss of generality let $\operatorname{bal}_G(X_e) = \Abs{X_e \cap V^+(G)} - \Abs{X_e \cap V^-(G)} \geq k+1$ say. This implies that at least $k+1$ paths of $\LLL^*$ have one endpoint in $X_e$---namely the vertex in which the path starts---and one endpoint in $\bar X_e$. Since $X_e$ induces a $\{X_e,\bar X_e\}$-cut in $G$, \cref{lem:Menger_for_Euler} implies that $\delta(X_e)\geq k+1$ as a contradiction to $\w(e) \leq k$ (recall \cref{def:carving}). The claim follows.
 \end{claimproof}

In light of \cref{claim:matching_degree_ones_low_cw}, for every $e \in E(T)$ we define $\operatorname{bal}_G(e) \coloneqq \operatorname{bal}(X_e)$; this is well-defined by symmetry, i.e.,  \cref{obs:balance_is_symmetrical} implies $\operatorname{bal}_G(e) = \operatorname{bal}(\bar X_e)$. Finally, we use \cref{claim:matching_degree_ones_low_cw} to prove a stronger version of the lemma.

\begin{claim}
     Let $G$ be a quasi-Eulerian digraph and $(T,\ell)$ a carving of $G$ witnessing its carving width $k \in \N$. Let $\tau \geq 2k+2$ be the defect of $G$. Then, there is a matching $M \in \operatorname{Match}(V^-(G),V^+(G))$ such that $(T,\ell)$ is a carving of $G+M$ such that for every $e \in E(T)$ it holds $\w_{G+M}(e) \leq \w_{G}(e) + \operatorname{bal}_G(e)$.
\end{claim}
\begin{claimproof}
    Towards a contradiction assume the contrary and let $G$ be a quasi-Eulerian digraph and $(T,\ell)$ a carving of $G$ witnessing its carving width $k \in \N$ refuting the claim such that
\begin{itemize}
    \item[1.] $G$ admits minimal defect $\tau \geq 2k+2$, and
    \item[2.] subject to 1. $k$ is minimal, and
    \item[3.] subject to 1. and 2. $G$ admits a minimal vertex set, and
    \item[4.] subject to 1. and 2. and 3. $G$ admits a minimal edge-set.
\end{itemize}
By 3. $G$ is connected, for else we could choose a connected component of $G$ refuting minimality (note that each connected component is quasi-Eulerian again). By 4. $G$ is loopless, since loops do not affect the carving width. Note that $\cw{G} \leq k$ implies that the maximum degree of $G$ is $k$. 

We use \cref{claim:matching_degree_ones_low_cw} to refute assumption 1. on the minimality of $G$ above. To this extent, let $L^- \subset V(T)$ such that $\ell^{-1}(L^-) = V^-(G)$. Let $\{t_1,\ldots,t_{\frac{\tau}{2}}\} = L^-$ be an enumeration of the set. For every $j \in \{1,\ldots,\frac\tau 2\}$ and $e^* \in E(T)$ let $T^j_{e^*}$ be the component of $T - e^*$ containing $t_j$. Choose $j \in \{1,\ldots,\frac\tau 2\}$ and $e^* \in E(T)$ such that $T_{e^*}^j$ is edge-minimal satisfying $\ell^{-1}(V(T^j_{e^*})) \cap V^+(G) \neq \emptyset$ (note that $t_j \in \ell^{-1}(V(T^j_{e^*})) \cap V^-(G)$). Let $X^- \coloneqq \ell^{-1}(V(T^j_{e^*})) \cap V^-(G)$ and $X^+ \coloneqq \ell^{-1}(V(T^j_{e^*})) \cap V^+(G)$ then both sets are non-empty by definition. 

Let $\lenks{e^*},\riets{e^*}$ be the two distinct edges of $T^j_{e^*}$, that share a common end with $e^*$; they exist since $T$ is cubic and $e^*$ cannot be a leaf edge. Let $T_{\lenks{e^*}}$ be the component of $T-\lenks{e^*}$ completely contained in $T^j_{e^*}$ and define $T_{\riets{e^*}}$ analogously. Without loss of generality assume that $t_j \in V(T_{\lenks{e^*}})$. By minimality of our choice of $e^*$ above we derive the following.
\begin{itemize}
    \item[$(\star)$] $X^- \subseteq \ell^{-1}(\leaves{T_{\lenks{e^*}}})$ and $X^+ \subseteq \ell^{-1}(\leaves{T_{\riets{e^*}}})$.
\end{itemize}
For otherwise $X^- \cap \ell^{-1}(\leaves{T_{\lenks{e^*}}}) \neq \emptyset$ and $X^+ \cap \ell^{-1}(\leaves{T_{\lenks{e^*}}}) \neq \emptyset$, say (the other case is analogous), and thus $\lenks{e^*}$ would contradict the minimality of $e^*$. By \cref{claim:matching_degree_ones_low_cw} applied to $\lenks{e^*}$ and $\riets{e^*}$ together with $(\star)$ we deduce that $\Abs{X^-} \leq k$ as well as $\Abs{X^+} \leq k$.

\smallskip

Let $M_{e^*} \in \operatorname{Match}(X^-,X^+)$ be a maximum matching and $G_{e^*} \coloneqq G+M_{e^*}$. Then $G_{e^*}$ has defect less than $\tau$ by construction since $\Abs{M_{e^*}} \geq 1$. Further $(T,\ell)$ is a carving for $G_{e^*}$: we claim that for every $e \in E(T)$ the following hold
\begin{itemize}
    \item[(1)] $\w_{G_{e^*}}(e) \leq \w_{G}(e) + \operatorname{bal}_G(e)$, and
    \item[(2)] $\operatorname{bal}_{G}(e) = \operatorname{bal}_{G_{e^*}}(e) + \delta_{G_{e^*}}(X_e) - \delta_{G}(X_e)$
\end{itemize}
By construction (1) is clear for every $e \notin E(T^j_{e^*})$ since the width of these edges did not alter, i.e., $\w_{G_{e^*}}(e) = \w_G(e)$. Similarly (2) follows for every $e \notin E(T^j_{e^*})$ from \cref{obs:balance_is_symmetrical} since $G_{e^*}$ is quasi-Eulerian and we did not add or remove any edges with endpoints in $\bar X_{e^*}$, and therefore  $\operatorname{bal}_G(\bar X_e^*) = \operatorname{bal}_{G_{e^*}}(\bar X_e^*)$ and $\delta_{G_{e^*}}(X_e) = \delta_{G}(X_e)$. Thus, let $e \in T^j_{e^*}$, then either $e \in T_{\lenks{e^*}}$ or $e \in T_{\riets{e^*}}$; without loss of generality assume that $e \in T_{\lenks{e^*}}$ for the other case is symmetrical. Let $\{X_e,\bar X_e\}$ be the partition induced by $e$ and without loss of generality assume that $X_e \subseteq \ell^{-1}(V(T_{e^*})$. By $(\star)$ we derive that $X^+ \cap X_e = \emptyset$.  Let $M^- \subseteq X^-$ be the set of vertices which are part of an edge of $M_{e^*}$. By construction we derive that $\w_{G_{e^*}}(e) = \w_G(e) + M^-\cap X_e$; the claim (1) follows. 
Similarly from $(\star)$ we derive that $\operatorname{bal}_{G_{e^*}}(e) = \Abs{X^-\cap X_e} - \Abs{M^- \cap X_e}$ and since $\operatorname{bal}_{G}(e) = \Abs{X^-\cap X_e}$ we derive that $\operatorname{bal}_{G}(e) = \operatorname{bal}_{G_{e^*}}(e) + \Abs{M^- \cap X_e}$. Clearly $\Abs{M^- \cap X_e} = \delta_{G_{e^*}}(X_e) - \delta_{G}(X_e)$ as these are exactly the new edges that ``add to the cut''; the claim (2) follows.

\smallskip

Since $G_{e^*}$ admits less defect than $G$ it is no counterexample to the claim, whence there is a matching $M' \in \operatorname{Match}(V^-(G_{e^*}),V^+(G_{e^*}))$ such that $(T,\ell)$ is a carving of $G_{e^*}+M'$ such that for every $e \in E(T)$ it holds $\w_{G_{e^*}+M'}(e) \leq \w_{G_{e^*}}(e) + \operatorname{bal}_{G_{e^*}}(e)$. Together with $(1)$ and $(2)$ we derive that $M \coloneqq M' \cup M_{e^*}$ is a perfect matching in $\operatorname{Match}(V^-(G),V^+(G))$ satisfying the claim as a contradiction to $G$ being a counterexample.
\end{claimproof}

The claim immediately implies the lemma concluding the proof.
\end{proof}

Next we introduce \emph{ripping}---an operation that takes a high-degree vertex and ``rips'' it into many vertices of degree one. The intuitive idea behind ripping is as follows: Given a digraph $G$ and a loopless vertex $v \in V(G)$ we define a new digraph $G'$ on the same edge-set as $G$, where all the adjacencies of edges with $v$ are replaced by new adjacencies with degree one vertices.

\begin{definition}[Ripping]\label{def:ripping}
    Let $\bar G=(G,\pi,X_1,\ldots,X_\ell)$ be a rooted Eulerian digraph of index $\ell \geq 1$ where $G=(V,E,\operatorname{inc})$. Let $v \in V(G) \setminus \bigcup_{i=1}^\ell X_i$ be loopless and let $d \coloneqq \frac{\deg(v)}{2}$. Let $v_1^-,\ldots,v_d^-$  and $v_1^+,\ldots,v_d^+$ be distinct new elements that are not part of $V \cup E$. Let $\{e_1^-,\ldots,e_d^-\} = \rho^-(v)$ and $\{e_1^+,\ldots,e_d^+\} = \rho^+(v)$. We define $\rip(G,v) \coloneqq (V',E,\operatorname{inc}')$ via
    \begin{itemize}
        \item  $V' \coloneqq (V \setminus \{v\}) \cup \{v_1^-,v_1^+,\ldots,v_d^-,v_d^+\},$
        \item $\operatorname{inc}' = \operatorname{inc}^- \cup \operatorname{inc}^+,$ with
    \begin{align*}
        \operatorname{inc}^- &\coloneqq \{(e,u) \in \operatorname{inc} \mid u\neq v\} \cup \{(e_i^-,v_i^-) \mid 1 \leq i \leq d\}, \text{ and }\\
        \operatorname{inc}^+ &\coloneqq \{(u,e) \in\operatorname{inc} \mid u \neq v \} \cup \{(v_i^+,e_i^+) \mid 1 \leq i \leq d\}.
    \end{align*}\end{itemize}
Let $\pi_G(v)$ be an ordering of $\rho(v)$ and let $X_v \coloneqq \{v_1^-,v_1^+,\ldots,v_d^-,v_d^+\}$. Then $\rho_G(v) = \rho_{\rip(G,v)}(X_v)$. Fix $\pi(X_v) \coloneqq \pi(v)$ by construction.
We define $\rip(\bar G;\pi(v)) \coloneqq (\rip(G,v),\pi,X_1,\ldots,X_\ell,X_v)$.

We call $\rip(\bar G; \pi(v))$ a \emph{rooted ripping of $v$ (in $\bar G$)} and similarly we call $\rip(G,v)$ a \emph{ripping of $v$ (in $G$)}. We refer to $X_v$ as the \emph{ripped ends (of $v$)}. 
\end{definition}
\begin{remark}
    Clearly $X_v$ is balanced by definition, thus the above is well-defined.
\end{remark}

The following relations are readily extracted from the definition.

\begin{observation}\label{obs:ripping_fundamentals}
    Let $\ell \geq 1$ and let $\bar G=(G,\pi,X_1,\ldots,X_\ell)$ be a rooted Eulerian digraph of index $\ell \geq 1$. Let $v \in V(G) \setminus \bigcup_{i=1}^\ell X_i$ be loopless and let $\pi(v)$ be an ordering of $\rho(v)$. Then 
    \begin{enumerate}
        \item \label{obs:ripping_fundamentals:1}$\bar G'\coloneqq \rip(\bar G;\pi(v))$ is a well-defined rooted Eulerian digraph of index $\ell +1$ and order $k + \delta(v)$ where $k$ is the order of $\bar G$,
        \item \label{obs:ripping_fundamentals:2} $E(G') = E(G)$,  $V(G) \setminus \{v\} \subset V(G')$ and for every $u \in V(G) \cap V(G')$ and every $e \in E(G)$ it holds $\operatorname{inc}_G(e,u) \iff \operatorname{inc}_{G'}(e,u)$ as well as $\operatorname{inc}_G(u,e) \iff \operatorname{inc}_{G'}(u,e)$.
    \end{enumerate}
\end{observation}
\begin{remark}
    By \crefthm{obs:ripping_fundamentals}{2} we derive that if $G$ is loopless then $G'$ is loopless.
\end{remark}

We readily extend the definition of ripping to $\Omega$-knitworks as follows.

\begin{definition}\label{def:ripping_knitwork}
    Let $\Omega$ be a well-quasi-order and $\GGG = \big(\bar G, \mu, \m, \Phi\big)$ an $\Omega$-knitwork of index $\ell \geq 1$. Let $v \in V(G) \setminus \dom(\m)$ be loopless. Let $\pi(v)$ be an ordering of $\rho(v)$ such that $\pi(v) = \mu(v)$ if $v \in \dom(\mu)$ and otherwise arbitrarily fixed. Let $\rip(\bar G; \pi(v))$ be a rooted ripping of $v$ with ripped ends $X_v$.
    We define $\rip(\GGG;\pi(v)) \coloneqq (\rip(\bar G, \pi(v)), \mu', \m', \Phi')$ by letting $\dom(\mu') = \dom(\mu)\setminus\{v\}$, $\dom(\m') = \dom(\m)\setminus\{v\}$ and fixing $\mu' \coloneqq \restr{\mu}{\dom(\mu')}$ as well as  $\m' = \restr{\m}{\dom(\m')}$. Similarly $\Phi'$ is defined to agree with $\Phi$ on $V\setminus \{v\}$ and is not defined on $X_v$.
\end{definition}

Using \cref{lem:matching_degree_ones_low_cw} and \cref{def:ripping}, we are ready to prove \cref{thm:bounded_tw_basecase_1}.

\begin{proof}[Proof of \cref{thm:bounded_tw_basecase_1}.]
 By \cref{lem:knitting_immersions_decomposition} it suffices to prove the lemma for controlled $\Omega$-knitworks. To see this note that by \crefthm{obs:bounded_tw_class_then_torso_stitch_bounded_tw}{2} we derive that $\torso(\mathbf{T}(k,d,\Delta,\ell;\Omega)) \subseteq \mathbf{T}(k,d,\Delta,\ell;\Omega)$ while $\stitch(\mathbf{T}(k,d,\Delta,\ell;\Omega))$ is a class of $\Omega$-knitworks on at most $k+1$ vertices  whence the class is well-quasi-ordered by \cref{lem:wqo_finitely_many_vertices}. Thus, we are left to prove the theorem for the torsos which are controlled, i.e., for the class $\mathbf{T}^\star(k,\ell,d,\Delta;\Omega) \subset \mathbf{T}(k,\ell,d,\Delta;\Omega)$ of controlled $(\Omega,\Delta)$-knitworks.

Let $(\GGG_i)_{i \in \N}$ be a sequence of controlled $\Omega$-knitworks $\GGG_i=(\bar G_i,\mu_i,\m_i,\Phi_i) \in \mathbf{T}(k,d,\Delta,\ell;\Omega)$ for every $i \in \N$; in particular $G_i$ is Eulerian for every $i \in \N$.  We continue by filtering the sequence for a ``nice'' subsequence as follows. If there exists an infinite index set $I \subseteq \N$ such that $\GGG_i \in \mathbf{T}(k,0,\Delta,\ell;\Omega)$ for every $i \in I$, then \cref{cor:wqo_bounded_cw_Delta_version} yields $\GGG_i \hookrightarrow \GGG_j$ for some $i<j$. Thus, assume without loss of generality that $G_i$ contains a vertex $x_i$ of degree greater than $\Delta$ for every $i \in \N$. Since $\Omega$ is a well-quasi-order, applying \cref{obs:wqo_yields_infinite_chain} there exists an infinite index set $I \subseteq \N$ such that $(\Phi_i(x_i))_{i \in \N}$ is a chain; for simplicity assume without loss of generality that $I = \N$. Again by \cref{cor:wqo_bounded_cw_Delta_version} we derive the following.

    \begin{claim}\label{claim:bounded_tw_basecase_1}
        If there exist $\Delta' \in 2\N$ and an infinite index set $I \subseteq \N$ such that $\GGG_i \in \mathbf{T}(k,0,\Delta',\ell;\Omega)$ for every $i \in I$, then there exist $i<j$ such that $\GGG_i \hookrightarrow \GGG_j$.
    \end{claim}

    By \cref{claim:bounded_tw_basecase_1} we derive that for every $\Delta' \geq \Delta$ there are only finitely many $i \in \N$ with $\deg(x_i) \leq \Delta'$. In particular there exists an infinite index set $I \subseteq \N$ such that $(\deg(x_i))_{i \in \N}$ is strictly monotonically increasing. Again for simplicity identify $I \cong \N$. Let $\bar G_i = (G_i, \pi,v_1^i,\ldots,v_\ell^i)$ for every $i \in \N$.

    \begin{claim}\label{claim:bounded_tw_basecase_2}
        There is an infinite index set $I \subseteq \N$ such that $\deg(v_t^i) = \deg(v_t^j)$ for every $i,j \in I$ and every $1 \leq t \leq \ell$.
    \end{claim}
    \begin{claimproof}
        By \cref{def:knitwork_large_degree} of $(\Omega,\Delta)$-knitworks, we derive that $\deg(v_t^i) \leq \Delta$ for every $i \in \N$. Thus by $\ell$ subsequent applications of the pigeonhole principle we derive the claim.
    \end{claimproof}

   Again without loss of generality we may assume that $I = \N$ satisfies \cref{claim:bounded_tw_basecase_2}. 
    By \cref{def:knitwork_large_degree} we have $x_i \notin \dom(\mu_i)$. For every $i \in \N$ fix an ordering $\pi(x_i)$ of $\rho(x_i)$ and let $$\GGG_i^\star \coloneqq \rip(\GGG_i;\pi(x_i))=\big( \rip(\bar G_i;\pi(x_i)), \mu_i^\star, \m_i^\star, \Phi_i^\star\big)$$ be a ripping of $x_i$ with ripped ends $X_{x_i}$; let $X_{x_i}^-,X_{x_i}^+ \subseteq X_{x_i}$ be the maximum subsets of vertices with out-degree $0$ and in-degree $0$ respectively. (Recall that ripping does not introduce loops by \cref{obs:ripping_fundamentals}).

    Clearly $G_i^\star$ is Eulerian up to $X_{x_i}$---denote its defect by $\tau_i$---and it satisfies $\tw{G_i^\star} \leq \tw{G_i}$ for every $i \in \N$. (Note however that the sequence $(\tau_i)_{i \in \N}$ is strictly monotonically increasing, in particular $(\Abs{X_{x_i}})_{i\in \N}$ is, whence $(\tw{\GGG_i^\star})_{i \in \N}$ is unbounded by \cref{def:knitworks_treewidth}.) By construction we derive that $\Delta(G_i^\star) \leq \Delta$ for every $i \in \N$.

    Thus, since $\tw{G_i^\star} \leq k $ for every $i \in \N$, \cref{thm:qualitative_equivalence_of_tw_and_cw} implies that $\cw{G_i^\star} \leq \Delta \cdot k$. For every $i \in \N$ let $(T_i,\ell_i)$ be a carving of $G_i^\star$ witnessing its carving width. By \cref{lem:matching_degree_ones_low_cw} there exists a perfect matching $M_i \in \operatorname{Match}(X_{x_i}^-,X_{x_i}^+)$ such that $(T_i,\ell_i)$ is a carving of $G_i^\star + M_i$ of width $\leq 2(k+\Delta)$. In particular $G_i^\star + M_i$ is Eulerian of carving width at most $2(k +\Delta)$ for every $i \in \N$. Again by \cref{thm:qualitative_equivalence_of_tw_and_cw} we derive that 
    $$\tw{G_i^\star + M_i} \leq 3(k+\Delta).$$ 
    
 Let $\star^-,\star^+$ be distinct new elements not in $V(\Omega)$ and define $\Omega^\star \coloneqq(V(\Omega) \cup \{\star^-,\star^+\},\preceq^\star)$ to be the well-quasi-order obtained via extension by letting $\preceq^\star$ agree with $\preceq$ on $V(\Omega)\times V(\Omega)$ and making $\star^-,\star^+$ incomparable to every other element. 
    For every $i \in \N$ we define $\Psi_i^\star:V(G_i^\star) \to V(\Omega^\star)$ via $\Psi_i^\star(u) \coloneqq \star^-$ for every $u \in X_{x_i}^-$ and $\Psi_i^\star(u) \coloneqq \star^+$ for every $u \in X_{x_i}^+$ and fixing $\Psi_i^\star(x) \coloneqq \Phi_i^\star(x)$ for every $x \in V(G_i)\setminus\{x_i\}$. 
    Finally we define $\HHH_i \coloneqq \big( (G_i^\star + M_i, \pi,x_1^i,\ldots,x_\ell^i),\mu_i^\star,\m_i^\star,\Psi_i^\star)$; this is an $\Omega^\star$-knitwork by definition.
    
    \begin{claim}\label{claim:bounded_tw_basecase_3}
        $\HHH_i \in \mathbf{T}(3(k+\Delta),0,\Delta,\ell;\Omega^\star)$.
    \end{claim}
    \begin{claimproof}
        By construction $\tw{H_i^\star + M_i} \leq 3(k+\Delta)$ and $\Delta(H_i) \leq \Delta$. Thus we are left to show that it is a well-linked $\Omega^\star$-knitwork of index $\ell$. The index is clear from the definition and the fact that it is a well-linked $\Omega^*$-knitwork follows immediately from the fact that $X_{x_i} \cap \dom_{H_i}(\mu_i^\star) = \emptyset$ and thus the edges $M_i$ have no influence on the well-linkedness nor on the definitions of $\mu_i^\star,\m_i^\star$ and $\Psi_i^\star$. The claim follows.
    \end{claimproof}

    By \cref{cor:wqo_bounded_cw_Delta_version} and \cref{obs:wqo_yields_infinite_chain} we derive that there exists $I \subseteq \N$ such that $\HHH_i \hookrightarrow \HHH_j$ by strong $\Omega^\star$-knitwork immersion for every $i<j$ with $i,j \in I$.

    \begin{claim}
        Let $i<j$ with $i,j \in I$. If $\HHH_i \hookrightarrow \HHH_j$ by $\Omega^\star$-knitwork immersion, then $\GGG_i \hookrightarrow^* \GGG_j$ by $\Omega$-knitwork immersion.
    \end{claim}
    \begin{claimproof}
        Let $\gamma: \HHH_i \hookrightarrow \HHH_j$ be a strong $\Omega^\star$-knitwork immersion. By construction $\Psi_i(X_{x_i}^-) = \{\star^-\}$ as well as $\Psi_i(X_{x_i}^+) = \{\star^+\}$ for every $i \in I$. Since $\star^-$ is only comparable to $\star^-$ and $\star^+$ is only comparable to $\star^+$ in $\Omega^\star$, we derive the following.
        \begin{itemize}
            \item[$(i)$] for $u \in V(H_i)$ it holds that $\gamma(u) \in X_{x_j}^- \iff u \in X_{x_i}^-$ as well as $\gamma(u) \in X_{x_j}^+ \iff u \in X_{x_i}^+$.
        \end{itemize}
        Since further for every $i \in I$ every vertex in $X_{x_i}$ is of degree two in $H_i$, we derive the following from Property $(i)$ and the \cref{def:knitwork_immersion} of knitwork immersion.
        \begin{itemize}
            \item[$(ii)$] If $e \in E(H_i)\setminus M_i$ and $u \coloneqq \tail_{H_i}(e) \in X_{x_i}$, then $u \in X_{x_i}^+$ and $\gamma(e)$ is a path starting with $e' \in E(H_j) \setminus M_j$ with $\tail_{H_j}(e') \in X_{x_i}^+$.
        \end{itemize}

        We construct $\gamma^\star\colon V(G_i) \cup E(G_i) \to G_j$ from $\gamma$ as follows. Recall that $V(H_i) = V(G_i) \setminus \{x_i\} \cup X_{x_i}$ for every $i \in I$ by construction. We set
            $$\restr{\gamma^\star}{V(G_i)}(v) \coloneqq \begin{cases} \gamma(v),&\text{if } v \neq x_i\\
        x_j,&\text{if } v=x_i,
            \end{cases}$$
        and hence \crefdef{def:immersion}{1} is satisfied.
 
        Similar to above, recall that $E(H_i) = E(G_i) \cup M_i$ for every $i \in I$ by construction. For $e=(u,v) \in E(H_i)$ let $P_e \coloneqq \gamma(e)$ be the respective path in $H_j$. There are two cases to consider.
        \begin{itemize}
            \item[$e \in M_i$] If $e \in M_i$ then $e \notin E(G_i)$ and we do not define $\gamma^\star$.
            \item[$e \notin M_i$] If $e \notin M_i$, then $e \in E(G_i)$. We define $\gamma^\star(e)$ to be the path obtained from $\gamma(e)$ as follows. If $\gamma(e) \cap M_j = \emptyset$ we simply set $\gamma^\star(e) \coloneqq \gamma(e)$; by construction $\gamma^\star(e)=\gamma(e)$ is a path in $G_j$. If both ends $u,v \notin X_{x_i}$, then by definition $\gamma^\star(u)=\gamma(u)$ as well as $\gamma^\star(v) = \gamma(v)$ and thus $\gamma^\star(e)$ starts in $\gamma^\star(u)$ and ends in $\gamma^\star(v)$; \crefdef{def:immersion}{3} is satisfied. If $u \in X_{x_i}$, say, then $\tail_{G_i}(e)=x_i$ by \cref{def:ripping} and again $\gamma^\star(e) \coloneqq \gamma(e)$ is a path in $G_j$ and it starts in $\gamma^\star(x_i)$ thus $3.$ is satisfied, and analogously if $v \in X_{x_i}$. (Note that not both $u$ and $v$ may be in $X_{x_i}$ since $G_i$ is loopless.)

            Thus assume that $\gamma(e) \cap M_j \neq \emptyset$; let $\gamma(e)=(e_1,\ldots,e_m)_{H_j}$. First note that $e_1,e_m \notin M_j$ by Property $(ii)$. In particular $e_1 \neq e_m$ since we assumed that $\gamma(e) \cap M_j \neq \emptyset$.
            
            We define $\gamma^\star(e)$ to be the subsequence of $\gamma(e)$ obtained by deleting every element of $M_j$ in the sequence. Note that $\gamma^\star(e)$ consists of at least the two edges $e_1,e_m \in E(G_j)$. We claim that $\gamma^\star(e)$ is a path in $G_j$, to see this let $(e_p,e_q) \subseteq \gamma^\star(e)$ be a subsequence of length $2$ (which exists since it contains at least $2$ edges). If $e_q = e_{p+1}$ there is nothing to show since then $\head_{H_j}(e_p) = \head_{G_j}(e_p)$. Otherwise $e_{p+1},\ldots,e_{p+{q-1}} \in M_j$ by definition of $\gamma^\star(e)$. But using Properties $(i)$ and $(ii)$ this implies that $\head_{H_j}(e_p) \in X_{x_i}^-$ as well as $\tail_{H_j}(e_q) \in X_{x_i}^+$. In particular, by construction $\head_{G_j}(e_p) = x_i$ and $\tail_{G_j}(e_q) = x_i$ and thus $(e_p,e_q)$ is a well-defined $2$-path in $G_j$ (here we lose strongness of the immersion). Since all the edges in $\gamma^\star(e)$ are distinct by definition, and since every subsequence of length $2$ is a well-defined path in $G_j$ we derive that $\gamma^\star(e)$ is a path in $G_j$. And similar to above one verifies that $\gamma^\star(e)$ satisfies \crefdef{def:immersion}{3}. 
        \end{itemize}
            Finally it is clear from construction that $\LLL \coloneqq \{\gamma^\star(e) \mid e\in E(G_j)\}$ is a linkage for the paths remain edge-disjoint by construction. Since \ref{def:immersion:1}, \ref{def:immersion:2}, and \ref{def:immersion:3} of \cref{def:immersion} are satisfied we derive that $\gamma^\star:G_i \hookrightarrow^* G_j$ is a weak immersion.

            \smallskip

            By definition of $\gamma^\star$ one easily verifies that it is a \emph{rooted} immersion~$\gamma^\star \colon \bar G_i \hookrightarrow^* \bar G_j$ as follows. Since $\gamma\colon \bar H_i \to \bar H_j$ is a rooted immersion we have by definition that $\gamma^\star(v_t^i) = \gamma(v_t^i) = v_t^j$, whence $\gamma^\star$ satisfies \crefdef{def:rooted_immersion}{2}. And using Property $(ii)$ above $\gamma^\star$ satisfies \crefdef{def:rooted_immersion}{3} since $\gamma^\star(e)$ and $\gamma(e)$ always agree in the first and last edge for every $e \in E(G_i)$; in particular for all edges of $\bigcup_{t=1}^\ell \rho_{G_i}(v_t^i)$.

            \smallskip

            Similarly, one easily verifies that $\gamma^\star$ is an $\Omega$-knitwork immersion since $x_i \notin \dom(\mu_i)$ and $x_j \notin \dom(\mu_j)$ and the immersion ``agrees'' with $\gamma$ away from $x_i$ and $x_j$. Note that by construction $\gamma^\star(x_i) = x_j$ and by our previous filtering we have $\Phi_i(x_i) \preceq \Phi_j(x_j)$. This concludes the proof. 
    \end{claimproof}

The claim immediately implies the lemma.
\end{proof}

We are confident that this technique can be lifted to prove the claim for $2$ and $3$ vertices of large degree using some more technicalities, but it ``fails'' for fixed $d \geq 4$, where new machinery needs to be deployed.

\bibliographystyle{alphaurl}
\newcommand{\etalchar}[1]{$^{#1}$}

\end{document}